\documentclass{article}
\usepackage[letterpaper, portrait, margin=1in]{geometry}

\usepackage[title]{appendix}%

\usepackage{amsmath,
	amssymb,amsfonts}

\usepackage{amsthm}

\let\labelindent\relax
\usepackage{enumitem}

\DeclareFontFamily{U}{mathx}{\hyphenchar\font45}
\DeclareFontShape{U}{mathx}{m}{n}{
	<5> <6> <7> <8> <9> <10> gen * mathx
	<10.95> mathx10 <12> <14.4> <17.28> <20.74> <24.88> mathx12
}{}
\DeclareSymbolFont{mathx}{U}{mathx}{m}{n}
\DeclareMathAccent{\widecheck}{0}{mathx}{"71}

\usepackage[T1]{fontenc}
\usepackage[utf8]{inputenc}
\usepackage[english]{babel}

\usepackage{orcidlink}

\usepackage{pgfplots}
\usepackage[normalem]{ulem}
\usepackage{graphicx}

\usepackage{subfig}

\usepackage{multirow}%
\usepackage{mathrsfs}%

\usepackage{xcolor}%
\usepackage{textcomp}%
\usepackage{manyfoot}%
\usepackage{booktabs}%
\usepackage{algorithm}%
\usepackage{algorithmicx}%
\usepackage{algpseudocode}%
\usepackage{listings}%

\DeclareMathAlphabet{\mymathbf}{U}{bbold}{m}{n}

\definecolor{OliveGreen}{rgb}{0, 0.5, 0}

\usepackage{accents}

\usepackage{hyperref}
\hypersetup{
	colorlinks=true,
	linkcolor=blue,
	citecolor=OliveGreen,
	filecolor=magenta,      
	urlcolor=cyan
}

\newcommand{\SSnec}{\mathcal{S}^{\texttt{nec}}}
\newcommand{\SSrem}{\mathcal{S}^{\texttt{rem}}}

\newtheorem{theorem}{Theorem}[section]
\newtheorem{definition}[theorem]{Definition}

\newtheorem{lemma}[theorem]{Lemma}

\newtheorem{example}[theorem]{Example}
\newtheorem{corollary}[theorem]{Corollary}
\newtheorem{proposition}[theorem]{Proposition}

\newcommand{\longthmtitle}[1]{\mbox{}\textup{\textsl{(#1):}}}

\newcommand{\CC}{\mathcal{C}}
\newcommand{\DD}{\mathcal{D}}
\newcommand{\EE}{\mathcal{E}}
\newcommand{\FF}{\mathcal{F}}
\newcommand{\GG}{\mathcal{G}}
\newcommand{\HH}{\mathcal{H}}
\newcommand{\II}{\mathcal{I}}
\newcommand{\JJ}{\mathcal{J}}

\newcommand{\MM}{\mathcal{M}}
\newcommand{\NN}{\mathcal{N}}
\newcommand{\PP}{\mathcal{P}}
\newcommand{\QQ}{\mathcal{Q}}
\newcommand{\RR}{\mathcal{R}}
\newcommand{\SSS}{\mathcal{S}}

\newcommand{\VV}{\mathcal{V}}
\newcommand{\WW}{\mathcal{W}}

\newcommand{\tillf}{\widetilde{f}}

\newcommand{\tillC}{\widetilde{C}}
\newcommand{\tillD}{\widetilde{D}}
\newcommand{\tillM}{\widetilde{M}}
\newcommand{\tillV}{\widetilde{V}}

\newcommand{\tilllamb}{\widetilde{\lambda}}
\newcommand{\tillGamma}{\widetilde{\Gamma}}

\newcommand{\tillDD}{\widetilde{\DD}}
\newcommand{\tillFF}{\widetilde{\FF}}

\newcommand{\tillHH}{\widetilde{\HH}}
\newcommand{\tillMM}{\widetilde{\MM}}

\newcommand{\tillPP}{\widetilde{\PP}}
\newcommand{\tillQQ}{\widetilde{\QQ}}
\newcommand{\tillRR}{\widetilde{\RR}}

\newcommand{\tillWW}{\widetilde{\WW}}
\newcommand{\hatf}{\widehat{f}}

\newcommand{\hatPP}{\widehat{\PP}}
\newcommand{\hatSSS}{\widehat{\SSS}}
\newcommand{\hatGamma}{\widehat{\Gamma}}

\newcommand{\real}{{\mathbb{R}}}

\newcommand{\realnonnegative}{{\mathbb{R}}_{\ge 0}}
\newcommand{\naturalnumbers}{\mathbb{N}}
\newcommand{\setdef}[2]{\{#1 \; | \; #2\}}
\newcommand{\setdefb}[2]{\Bigl\{#1 \; | \; #2\Bigr\}}

\newcommand{\map}[3]{#1:#2 \rightarrow #3}
\newcommand{\abs}[1]{\ensuremath{\left\lvert{#1}\right\rvert}}
\newcommand{\norm}[1]{\ensuremath{\| #1 \|}}

\newcommand{\classfn}{\mathcal{K}}

\newcommand{\aset}{\RR^{\operatorname{act}}}
\newcommand{\uset}{\RR^{\operatorname{use}}}
\newcommand{\tilaset}{\widetilde{\RR}^{\operatorname{act}}}
\newcommand{\tiluset}{\widetilde{\RR}^{\operatorname{use}}}
\newcommand{\hatuset}{\widehat{\RR}^{\operatorname{use}}}

\newcommand{\Ja}{\JJ^{\operatorname{act}}}
\newcommand{\Ju}{\JJ^{\operatorname{use}}}
\newcommand{\tillJa}{\widetilde{\JJ}^{\operatorname{act}}}

\newcommand{\vin}{v^{\operatorname{in}}}
\newcommand{\vout}{v^{\operatorname{out}}}
\newcommand{\VI}{\operatorname{VI}}
\newcommand{\SOL}{\operatorname{SOL}}
\newcommand{\lmvec}{\lambda^{\operatorname{vec}}}
\newcommand{\lmWE}{\lambda^{\operatorname{WE}}}
\newcommand{\hatlmWE}{\widehat{\lambda}^{\operatorname{WE}}}
\newcommand{\tillmWE}{\widetilde{\lambda}^{\operatorname{WE}}}
\newcommand{\checklmWE}{\widecheck{\lambda}^{\operatorname{WE}}}
\newcommand{\tillmvec}{\widetilde{\lambda}^{\operatorname{vec}}}
\newcommand{\delmplus}{\delta \lambda^+}
\newcommand{\delmminus}{\delta \lambda^-}
\newcommand{\tildelmplus}{\delta \widetilde{\lambda}^+}
\newcommand{\tildelmminus}{\delta \widetilde{\lambda}^-}
\newcommand{\oprocendsymbol}{\hbox{$\bullet$}}
\newcommand{\oprocend}{\relax\ifmmode\else\unskip\hfill\fi\oprocendsymbol}

\newcommand{\coco}{\operatorname{coco}}

\newcommand{\vecPP}{\vec{\PP}}
\newcommand{\vecDD}{\vec{\DD}}
\newcommand{\vecD}{\vec{D}}
\newcommand{\vecM}{\vec{\MM}}
\newcommand{\veclamb}{\vec{\lambda}}
\newcommand{\vecmWE}{\vec{\lambda}^{\operatorname{WE}}}

\usepackage[prependcaption,colorinlistoftodos]{todonotes}

\allowdisplaybreaks

\newcommand{\ifinclude}[1]{}
\renewcommand{\ifinclude}[1]{#1}

\title{Wardrop equilibrium and Braess's paradox for varying demand}
\author{Jasper Verbree, 
	Ashish Cherukuri
	\thanks{J. Verbree and A. Cherukuri 
		are with the ENgineering and TEchnology institute Groningen (ENTEG), University of Groningen, 
		\texttt{\{j.verbree, a.k.cherukuri\}@rug.nl}.}
}

\begin{document}
	\maketitle
	\begin{abstract}
		This work explores the relationship between the set of Wardrop equilibria~(WE) of a routing game, the total demand of that game, and the occurrence of Braess's paradox~(BP). The BP formalizes the counter-intuitive fact that for some networks, removing a path from the network decreases congestion at WE. For a single origin-destination routing games with affine cost functions, the first part of this work provides tools for analyzing the evolution of the WE as the demand varies. It characterizes the piece-wise affine nature of this dependence by showing that the set of directions in which the WE can vary in each piece is the solution of a variational inequality problem. In the process we establish various properties of changes in the set of used and minimal-cost paths as demand varies. As a consequence of these characterizations, we derive a procedure to obtain the WE for all demands above a certain threshold. The second part of the paper deals with detecting the presence of BP in a network. We supply a number of sufficient conditions that reveal the presence of BP and that are computationally tractable. We also discuss a different perspective on BP, where we establish that a path causing BP at a particular demand must be strictly beneficial to the network at a lower demand. Several examples throughout this work illustrate and elaborate our findings.
	\end{abstract}

	\section{Introduction}
	With the ever increasing amount of traffic that finds its way to the road, the challenges of designing efficient traffic networks becomes more and more prevalent. Unfortunately, when drivers make routing choices in a selfish manner, that is, when traffic routes according to Wardrop equilibrium (WE)~\cite{JRC-NESM:11}, road networks are susceptible to innefficiencies which are hard to predict in advance, and hard to detect when they occur. For instance, in some cases travel cost of all drivers could be decreased by removing some roads in the network; a phenomenon known as Braess's paradox (BP). Constructing a road that subsequently increases travel cost of all drivers using a network is of course the epitome of innefficiency. However, detecting whether or not a traffic network is or will be susceptible to BP, even when using a simple model in which travel costs on roads are affine functions of the traffic on that road, is in general NP-hard \cite{TR:06}. 
	
	In addition to being computationally expensive to detect, BP is also known to be highly dependent on the amount of traffic, or demand, that traverses a network. Though the removal of a road can decrease travel cost for one level of demand, it can increase travel cost at other levels of demand. Removal of a road responsible for BP therefore requires careful consideration of the effect of that removal for the whole range of demands which the network can experience.
	
	Motivated by the above listed observations, we study the relation between these two challenges. To be specific, we study non-atomic, single origin-destination routing games with affine cost functions on the edges of the network. For these type of games we analyze the effects that changes in the demand have on the WE of these games, and then use the obtained observations to gain insight in how the effects of BP change as the demand varies. On the basis of these insights we can give various easy-to-check sufficient conditions revealing the presence of BP.

	\subsection*{Literature review}
	Wardrop equilibria(WE), first introduced in \cite{JGW:52}, have long been an important part of traffic modeling and control. The underlying idea of this concept is that when the flow over a traffic network is a WE, no driver can achieve a lower travel cost by switching its routing choices. In this way it captures the behavior of drivers minimizing their individual travel cost in an intuitive way. The concept has practical use since the set of WE is equivalently given by the set of solutions of a Variational Inequality (VI), which can be computed efficiently. We refer to \cite{JRC-NESM:11} for a survey on the subject of WE in general.
	
	Research on routing games with changing demand, and the effects thereof on WE and the associated costs, often analyzes a more general framework then we consider here. In \cite{AN-DP-PD:07}, WE are studied in a dynamic setting, where demand is time-dependent. Work in \cite{RC-VD-MS:21} is focused on studying the relation between the total cost under WE and the minimally attainable total cost, the difference between which is known as the price of anarchy. A  thorough investigation of the continuity of WE of routing games and monotonicity of the associated costs with respect to changes in the demand is done in~\cite{MAH:78}. When multiple Origin-Destination (OD) pairs are considered, \cite{CF:79} shows how costs under WE can decrease as demand increases. In \cite{MP:04,MJ-MP:07} the sensitivity of WE to changes in a broader class of parameters than only the demand is investigated. 
	Most closely related to our work in this context is \cite{MP:04}. The conclusion obtained in the first part of this work are similar to results obtained in \cite{MP:04}, the difference being that we consider a more limited set of cost functions (namely affine cost functions), only consider a single OD-pair, and crucially, do not constrain the WE to be unique. In addition, several intermediate results obtained in our analysis  turn out to be useful in the second part of this paper, concerning the Braess's paradox.
	
	Similar to the subject of Wardrop equilibria, the related concept of Braess's paradox (BP), has been the subject of a lot of research since its introduction in~\cite{DB:68} (see~\cite{DB-AN-TW:05} for an English translation). Various approaches to detecting the presence of BP have been studied, of which we mention a few examples here. In~\cite{JNH-RAA:01}, a necessary and sufficient condition for the occurrence of BP are given based on properties of the solutions of a minimization problem. Similarly, in~\cite{CM-QC-SA-BS-VND:19}, a comparison between flows that minimize specific objectives, such as WE, and the flow that minimizes the average cost are used to identify links that potentially cause BP. In~\cite{SAB-AC-MT-CB:14}, a heuristic method is employed first identifying single links which cause BP, and subsequently finding the subset of all these links that when removed lead the lowest cost under WE. Rather than focusing on a network with given demand and cost functions, the work~\cite{IM:06} investigates the relation between network topology and BP, showing that any network that is not series-parallel is subject to BP for some assignment of cost functions. However, despite these investigations, tractable options for detecting BP remain limited, and the general problem of detecting BP has in fact been shown to be NP-hard \cite{TR:06}. In light of this computational intractability, the likelihood of a network being subject to Braess's paradox has been studied. In~\cite{GV-TR:06}, it is shown that for randomly generated Erd\"{o}s-Renyi graphs, the likelihood of there existing a set of edges which cause BP at some demand tends to one as the number of vertices tends to infinity. In~\cite{RS-WIZ:83}, necessary and sufficient conditions for the occurrence of BP are given, and it is shown, under the limiting assumption that paths that were used under WE before removal of an edge are still used under WE after removal as well, that BP is as likely to occur as not. Closer to the subject of our work, the relation between demand and BP has also been studied in several papers. For the archetypical example of the Wheatstone network with affine costs, BP has been shown to only appear within a finite range of demands \cite{EIP-SLP:97,VZ-EA:15}. It has also been shown that under somewhat restrictive assumptions an edge responsible for BP will for higher levels of demand remain unused \cite{AN:10}.

	\subsection*{Statement of contributions}
	The aim of this paper is to study the relationship between changing demand, Wardrop equilibrium (WE) and Braess's paradox (BP), in the context of non-atomic, single-origin destination routing games with affine, non-decreasing cost functions on the edges. The paper can be broadly divided into two parts, the first characterizes the change in WE as demand shifts and the second considers Braess's paradox, its identification and the effect of demand on it.
	
	When the WE is unique for all demand values, the map from the demand to the WE and the cost incurred by each path at WE (loosely termed as the WE cost-vector) is piece-wise affine. These pieces are defined by a set of points that we refer to as breakpoints. That is, the WE and the WE cost-vector vary linearly with demand between any two consecutive breakpoints. These breakpoints are finite in number. When the WE is not unique, the situation is similar, the only change being that the map from the demand to WE is set-valued. In the first part of our work, our contributions consist of the following results concerning the evolution of the set of WE:
	\begin{enumerate}
		\item In (Corollary~\ref{cor:interval-of-active-set} and Lemma~\ref{lem:inclusion-active-and-used}), we analyze how the used sets and active sets evolve between and at breakpoints. The used set is the set of paths that carry non-zero flow for at least on WE. The active set is the set of paths that incur minimum cost at WE among all paths. 
		\item The main result of the first part (Proposition~\ref{prop:characterize-directions-of-increase}) is the characterization of the directions in which the WE can vary as the demand increases. We call these directions the directions of increase. We show that these can be computed by solving a VI problem. Essentially this is a generalization of results in \cite{MP:04}, where a similar property was derived under the assumption that WE are unique.
		\item In the process of establishing the main result, we also characterize the map from demand to the WE-cost (cost incurred by the entire flow at WE) and from demand to the WE cost-vector (cost of each path at a WE) between consecutive breakpoints (Proposition~\ref{prop:evolution-lmvec} and Corollary~\ref{cor:evolution-lmWE}). 
		\item The final set of results in this section studies the behavior of WE in the ``final interval'' (Section~\ref{sec:obtain-DM}); that is, the interval that starts at the largest finite breakpoint and extends till positive infinity. Specifically, we give a method to compute the largest finite breakpoint (Corollary~\ref{cor:obtainDM}), the direction in which WE cost and WE cost-vector evolve in the final interval (Proposition~\ref{prop:characterize-D-infinity} and Lemma~\ref{lem:obtaining-lambda-D_M}), and the active set in the final interval (Lemma~\ref{eq:finding-JaM}). 
	\end{enumerate}
	Using the results of the first part, we study Braess's paradox and give the following main insights that form the second part of our contributions:
	\begin{enumerate}
		\item We introduce the concept of a \emph{necessary} set of paths, which for a given demand is a set of paths on which the total amount of flow is positive under any WE. We show in Proposition~\ref{prop:unnecesary-and-BP} that when a set of paths is not necessary at a demand, then this indicates that it is either unnecessary for all lower levels of demand, or this set is responsible for BP on some lower level of demand.
		\item  Next we provide a set of affine functions which give a set of sufficient conditions for the presence of BP (Lemma~\ref{lem:lambda_vs_u}). Specifically, if at any demand, the WE-cost of the game is more than the value that any of these affine functions take at this demand, then the network is subject to BP. 
		In fact we show that the network is subject to BP if and only if its WE-cost is less than one of these affine functions (Proposition~\ref{prop:BP-ifandonlyif}). However, obtaining all of these affine functions is computationally intractable. Therefore we provide two more sufficient conditions, based on these affine fuctions, that detect BP at a particular demand in a tractable manner. The first one (Corollary~\ref{prop:delta-lambda-increase-BP}) shows that  if the slope of WE-cost increases at any breakpoint, then the network is subject to BP. Additionally, if at a given demand, the set of directions of increase does not overlap with a simplex, then the routing game is subject to BP at that demand (Proposition~\ref{prop:losing-flow-implies-BP}) 
		\item The next set of results views BP in a broader perspective and explores its dependence on the demand. We show that any set of paths responsible for BP at some demand must have been strictly beneficial for some lower level of demand (Corollary~\ref{cor:BP-beneficial}). Here, beneficial means that removal of that set of paths will strictly increase the WE-cost. Motivated by this observation, we introduce two simple metrics that measure the value of a set of paths over a range of demands. For the first metric, we show that over a range of demand, removing a set of paths does not have a negative impact cost-wise (Proposition~\ref{pr:J}). When considering the second metric, we show that the over a range of demand, removing a set of paths does not have a positive impact cost-wise (Proposition~\ref{pr:W}). Thus, these two results give a guideline for the planner to remove paths based on the metric that needs to be improved on.
	\end{enumerate}
	Several results listed above are supported with illustrative examples that clarify nuances and aid understanding. 
	
	\subsection*{Notation}
	The set of natural numbers is denoted by $\naturalnumbers$. For any $n \in \mathbb{N}$ we write ${[n] := \{1,2,\cdots,n\}}$ and ${[n]_0 := \{0,1,\cdots,n\}}$. We use $\mymathbf{0}_n$ to denote the $n$-dimensional vector of zeros. If the dimension is clear from the context, we simply write $\mymathbf{0}$. The hyperplane in $\real^n$ containing set of points whose elements add up to $D \in \real$ is written as ${\HH_D := \setdef{f \in \real^n}{\sum_{i \in [n]} f_i = D}}$. For any $f^-,f^+ \in \real^n$ and $\mu \in [0,1]$ we use $\coco_{\mu}(f^-,f^+) : = \mu f^- + (1 - \mu)f^+$ to denote the convex combination of $f^-$ and $f^+$.
	For a set $\PP \subseteq [n]$ we use ${\PP^c = \setdef{p \in [n]}{p \notin \PP}}$ to denote the complement of that set, and write $\abs{\PP}$ for the number of elements in $\PP$. For a matrix $A \in \real^{n \times n}$, we let $A_i \in \real^{1 \times n}$ denote the $i$-th row of that matrix. Given a map ${g: \real \rightarrow \real^n}$, we use $\frac{\partial^+}{\partial x} g(x) : = \lim_{h \rightarrow 0^+}\frac{g(x + h) - g(h)}{h}$ and ${\frac{\partial^-}{\partial x} g(x) : = \lim_{h \rightarrow 0^-}\frac{g(x + h) - g(h)}{h}}$ to denote the right-hand and left-hand derivative of $g$ at $x$, respectively, whenever these limits exists.
	
	\section{Model and problem statement} \label{sec:model}
	Consider a network defined by a directed graph $\GG = (\VV,\EE)$, where $\VV = [N]$, $N \in \mathbb{N}$, is the set of vertices and $\EE \subseteq \VV \times \VV$ is the set of edges. Each edge $e_k \in \EE$ consists of an ordered pair of vertices $(\vin_k,\vout_k)$, where $\vin_k, \vout_k \in \VV$. A path $p$ from $v_o$ to $v_{d}$ is then an ordered set of edges $(e_{p_1},\cdots,e_{p_l})$ such that $\vin_{p_1} = v_{o}$, $\vout_{p_l} = v_d$ and $\vout_{pk} = \vin_{pk+1}$ for all $k \in [l-1]$. In addition, we define paths as being acyclic, meaning that $\vout_{pj} \neq \vin_{p_i}$ for all $1 \leq i \leq j \leq l$. 
	To this network, we associate an origin $v_o \in \VV$ and a destination $v_d \in \VV$. We denote the set of paths in the graph starting at $v_o$ and ending at $v_d$ by $\PP$. Consider an amount of traffic, denoted $D \in \realnonnegative$, and referred to as the \emph{demand}, that needs to be routed from origin to destination using paths in $\PP$. The way the traffic is divided over the paths gives rise to a flow $f \in \realnonnegative^n$, where $n = |\PP|$ is the number of paths from origin to destination, and an element $f_p$ of $f$ denotes the amount of flow that is routed over path $p$. The flow is assumed to be non-atomic, meaning that $f$ is a continuous variable. The set of \emph{feasible} flows is then given by
	\begin{equation}\label{eq:feasible}
		\FF_D := \setdefb{f \in \realnonnegative^n}{\sum_{p \in \PP}f_p = D}.
	\end{equation}
	For a flow $f$, the amount of traffic on edge $e_k \in \EE$ is given~by
	\begin{equation}\label{eq:flow-e-p}
		f_{e_k} := \sum_{\setdef{p \in \PP}{e_k \in p}}f_p.
	\end{equation} 
	We associate a real-valued cost function $f_{e_k} \mapsto C_{e_k}(f_{e_k})$ to each edge $e_k$, which models the time it takes to traverse an edge $e_k$ given the flow on that edge. We assume that these functions are in the set $\classfn$ of continuous, non-negative, non-decreasing, affine functions. We write
	\begin{equation}\label{eq:edge-cost}
		C_{e_k}(f_{e_k}) := \alpha_{e_k} f_{e_k} + \beta_{e_k},
	\end{equation}
	where $\alpha_{e_k},\beta_{e_k} \in \realnonnegative$.
	The cost of traversing a path $p$ is given by the summation of the costs of all the edges in the path, that is,
	\begin{equation*}
		C_p(f) := \sum_{e_k \in p} C_{e_k}(f_{e_k}).
	\end{equation*}
	For notational convenience, we define the following:
	\begin{align*}
		C_e(f) 	&:= \big(C_{e_1}(f_{e_1}),\cdots,C_{e_q}(f_{e_q})\big)^\top,	\\
		C(f)	:= \big(C_1(f),&\cdots,C_n(f)\big)^\top,	\text{ and } \,
		\CC		:= \{C_{e_k}\}_{e_k \in \EE},
	\end{align*}
	where $q := \abs{\EE}$. Here, $C_e$ collects the costs of all edges in a vector and similarly, $C$ collects the costs of all paths. Combining~\eqref{eq:flow-e-p} and~\eqref{eq:edge-cost}, we can express the \emph{path-cost function} as:
	\begin{equation}\label{eq:path-cost}
		C(f) = A f + \beta,
	\end{equation}
	where $\beta := (\beta_p)_{p \in \PP}$ is the vector with entries $\beta_p = \sum_{e_k \in p}\beta_{e_k}$ and $A \in \realnonnegative^{n \times n}$ is a matrix with the $(p,r)$-th entry given by ${A_{pr} = \sum_{e_k \in (p \cap r)} \alpha_{e_k}}$. Therefore, $A$ is symmetric and in fact positive semi-definite. This can be seen by noting that ${A = B^\top Q B}$, where $B \in \real^{ q \times n}$ is defined by $B_{k,i} = 1$ if $e_k \in p_i$ and $B_{k,i} = 0$ otherwise, and $Q \in \real^{q \times q}$ is a diagonal matrix with $Q_{k,k} = \alpha_{e_k} \geq 0$.
	
	Taken together, the set of paths $\PP$, the set of cost functions $\CC$ over the edges in the paths of $\PP$, and the demand $D \geq 0$ define a \emph{routing game}, which we denote with the tuple $(\PP,\CC,D)$. We will also refer to $(\PP,\CC)$ as a routing game when we want to consider it for  multiple levels of demand.
	We assume that participants of a routing game want to minimize the cost that they incur as they traverse the network, and choose a path accordingly. To include this in our model, we assume that the resulting flow over paths is a \textit{Wardrop equilibrium} (WE).
	\begin{definition} \label{def:WE} \longthmtitle{Wardrop Equilibrium}
		Given a routing game $(\PP,\CC,D)$, a flow ${f^D \in \realnonnegative^n}$ is called a Wardrop equilibrium if $f^D \in \FF_D$ and for every $p \in \PP$ such that $f^D_p > 0$ we have
		\begin{equation} \label{eq:WE-condition}
			C_p(f^D) \leq C_r(f^D) \quad \text{for all } r \in \PP.
		\end{equation}
		We will denote the set of all Wardrop equilibria as $\WW_D$. \oprocend
	\end{definition}
	The intuition behind this definition is that when the flow is in WE, none of the participants can decrease their travel-time by choosing another path. Note that this implies that all paths receiving a positive amount of flow incur the same cost. 
	
	A flow $f \in \FF_{D}$ satisfying \eqref{eq:WE-condition} exists~\cite[Section 3.1.2]{MJB-CBM-CBW:55}, and thus we know that a WE exists. In general, there can be multiple solutions to \eqref{eq:WE-condition}, showing that WE are 
	not necessarily unique. However, when $\alpha_{e_k} > 0$ for all $e_k \in \EE$, then for each edge, the flow at different WE is same; that is, if $f^D, \hatf^D \in \WW_D$, then $f^D_{e_k} = \hatf^D_{e_k}$~\cite{JRC-NESM:11}.
	When $\alpha_{e_k} = 0$ for some $e_k \in \EE$, then the flow on edges are not necessarily the same for all WE, but the edge costs are; that is , for $f^D, \hatf^D \in \WW_D$, we have $C_{e_k}(f^D) = C_{e_k}(\hatf^D)$ for all $e_k \in \EE$. In fact, $\hatf^D \in \FF_{D}$ is a WE if and only if $C_{e_k}(f^D) = C_{e_k}(\hatf^D)$. This property is sometimes referred to as the \emph{essential uniqueness} of a WE \cite{TR-ET:02}. Note that since the functions in $\CC$ are affine, this gives a way of representing $\WW_{D}$ as a set of points satisfying a number of affine equality constraints. Hence, $\WW_{D}$ is convex. Furthermore, since for each edge the cost is equal for all WE in $\WW_{D}$, it follows that the same holds for the costs on all paths, that is, $C(f^D) = C(\hatf^D)$. As a consequence of this property,  we can define the set of paths that are \emph{active} at demand $D$ as
	\begin{equation} \label{eq:relevant-road-constraint}
		\aset_D  :=  \setdef{p  \in  \PP }{  C_p(f^D)  \leq  C_r(f^D), \, \,  \forall f^D  \in  \WW_D,  \, \,  \forall r  \in  \PP}.
	\end{equation}
	This set collects all the paths that have minimum cost under any WE for the demand $D$. Similarly, we define the set of paths that are \emph{used} at demand $D$ as
	\begin{equation} \label{eq:used-road-constraint}
		\uset_D := \setdef{p \in \PP}{\exists f^D \in \WW_D \text{ such that } f^D_p > 0}.
	\end{equation}
	Note that from \eqref{eq:WE-condition}, we have $\uset_D \subseteq \aset_D$.
	The set $\aset_{D}$ and $\uset_{D}$ will play an essential role in the upcoming exposition where we analyze the evolution of WE as demand varies. As such, we define the following notation: for a set ${\RR = \{p_1,p_2,\cdots,p_m\} \subseteq \PP}$ and a flow $f \in \real^n$ we write
	\begin{align*}
		f_\RR &:= (f_{p_1},f_{p_2}, \cdots, f_{p_m})^\top,	\\
		C_\RR(f) &:= \big(C_{p_1}(f),C_{p_2}(f), \cdots, C_{p_m}(f)\big)^\top.
	\end{align*}
	That is, $f_\RR$ and $C_\RR$ collect the flow and costs for all paths in $\RR$. 
	We define a map 
	${\map{\lmWE}{\realnonnegative}{\realnonnegative}}$ that stands for the cost experienced by participants under WE as
	\begin{equation}\label{eq:lambda-defined}
		\lmWE(D) := \min_{p \in \PP} C_p(f^D) = C_r(f^D), \quad r \in \aset_D.
	\end{equation}
	Similarly, we define the map $\map{\lmvec}{\realnonnegative}{\realnonnegative^{n}}$, ${D \mapsto \lmvec(D) := C(f^D)}$, where $f^D$ is any flow that belongs to the set $\WW_D$. We will refer to $\lmWE$ and $\lmvec$ as the \emph{WE-cost} and the \emph{WE cost-vector},
	respectively. We will also denote the $p$-th component of $\lmvec(D)$ by $\lmvec_p(D)$.
	
	The exposition of the following sections can be broadly categorized in two themes. The first section (Section~\ref{sec:variation-of-WE}) gathers results related to behavior of the WE as the demand increases. The second one (Section~\ref{sec:BP-results}) introduces the notion of Braess's paradox and studies how to detect it and how it varies with the demand. Wherever necessary we illustrate the intuition and implications of our results using appropriate examples. 
	
	\section{Preliminaries} \label{sec:preliminaries}
	Here we introduce some preliminary material on how the set of WE can be obtained by solving either a variational inequality problem (VI) or alternatively an optimization problem. We then gather some useful results on the evolution of the set of WE and the active and used sets as the demand varies.
	\begin{definition} \label{def:VI}\longthmtitle{Variational inequalities (VIs)}
		Given a map $G: \real^n \rightarrow \real^n$ and a set $\mathcal{X} \subset \real^n$, the associated variational inequality problem, denoted $\VI(\mathcal{X},G)$ is to find $x^* \in \mathcal{X}$ such that the following holds:
		\begin{equation} \label{eq:VI-condition}
			G(x^*)^\top(x - x^*) \geq 0, \enskip \forall x \in \mathcal{X}.
		\end{equation}
		The set of solutions $x^* \in \mathcal{X}$ satisfying the above property is denoted as $\SOL(\mathcal{X}, G)$. \oprocend
	\end{definition}
	\begin{proposition}\longthmtitle{WE as the solution of a VI~\cite{MJS:79}} \label{prop:BP-VI}
		Let $(\PP,\CC,D)$ be given. A flow $f^D$ is a Wardrop equilibrium if and only if $f^D \in \SOL(\FF_D,C)$.
	\end{proposition}
	We note that when $C(f) = Af$ for some matrix $A \in \real^{n \times n}$, we use the notation $\SOL(\FF,A) := \SOL(\FF,C)$.
	
	Alternatively, WE can be characterized as the solution of a convex optimization problem.
	This characterization will give us the tools for analyzing the evolution of the WE cost $\lmWE$ as demand changes.
	\begin{proposition}\longthmtitle{WE as the solution of an optimization problem~\cite{JRC-NESM:11}} \label{prop:beckmann}
		Let $(\PP,\CC,D)$ be given. A flow $f^D$ is a Wardrop equilibrium if and only if it is an optimal solution of the following minimization problem
		\begin{equation} \label{eq:vD-defined}
			V(D) := \min_{f \in \FF_D} \sum_{e_k \in \EE} \int_0^{f_{e_k}}C_{e_k}(z)dz.
		\end{equation}
	\end{proposition} 
	The relationship between $V$ and $\lmWE$ is given in the following useful result from~\cite{RC-VD-MS:21}.
	\begin{proposition}\longthmtitle{Properties of WE cost $\lmWE$~\cite{RC-VD-MS:21}} \label{prop:pw-affine}
		Let $(\PP,\CC)$ be given. The functions $\lmWE$ and $\lmvec$ are continuous, piece-wise affine, and $\lmWE$ is non-decreasing on the interval $[0,\infty)$. Moreover, for any $D > 0$ we have
		\begin{equation*}
			\frac{\partial}{\partial D}V(D) = \lmWE(D).
		\end{equation*}
		Consequently, $V$ is convex and differentiable on $(0,\infty)$.
	\end{proposition}
	The next result summarizes properties of the active set $\aset_{D}$, and how it changes as the demand varies. Recall that for any two vectors $f^-,f^+ \in \real^p$ and any scalar ${\mu \in [0,1]}$, we use the notation ${\coco_\mu (f^-,f^+) = \mu f^- + (1 - \mu)f^+}$ to represent convex combinations.
	\begin{lemma}\longthmtitle{Evolution of the active set~\cite{RC-VD-MS:21}} \label{lem:facts-on-evolution-for-affine}
		Let $(\PP,\CC)$ be given.
		For any ${0 \le D^- \le D^+}$ that satisfy ${\aset_{D^-} = \aset_{D^+}}$, the following hold
		\begin{enumerate}
			\item For all $D \in [D^-,D^+]$, we have $\aset_{D} = \aset_{D^-} = \aset_{D^+}$.
			\item If $f^{D^-} \in \WW_{D^-}$, $f^{D^+} \in \WW_{D^+}$ and ${\mu \in [0,1]}$,
			then $\coco_{\mu}(f^{D^-},f^{D^+}) \in \WW_{T}$, where $T = \coco_{\mu}(D^-,D^+)$.
		\end{enumerate}
	\end{lemma}
	A consequence of the above result is that for any ${D \in \realnonnegative}$ there exists an interval $[D^-,D^+]$ containing $D$ such that $\aset_{T} = \aset_D$ for all $T \in (D^-,D^+)$, while $\aset_{T} \neq \aset_D$ for all $T \in [D^-,D^+]^c$. This observation, together with the fact that there are only a finite number of subsets of $\PP$ leads to the upcoming result. It states that there is a finite set of disjoint open intervals covering the non-negative real line, excepting a finite number of points, such that the active and used set remain constant in each of these intervals. Furthermore, no two such active sets, associated to different open intervals, are equal, and the same holds for the associated used sets.
	\begin{corollary} \longthmtitle{Piece-wise constant evolution of the active and used sets}\label{cor:interval-of-active-set}
		Let $(\PP,\CC)$ be given. There exists a finite set of points ${\DD := \{  D_0,D_1,\cdots,D_M,D_{M+1} \} \subset \realnonnegative \cup \{+\infty\}}$ with $D_0  =  0$, $D_{M+1} =  \infty$ and $D_{j} > D_{j - 1}$ for all ${j \in [M+1]}$, and corresponding sets of subsets of $\PP$ denoted  $\{\Ja_0,\Ja_1,\cdots,\Ja_M\}$ and $\{\Ju_0,\Ju_1,\cdots,\Ju_M\}$, such that, for all $i \in [M]_0$ and $D \in (D_i,D_{i+1})$, we have 
		\begin{equation*}
			\aset_D = \Ja_i, \quad \uset_D = \Ju_i .
		\end{equation*}
		Furthermore, $\Ja_i \neq \Ja_j$ and  $\Ju_i \neq \Ju_j$ for all $i \neq j$.
	\end{corollary}
	In the above result, the claim regarding the piece-wise constant evolution of active sets was established in~\cite[Section 4]{RC-VD-MS:21}. This property in combination with Lemma~\ref{lem:facts-on-evolution-for-affine}-2 and~\eqref{eq:path-cost} yields the claim regarding piece-wise evolution of the used sets. For the sake of completeness, a proof of this last part is given in Appendix~\ref{ap:proof-of-corollary-used-sets}.
	
	We call the demands in the set $\DD$ \emph{breakpoints}, as these are the points where the active set changes. Note that the above result makes no claims about the relationship between the active set at a breakpoint and the adjacent open intervals. We will shed light on this aspect in Lemma~\ref{lem:inclusion-active-and-used} in the subsequent section. Finally, we note that the points in the set $\DD$ are also the points where the WE cost-vector $\lmvec$ is non-differentiable, further motivating our choice for the term \emph{breakpoints}. We formally establish this in Proposition~\ref{prop:evolution-lmvec} in the next section.
	
	\section{Evolution of WE: Directions of increase and its properties} \label{sec:variation-of-WE} 
	In this section we analyze the evolution of the WE as the demand increases. 
	We demonstrate our results throughout this section using the two examples explained below. They will provide the necessary intuition and context for the technical exposition. 
	\begin{figure}
		\centering
		\begin{minipage}[t]{0.5 \textwidth}
			\centering
			\begin{tikzpicture}[,->,shorten >=1pt,auto,node distance=1.5 cm,
				semithick]
				
				\tikzset{VertexStyle/.style = {shape          = circle,
						ball color     = white,
						text           = black,
						inner sep      = 2pt,
						outer sep      = 0pt,
						minimum size   = 18 pt}}
				\tikzset{EdgeStyle/.style   = {
						double          = black,
				}}
				\tikzset{LabelStyle/.style =   {draw,
						fill           = white,
						text           = black}}
				\node[VertexStyle](A){$v_o$};
				\node[VertexStyle, above right=of A](B){};
				\node[VertexStyle, below right=of A](C){};
				\node[VertexStyle, below right=of B](D){$v_d$};
				
				\path (A)  edge      node{$e_1 $} (B)
				(A)  edge             node[below]{$e_3 \enskip$} (C)
				(B)	edge[dashed]			node{$e_5$}	(C)
				(B)	edge			node{$e_2$}	(D)
				(C)	edge			node[below]{$\enskip e_4$}	(D);
			\end{tikzpicture}
			\caption{The Wheatstone network.}
			\label{fig:wheatstone}
		\end{minipage}%
		\begin{minipage}[t]{0.5 \textwidth}
			\centering
			\includegraphics[width = \textwidth]{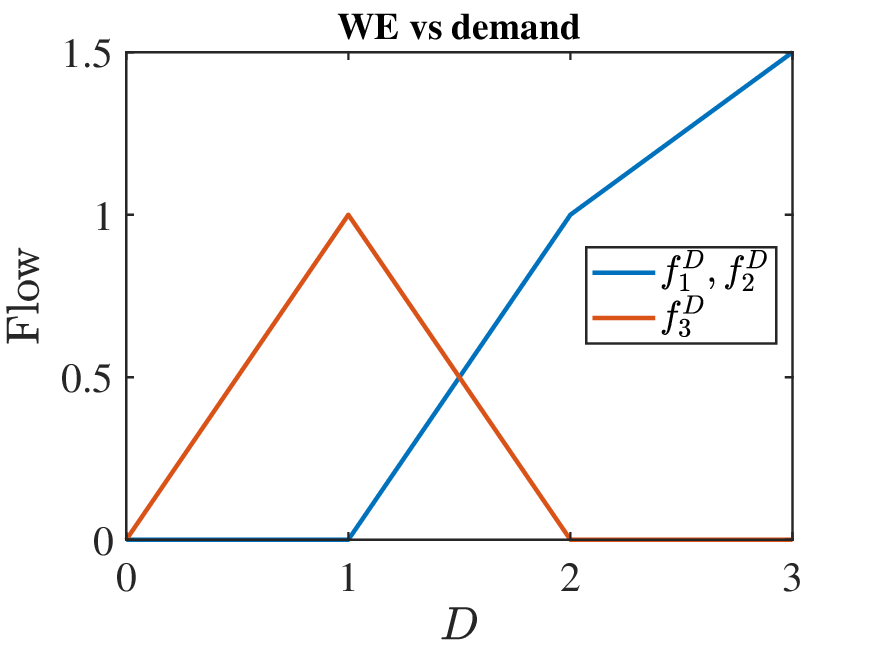}
			\caption{The WE at different demands for the routing game defined by the Wheatstone network (Figure~\ref{fig:wheatstone}) and costs~\eqref{eq:Wheatstone-simple-path-cost}.}
			\label{fig:wheatstone-WE}
		\end{minipage}
	\end{figure}
	\begin{example}\longthmtitle{Evolution of WE} \label{ex:secIV}
		{\rm We discuss here two simple networks:
			\begin{enumerate}[wide= 0pt,label=(Case-\alph*), ref=\ref{ex:secIV}\alph*]
				\item \label{ex:secIV-1} For the first example, consider the Wheatstone network depicted in Figure~\ref{fig:wheatstone}, with edge-cost functions given by
				\begin{equation} \label{eq:Wheatstone-simple-edge-costs}
					\begin{aligned}
						&C_{e_1}(f_{e_1}) = f_{e_1},	&	&C_{e_2}(f_{e_2}) = 1,	\\
						&C_{e_3}(f_{e_3}) = 1,	&	&C_{e_4}(f_{e_4}) = f_{e_4}, \hspace{8 pt}	\\
						&	&C_{e_5}(f_{e_5}) = 0.	&
					\end{aligned}
				\end{equation}
				There are three paths from the origin to the destination in this network, namely, ${p_1 = (e_1,e_2)}$, $p_2 = (e_3,e_4)$ and ${p_3 = (e_1,e_5,e_4)}$, with path-cost functions given respectively by
				\begin{equation} \label{eq:Wheatstone-simple-path-cost}
					\begin{split}
						C_1(f) &= f_{1} + f_3 + 1,	\\
						C_2(f) &= f_2 + f_3 + 1,	\\
						C_3(f) &= f_1 + f_2 + 2f_3,
					\end{split}
				\end{equation}
				where we recall that $f =(f_1,f_2,f_3)$ is the vector of path flows. Having defined the routing game $(\PP,\CC)$, the explicit expression for the WE as a function of the demand $D$, which is illustrated in Figure~\ref{fig:wheatstone-WE}, is given by the following:
				\begin{equation} \label{eq:Wheatstone-simple-WE}
					f^D  = 
					\begin{cases}
						\left(\begin{array}{ccc}
							0,	&0,	&D
						\end{array}\right)^\top & \text{for } D \in [0,1],	\\
						\left(\begin{array}{ccc}
							D-1,	&D-1,	&2-D
						\end{array}\right)^\top & \text{for } D \in [1,2],	\\
						\left(\begin{array}{ccc}
							\frac{D}{2},	&\frac{D}{2},	&0
						\end{array}\right)^\top & \text{for } D \in [2,\infty).
					\end{cases}
				\end{equation}
				It is important to note that $f^D$ changes continuously, and moreover, it evolves in a piece-wise affine manner. The points at which the evolution changes from one affine piece to the next will later turn out to be exactly the breakpoints in $\DD$ defined in Corollary~\ref{cor:interval-of-active-set}. Notice that in this example the WE is unique for any demand and  the evolution of the WE as demand increases is fully characterized by the right-hand derivative of the map $D \mapsto f^D$:
				\begin{equation*} 
					f^{\delta}(D) := \frac{\partial^+}{\partial D} f^D = 
					\begin{cases}
						\left(\begin{array}{ccc}
							0,	&0,	&1
						\end{array}\right)^\top & \text{for } D \in [0,1),	\\
						\left(\begin{array}{ccc}
							1,	&1,	&-1
						\end{array}\right)^\top & \text{for } D \in [1,2),	\\
						\left(\begin{array}{ccc}
							\frac{1}{2},	&\frac{1}{2},	&0
						\end{array}\right)^\top & \text{for } D \in [2,\infty).
					\end{cases}
				\end{equation*}
				\begin{figure}
					\centering
					\begin{tikzpicture}[,->,shorten >=1pt,auto,node distance=1.5 cm,
						semithick]
						
						\tikzset{VertexStyle/.style = {shape          = circle,
								ball color     = white,
								text           = black,
								inner sep      = 2pt,
								outer sep      = 0pt,
								minimum size   = 18 pt}}
						\tikzset{EdgeStyle/.style   = {
								double          = black,
						}}
						\tikzset{LabelStyle/.style =   {draw,
								fill           = white,
								text           = black}}
						\node[VertexStyle](A){$v_o$};
						\node[VertexStyle, right=of A](B){};
						\node[VertexStyle, right=of B](C){$v_d$};
						
						\path (A) edge[bend left= 30] node{$e_1 $} 	(B)
						(A) edge[bend right= 30] node[below]{$e_3$}	(B)
						(B) edge[bend left=  30] node{$e_2$}		(C)
						(B) edge[bend right= 30] node[below]{$e_4$}	(C);
					\end{tikzpicture}
					\caption{The Wheatstone network after merging the top and bottom nodes.}
					\label{fig:three-node-four-edge}
				\end{figure}
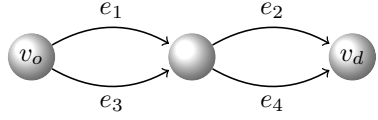
				Note that on the intervals where the map $D \mapsto f^D$ is affine, naturally the direction in which the WE moves is constant.
				\item \label{ex:secIV-2} When the WE are not unique, the situation can become more complicated. To illustrate this, we slightly modify our example. Instead of connecting the top and bottom node in the Wheatstone network (Figure~\ref{fig:wheatstone}) with the edge $e_5$, we merge them into one node.
				The resulting network is depicted in Figure~\ref{fig:three-node-four-edge}. The edge-cost functions in $\CC$ are still given by \eqref{eq:Wheatstone-simple-edge-costs}, except that edge $e_5$ no longer exists. There are now four paths in $\PP$, namely $p_1 = (e_1,e_2)$, $p_2 = (e_3,e_4)$, $p_3 = (e_1,e_4)$, and $p_4 = (e_3,e_2)$, and the path-cost function is given by
				\begin{equation} \label{eq:Wheatstone-double-cost}
					C(f) = Af + b = \left(\begin{array}{cccc}
						1	&0	&1	&0	\\
						0	&1	&1	&0	\\
						1	&1	&2	&0	\\
						0	&0	&0	&0	\\
					\end{array}\right)f + \left(\begin{array}{c}
						1	\\
						1	\\
						0	\\
						2	\\
					\end{array}\right).
				\end{equation}
				Note that $A f^0 = \mymathbf{0}$ for $f^0 = (-1, -1, 1, 1)^\top$ and so, for any flow $f$ and any value $\epsilon > 0$, we have $C(f + \epsilon f^0) = C(f)$. In other words, rerouting equal amounts of flow from $p_1$ to $p_2$ and from $p_3$ to $p_4$, or vice versa, does not change the path-cost. As a consequence, the WE of $(\PP,\CC)$ are not always unique, and we instead find a set of WE given by
				\begin{equation}  \label{eq:three-node-four-edge-WE}
					\begin{aligned}
						\WW_D  :=    
						\enskip \begin{cases}
							\left\{ 
							\left(
							\begin{array}{cccc}
								0	&0	&D	&0
							\end{array}
							\right) ^\top 
							\right\} &
							\text{for }  D  \in  [0,1],
							\\
							\setdef{ f  \in  \FF_D   }{ 	f_1  +  f_3  =  1, f_2  +  f_3  =  1 } & 
							\text{for }  D  \geq  1.	\\
						\end{cases}
					\end{aligned}
				\end{equation}
				Using this, we obtain the directions along which the WE can move as the demand increases as:
				\begin{equation} \label{eq:three-node-four-edge-increase}
					\begin{aligned}
						\Gamma_D  := 
						\begin{cases}
							\setdef{f^\delta  \in  \HH_{1} }{   f^\delta_1, f^\delta_2, f^\delta_4 = 0, \hspace{2 pt} f^\delta_3 = 1},	& \text{for } D  \in  [0,1), 	
							\\
							\setdef{f^\delta  \in  \HH_1 }{    f^\delta_1,f^\delta_2,f^\delta_4 \geq 0, \hspace{2 pt}  f^\delta_1  =  f^\delta_2  =  - f^{\delta}_3}, &\text{for } D  =  1,
							\\
							\setdef{f^\delta  \in  \HH_1 }{   f^\delta_1 + f^\delta_3  =  0, \hspace{2 pt} f^\delta_2 + f^\delta_3  =  0}, &\text{for } D  \geq  1.
						\end{cases}
					\end{aligned}
				\end{equation}
				Instead of a single direction, we find a set of directions along which the set of WE evolves. That is, for every $D$ and every $f^\delta \in \Gamma_D$, there exist $f^D \in \WW_{D}$ and $\epsilon^* > 0$ such that for all ${\epsilon \in [0,\epsilon^*)}$ the flow ${f^D + \epsilon f^\delta}$ is a WE for demand $D + \epsilon$. Note that, as in the previous example, the non-negative real line is divided into intervals such that the set of directions $\Gamma_{D}$ is constant in each interval. This possibility of nonunique WE introduces most of the technicalities that need to be addressed in the upcoming results. \oprocend
			\end{enumerate}
		}
	\end{example}
	Our first objective is to prove that the WE changes continuously with respect to the demand $D$.
	\begin{lemma}\longthmtitle{Continuity of the map $D \mapsto \WW_D$}\label{lem:continuity-under-affine-costs}
		Let $(\PP,\CC)$ be given. The (set-valued) map $D \mapsto \WW_{D}$ is continuous; that is, for every $D \geq 0$ and $\epsilon > 0$ the following hold:
		\begin{enumerate}
			\item There exists a $\delta > 0$ such that for all $f^D \in \WW_{D}$ and $T \ge 0$ satisfying $\abs{T - D} < \delta$ there exists $f^T \in \WW_{T}$ such that $\norm{f^D - f^T} < \epsilon$. 
			\item There exists a $\delta > 0$ such that for all $T \ge 0$ satisfying $\abs{T - D} < \delta$ and all $f^T \in \WW_{T}$ there exists $f^D \in \WW_{D}$ such that $\norm{f^D - f^T} < \epsilon$. 
		\end{enumerate}
	\end{lemma}
	The first and second part of the above definition are known as lower- and upper-semicontinuity, respectively \cite{JPA-HF:09}, or alternatively as inner- and outer-semicontinuity, respectively~\cite{RTR-RJBW:09}. 
	The result follows from \cite[Theorem 4.2]{SMR:07} and noting that since $A$ is positive semi-definite and symmetric, it is \emph{cocoercive} \cite[Theorem 3.3]{SMR:07}.
	With the above continuity property established, we can start our investigation into the evolution of the set of WE by looking at the evolution of the active and used sets. Our next result shows how these sets evolve as the demand varies in the interval $[D_i,D_{i+1}]$, where ${D_i,D_{i+1} \in \DD}$. Recall that the set $\DD$ contains the points where the active and used sets change (see Corollary~\ref{cor:interval-of-active-set}).
	\begin{lemma}\longthmtitle{Relationship between active and used sets over an interval}\label{lem:inclusion-active-and-used}
		For a given $(\PP,\CC)$, let $D_i,D_{i+1} \in \DD$ and let $\Ja_i \subseteq \PP$ and $\Ju_i \subseteq \PP$ be the associated active  and used sets on $(D_i,D_{i+1})$, respectively. Then, we have
		\begin{subequations}
			\begin{align}
				\uset_{D_i} &\subseteq \Ju_i \subseteq 	\Ja_i \subseteq \aset_{D_i},	\label{eq:series-inclusions-one}
				\\
				\uset_{D_{i+1}} &\subseteq \Ju_i 	\subseteq \Ja_{i} \subseteq \aset_{D_{i+1}}. \label{eq:series-inclusions-two}
			\end{align}
		\end{subequations}
	\end{lemma} 
	\ifinclude{
		\begin{proof}
			We will first establish~\eqref{eq:series-inclusions-one}. Let $p \in \aset_{D_i}$ and $r \in (\aset_{D_i})^c$. It follows that $\lmvec_p(D_i) < \lmvec_r(D_i) $. Continuity of the map $D \mapsto \WW_{D}$, as proven in Lemma~\ref{lem:continuity-under-affine-costs}, implies continuity of the map $\lmvec(\cdot)$. Therefore, it follows that for small enough $\epsilon > 0$ we have $\lmvec_p(T) < \lmvec_r(T)$ for all $T \in [D_i,D_i + \epsilon)$. It follows that $r \in (\aset_T)^c$ for all $T \in [D_i,D_i + \epsilon)$, which shows that $r \in (\Ja_i)^c$. Thus we have $\Ja_i \subseteq \aset_{D_i}$. Similarly, it follows from Lemma~\ref{lem:continuity-under-affine-costs} that for small enough $\epsilon > 0$ there exists $f^T \in \WW_{T}$ such that $f^T_p > 0$ for all $p \in \uset_{D_i}$ and $T \in[D_i,D_i + \epsilon)$. This shows that $\uset_{D_i} \subseteq \uset_T$ for all $T \in [D_i,D_i + \epsilon)$. Thus we have $\uset_{D_i} \subseteq \Ju_{i}$. By definition of the active and used sets, we also have $\Ju_{i} \subseteq \Ja_i$, proving~\eqref{eq:series-inclusions-one}. The result~\eqref{eq:series-inclusions-two} concerning $\uset_{D_{i+1}}$ and $\aset_{D_{i+1}}$ follows by the same arguments considering the interval $(D_{i+1} - \epsilon,D_{i+1}]$.
		\end{proof}
	}
	The implication of the above is that when the demand $D$ moves from the point $D_i$ into the interval $(D_i,D_{i+1})$, the used set $\uset_{D}$ can only gain elements, while the active set $\aset_{D}$ can only lose elements. When the demand $D$ then moves from the interval $(D_i,D_{i+1})$ to the point $D_{i+1}$ the situation is reversed. That is, $\uset_{D}$ can only lose elements, while  $\aset_{D}$ can only gain elements. Also note that since $\Ju_i \neq \Ju_j$ and $\Ja_i \neq \Ja_j$ for all $i \neq j$, both the active and the used set must change as $D$ moves from $(D_{i-1},D_i)$ to $(D_{i},D_{i+1})$.  Turning our attention back to Example~\ref{ex:secIV-1}, using~\eqref{eq:Wheatstone-simple-path-cost} and~\eqref{eq:Wheatstone-simple-WE}, we derive $\aset_{D}$ and $\uset_{D}$ as:
	\begin{equation*}
		\begin{aligned}
			(\aset_{D} , \uset_{D})  =  \begin{cases}
				(\{p_3\}, \emptyset)	 &\text{for } D = 0,	\\
				(\{p_3\}, \{p_3\}) &\text{for } D \in (0,1),	\\
				(\{p_1,p_2,p_3\}, \{p_3\})	 &\text{for } D = 1,	\\
				(\{p_1,p_2,p_3\}, \{p_1,p_2,p_3\}) &\text{for } D \in (1,2),	\\
				(\{p_1,p_2,p_3\}, \{p_1,p_2\}) &\text{for } D = 2,	\\
				(\{p_1,p_2\}, \{p_1,p_2\})	 &\text{for } D \in (2,\infty).	\\
			\end{cases}
		\end{aligned}
	\end{equation*}
	This illustrates the behavior stated in Lemma~\ref{lem:inclusion-active-and-used}. This lemma also allows us to establish the following minor extension of Lemma~\ref{lem:facts-on-evolution-for-affine}:
	\begin{corollary} \longthmtitle{Convex combinations of WE in $[D_i,D_{i+1}]$} \label{cor:extended-affine-combinations}
		For a given $(\PP,\CC)$ and $D_i, D_{i+1} \in \DD$, let $D,T \in [D_i,D_{i+1}]$. Then, for any  ${f^{D} \in \WW_{D}}$, $f^{T} \in \WW_{T}$, and ${\mu \in [0,1]}$, we have $\coco_\mu(f^{D},f^{T}) \in \WW_{\coco_\mu(D,T)}$. 
	\end{corollary}
	\ifinclude{
		\begin{proof}
			The result for the case $D,T \in (D_i,D_{i+1})$ is already stated in Lemma~\ref{lem:facts-on-evolution-for-affine}. Now consider the case ${D = D_i}$ and ${T = D_{i+1}}$. For given $\mu \in [0,1]$, $f^D \in \WW_D$, and $f^T \in \WW_T$, denote ${f^\mu := \coco_\mu(f^{D},f^{T})}$. Since $f^{D}, f^{T} \geq 0$, we have $f^\mu \geq 0$, and it follows in a straigthforward manner that $f^\mu \in \FF_{T_\mu}$, where $T_\mu := \coco_\mu(D,T)$.
			
			Now let $p \in \PP$ be a path such that $f^\mu_p > 0$. It follows that either $f^{D}_p > 0$ or $f^{T}_p > 0$. Since $D = D_i$ and $T = D_{i+1}$, we obtain either $p \in \uset_{D_i}$ or $p \in \uset_{D_{i+1}}$. In both cases, Lemma~\ref{lem:inclusion-active-and-used} implies that $p \in \Ju_{i}$, and subsequently the same result implies that $p \in \aset_{D_i} \cap \aset_{D_{i+1}}$. Thus, we have
			\begin{align*}
				C_p(f^{D}) &\leq C_r(f^{D}), \quad \forall 	r \in \PP,
				\\
				C_p(f^{T}) &\leq C_r(f^{T}), \quad \forall 	r \in \PP.
			\end{align*}
			Since the function $C$ is affine and $f^\mu$ is a convex combination of $f^D$ and $f^T$, we get $C_p(f^\mu) \leq C_r(f^\mu)$ for all $r \in \PP$. This establishes the WE condition \eqref{eq:WE-condition} and thus we have shown that $f^\mu \in \WW_{T_\mu}$. The cases $D = D_i$, $T \in (D_i,D_{i+1})$ and $D \in (D_i,D_{i+1})$, $T = D_{i+1}$ follow using similar arguments.
		\end{proof}
	}
	The main goal of this section is to characterize the evolution of the set of WE. In Example~\ref{ex:secIV}, we noted the existence of a (set of) \emph{direction(s)} $f^\delta$ along which the WE moved when the demand increases in the interval between points of $\DD$. We call such a direction a \emph{direction of increase}. The collection of all such vectors, denoted $\Gamma_D$, is referred to as the set of directions of increase. Our goal is thus to characterize this set $\Gamma_{D}$, which is formally defined as follows:
	\begin{definition} \longthmtitle{Set of directions of increase} \label{def:directions-of-increase}
		Let $(\PP,\CC,D)$ be given. The set of \emph{directions of increase} $\Gamma_D$ is the set of all directions $f^\delta \in \HH_{1}$ in which the flow can be increased, starting from some WE in $\WW_{D}$, such that the new flow is a WE as long as the increase is small enough. That is,
		\begin{align*}
			\Gamma_{D} := \setdef{f^\delta &\in \HH_{1}
			}{
				\exists f^D \in \WW_{D}, \hspace{2 pt} 	\bar{\epsilon} > 0 \text{ such that } f^{D} + \epsilon f^\delta \in \WW_{D + \epsilon} \hspace{2 pt} \forall \epsilon \in [0,\bar{\epsilon}]}. \oprocendsymbol
		\end{align*}
	\end{definition} 
	This definition of $\Gamma_{D}$ puts emphasis on the local properties of this set. However, we note that for any $D \in [D_i,D_{i+1})$ every WE $f^T \in \WW_{T}$, where $T \in (D,D_{i+1}]$, is of the form $f^D + (T - D)f^\delta$, where $f^D \in \WW_{D}$ and $f^\delta \in \Gamma_{D}$. In other words, all WE in the range $[D,D_{i+1}]$ are characterized by $\WW_{D}$ and $\Gamma_{D}$. This fact can be derived from Corollary~\ref{cor:extended-affine-combinations}.
	
	\subsection{Characterizing $\Gamma_{D}$}
	Before we address the evolution of the WE by deriving an expression for $\Gamma_D$, we first shift our attention to the evolution of the WE-cost $\lmWE$ and the WE cost-vector $\lmvec$. We characterize the evolution of the WE cost-vector $\lmvec$ and relate it to the set $\Gamma_{D}$. 

	\begin{proposition} \label{prop:evolution-lmvec} \longthmtitle{The evolution of $\lmvec$}
		Let $(\PP,\CC)$ be given. For any ${i \in [M]_0}$, there exists a vector $\delta C^i$ such that the following hold:
		\begin{enumerate}
			\item for all $D \in [D_i,D_{i+1})$ and $f^\delta \in \Gamma_{D}$
			\begin{equation} \label{eq:fdelta-hyperplane}
				Af^\delta = \delta C^i,
			\end{equation}
			\item for all $T \in [D_i,D_{i+1}]$
			\begin{equation}\label{eq:lm-t-d}
				\lmvec(T) = \lmvec(D_i) + (T - D_i)\delta C^i,
			\end{equation}
			\item $\delta C^i \neq \delta C^{i+1}$.
		\end{enumerate}
	\end{proposition}
	\ifinclude{
		\begin{proof}
			We start with the second claim, which is a consequence of the affine form of $C$, given in \eqref{eq:path-cost}, and the convexity result in Corollary~\ref{cor:extended-affine-combinations}. Let $f^{D_i} \in \WW_{D_i}$ and $f^{D_{i+1}} \in \WW_{D_{i+1}}$ and define
			\begin{equation*}
				f^{\delta_0} := (D_{i+1} - D_i)^{-1}(f^{D_{i+1}} - f^D).
			\end{equation*}
			By Corollary~\ref{cor:extended-affine-combinations} any convex combination of $f^{D_i}$ and $f^{D_{i+1}}$ is a WE. To be specific, pick some $\mu \in [0,1]$ and let $f^{\mu} := \coco_\mu(f^{D_i},f^{D_{i+1}})$. Then, $f^{\mu} \in \WW_{T_\mu}$, where ${T_\mu = \coco_\mu(D_i,{D_{i+1}})}$. Furthermore, we have 
			\begin{align*}
				f^\mu	&= f^{D_i} + (1-\mu)(f^{D_{i+1}} - 	f^{D_i})	\\
				&= f^{D_i} + (T_\mu - D_i)\frac{f^{D_{i+1}} 	- f^{D_i}}{{D_{i+1}} - D_i}	\\
				&= f^{D_i} + (T_\mu - D_i)f^{\delta_0}.
			\end{align*}
			Using the affine nature of $C$, given in \eqref{eq:path-cost}, and the definition of $\lmvec$ we derive
			\begin{equation} \label{eq:affine-variation-cost}
				\begin{aligned}
					\lmvec(T_{\mu}) &= C(f^\mu)	\\
					&= C\big(f^{D_i} + (T_\mu - 	D_i)f^{\delta_0}\big)	\\
					&= C(f^{D_i}) + (T_\mu - D_i) A 	f^{\delta_0}	\\
					&= \lmvec(D_i) + (T_\mu - D_i) A 	f^{\delta_0}.
				\end{aligned}
			\end{equation}
			Since the above holds for all $\mu \in [0,1]$ it follows that, for all  $T \in [D_i,D_{i+1}]$, we have ${\lmvec(T) = \lmWE(D_i) + (T_\mu - D_i) A f^{\delta_0}}$. Setting $\delta C^i := A f^{\delta_0}$ the second statement is proven.
			
			To show the first statement, let $D^- \in [D_i,D_{i+1})$ and let $f^\delta \in \Gamma_{D^-}$. It follows that there exists $D^+ \in (D^-,D_{i+1}]$ and $f^{D^+} \in \WW_{D^+}$ such that
			\begin{equation*}
				f^{D^+} = f^{D^-} + (D^+-D^-)f^\delta.
			\end{equation*}
			Using the same derivation as in \eqref{eq:affine-variation-cost} we find
			\begin{align}
				\lmvec(D^+) = \lmvec(D^-) + (D^+ - D^-)A f^\delta. \label{eq:comp1}
			\end{align}
			However, from \eqref{eq:lm-t-d} we have
			\begin{align}
				\lmvec(D^+) &= \lmvec(D_i) + (D^- - 	D_i)\delta C^i + (D^+ - D^-) \delta C^i	\notag
				\\
				&= \lmvec(D^-) + (D^+ - D^-) \delta C^i, 	\label{eq:comp2}
			\end{align}
			where we have again used $\lmvec(D^-) = \lmvec(D_i) + (D^- - D_i)\delta C^i$. Comparing expressions~\eqref{eq:comp1} and~\eqref{eq:comp2}, we get  $A f^\delta = \delta C^i$, proving the first statement.
			
			We show the third statement using contradiction. To this end, assume that ${\delta C^i = \delta C^{i+1}}$ for some $i \in [M]_0$ and let ${p,r \in \Ja_{i}}$. By definition of the active set it follows that $\lmvec_{p}(D) = \lmvec_{r}(D)$ for all $D \in (D_i,D_{i+1})$. Combined with \eqref{eq:lm-t-d} this implies $\delta C^i_p =\delta C^i_r$. Under the assumption that $\delta C^i = \delta C^{i+1}$ we then have, by \eqref{eq:lm-t-d}, that $\lmvec_{p}(D) = \lmvec_{r}(D)$ for all $D \in (D_i,D_{i+2})$ and $p,r \in \Ja_{i}$. As a consequence either all of the paths in $\Ja_{i}$ are in the active set on the interval $(D_{i+1},D_{i+2})$ or none of them are. That is, one of the following holds: $\Ja_{i} \cap \Ja_{i+1} = \emptyset$ or $\Ja_{i} \subseteq \Ja_{i + 1}$. However, note that from Lemma~\ref{lem:inclusion-active-and-used} we obtain $\uset_{D_{i+1}} \subseteq \Ju_i \subseteq \Ja_{i}$ and $\uset_{D_{i+1}} \subseteq \Ju_{i+1} \subseteq \Ja_{i + 1}$. Since $D_{i+1} > 0$,  we have $\uset_{D_{i+1}} \neq \emptyset$ and it follows that $\Ja_{i} \cap \Ja_{i+1}\neq \emptyset$. Thus, we must have  $\Ja_{i} \subseteq \Ja_{i + 1}$.
			
			Now since $\Ja_{i} \neq \Ja_{i + 1}$, there must exist some path $r' \in (\Ja_{i})^c \cap \Ja_{i+1}$. Since $r'$ is not in the active set $\Ja_{i}$, one can find a path $p' \in \Ja_i$ and a demand $D \in (D_{i},D_{i+1})$ such that $\lmvec_{p'}(D) < \lmvec_{r'}(D)$. On the other hand, by Lemma~\ref{lem:inclusion-active-and-used}, $r' \in \Ja_{i+1}$ implies ${r' \in \aset_{D_{i+1}}}$ and so,  $\lmvec_{r'}(D_{i+1}) \le \lmvec_{p'}(D_{i+1})$. Combining the above two facts with the affine form~\eqref{eq:lm-t-d}, we deduce that $\delta C^i_{p'} > \delta C^i_{r'}$. Since $\delta C^{i+1} = \delta C^i$ it follows from~\eqref{eq:lm-t-d} that $\lmvec_{p'}(D) > \lmvec_{r'}(D)$ for all $D \in (D_{i+1},D_{i+2})$. However, this contradicts the fact that $p' \in \Ja_i \subseteq \Ja_{i+1}$. 
			
			We see that we arrive at a contradiction, and therefore the premise must be false. We conclude that ${\delta C^{i+1} \not = \delta C^i}$.
		\end{proof}
	}

	Proposition~\ref{prop:evolution-lmvec} shows that directions of increase at any demand $D$ belong to a subspace and the map defining this subspace relates to the map $\lmvec$. The above result also shows that $\lmvec$ is affine on the intervals between the breakpoints in $\DD$ and non-differentiable at the points in $\DD$. The latter fact is due to the last conclusion in the result. 
	Below we prove a similar result for the evolution of $\lmWE$.
	\begin{corollary} \label{cor:evolution-lmWE} \longthmtitle{Evolution of $\lmWE$}
		Let $(\PP,\CC)$ be given. For any ${i \in [M]_0}$, there exists a value $\delta \lambda^i \geq 0$ such that for all $T \in [D_i,D_{i+1}]$
		\begin{equation*}
			\lmWE(T) = \lmWE(D_i) + (T-D_i) \delta 	\lambda^i.
		\end{equation*}
		Furthermore, $\delta \lambda^i = \min_{r \in \Ja_{i}}\delta C^i_r$.
	\end{corollary}
	\begin{proof}
		The result follows from \eqref{eq:lm-t-d} and noting that for any $i \in [M]_0$ and ${D \in (D_i,D_{i+1})}$ we have $\lmWE(D) = \lmvec_{p}(D)$ for any $p \in \Ja_{i}$.
	\end{proof}
	We illustrate Proposition~\ref{prop:evolution-lmvec} and Corollary~\ref{cor:evolution-lmWE}, as well as some of the nuances of the evolution of the cost with an example:		
	\begin{example}\longthmtitle{$\lmWE$ can be differentiable at points in $\DD$} \label{ex:break-in-active-set-vs-break-in-cost}
		{\rm
			\begin{enumerate}[wide= 0pt,label=(\alph*), ref=\ref{ex:break-in-active-set-vs-break-in-cost}\alph*]
				\item  \label{ex:break-in-active-set-vs-break-in-cost-1} First we show that, even though we know $\delta C^i \neq \delta C^{i+1}$, we can have $\delta C^i = \delta C^j$ whenever $j \notin \{i-1,i,i+1\}$. That is, the direction in which $\lmvec$ varies as demand changes is different in consecutive intervals between breakpoints but it can be same for non-consecutive intervals. Note that this differs from how active and used sets vary, (see Corollary~\ref{cor:interval-of-active-set}),since these sets are not same for any pair of intervals. 
				Consider the network depicted in Figure~\ref{fig:wheatstone-with-parallel}, which is the Wheatstone network with the additional edge $e_6$. 
				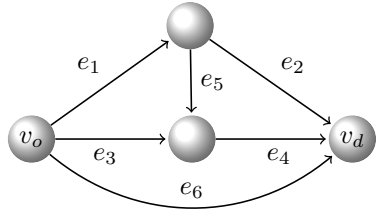
\begin{figure}
					\centering
					\begin{tikzpicture}[,->,shorten >=1pt,auto,node distance=1.5 cm,
						semithick]
						
						\tikzset{VertexStyle/.style = {shape          = circle,
								ball color     = white,
								text           = black,
								inner sep      = 2pt,
								outer sep      = 0pt,
								minimum size   = 18 pt}}
						\tikzset{EdgeStyle/.style   = {
								double          = black,
						}}
						\tikzset{LabelStyle/.style =   	{draw,
								fill           = white,
								text           = black}}
						\node[VertexStyle](A){$v_o$};
						\node[VertexStyle, above right=of 	A, xshift=6mm](B){};
						\node[VertexStyle, right=of A](C){};
						\node[VertexStyle, right=of C](D){$v_d$};
						
						\path (A)  edge      node{$e_1 $} (	B)
						(A)  edge             	node[below]{$e_3 \enskip$} (C)
						(B)	edge			node{$e_5$}	(C)
						(B)	edge			node{$e_2$}	(D)
						(C)	edge				node[below]{$\enskip e_4$}	(D)
						(A) edge[bend right = 40]      	node{$e_6 $} (D);
					\end{tikzpicture}
					\caption{The Wheatstone network with an added parallel path.}
					\label{fig:wheatstone-with-parallel}
				\end{figure}
				For $e_1$, $e_2$, $e_3$, $e_4$, and $e_5$ we use the cost functions given in \eqref{eq:Wheatstone-simple-edge-costs}, and for the new edge $e_6$ we use
				\begin{equation*}
					C_{e_6}(f_{e_6}) := 2.1.
				\end{equation*}
				There are four paths through this network, given by ${p_1 := (e_1,e_2)}$, $p_2 := (e_3,e_4)$, $p_3 := (e_1,e_5,e_4)$ and ${p_4 := (e_6)}$. The resulting path-cost function is given by 
				\begin{equation*}
					C(f) = Af + b = 	\left(\begin{array}{ccccc}
						1	&0	&1	&0	\\
						0	&1	&1	&0	\\
						1	&1	&2	&0	\\
						0	&0	&0	&0	\\
					\end{array}\right)f + 	\left(\begin{array}{c}
						1	\\
						1	\\
						0	\\
						2.1	\\
					\end{array}\right)
				\end{equation*}
				As long as the WE-cost of this game is lower than the constant cost of path $p_4$, the WE will be the same as that of the game in Example~\ref{ex:secIV-1}, with in addition $f_{p_4} = 0$. Furthermore, as demand increases and the WE-cost reaches the constant cost of path $p_4$, all subsequent flow will be routed on path $p_4$.  That is, we have
				\begin{equation*}
					f^D  = 
					\begin{cases}
						\left(\begin{array}{cccc}
							0,	&0,	&D,	&0
						\end{array}\right)^\top & \text{for 	} D \in [0,1],	\\
						\left(\begin{array}{cccc}
							D  - 1,	&D  - 1,	&2 - D,	&0
						\end{array}\right)^\top & \text{for 	} D \in [1,2],	\\
						\left(\begin{array}{cccc}
							\frac{D}{2},	&\frac{D}{2},		&0,	&0
						\end{array}\right)^\top & \text{for 	} D \in [2,2.2],	\\
						\left(\begin{array}{cccc}
							1.1,	&1.1,	&0,	&D  -  2.2
						\end{array}\right)^\top & \text{for } D \in [2.2,\infty).
					\end{cases}
				\end{equation*}
				Consequently, we obtain
				\begin{equation*}
					\begin{aligned}
						\lmvec(D)  =  \begin{cases}
							\left(\begin{array}{cccc}
								1  +  D,	&1  +  D,		&2D,	&2.1  
							\end{array}\right)^\top  		&\text{for } D  \in  [0,1],	\\
							\left(\begin{array}{cccc}
								2,	&2,	&2,	&2.1  
							\end{array}\right)^\top  		&\text{for } D  \in  [1,2],	\\
							\left(\begin{array}{cccc}
								1  +  \frac{D}{2},	&1  	+  \frac{D}{2},	&D,	&2.1  
							\end{array}\right)^\top  		&\text{for } D  \in  [2,2.2],\\
							\left(\begin{array}{cccc}
								2.1,	&2.1,	&2.2,		&2.1  
							\end{array}\right)^\top  		&\text{for } D  \in  [2.2,\infty).
						\end{cases}
					\end{aligned}
				\end{equation*}
				\begin{figure}
					\centering
					\includegraphics[width = 0.6\textwidth]{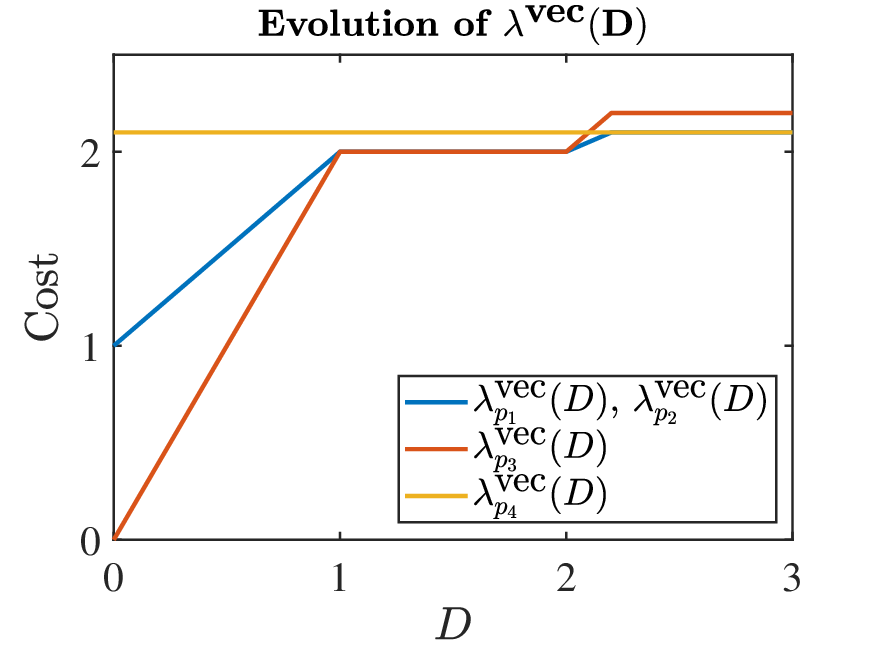}
					\caption{The evolution of $\lmvec(D)$ for the routing game discussed in Example~\ref{ex:break-in-active-set-vs-break-in-cost-1}.}
					\label{fig:constant-cost-evolution}
				\end{figure}
				In Figure~\ref{fig:constant-cost-evolution}, we plot these maps, and it is immediately apparent that on the intervals $D \in (1,2)$ and $D \in (2.2,\infty)$ the cost of all paths remain constant. In other words, we have $\delta C^1 = \delta C^3 = \mymathbf{0}$.
				
				\item \label{ex:break-in-active-set-vs-break-in-cost-2} Note that even though $\lmvec$ is necessarily not differentiable at the points in $\DD$, the same does not hold for $\lmWE$. To see this, again consider the network in Figure~\ref{fig:wheatstone-with-parallel}, where as before the cost functions of the edges $e_1$, $e_2$, $e_3$ and $e_4$ are given by \eqref{eq:Wheatstone-simple-edge-costs}, but for the edges $e_5$ and $e_6$ we set
				\begin{equation*}
					C_{e_5}(f_{e_5}) := f_{e_5}, \quad		C_{e_6}(f_{e_6}) := 2 + f_{e_6}.
				\end{equation*}
				The resulting path-cost function is given by
				\begin{equation*}
					C(f) = Af + b = 	\left(\begin{array}{ccccc}
						1	&0	&1	&0	\\
						0	&1	&1	&0	\\
						1	&1	&3	&0	\\
						0	&0	&0	&1	\\
					\end{array}\right)f + 	\left(\begin{array}{c}
						1	\\
						1	\\
						0	\\
						2	\\
					\end{array}\right)
				\end{equation*}
				and we get the following expression for the WE:
				\begin{equation*}
					f^D  = 
					\begin{cases}
						\left( \begin{matrix} 0, & 0, & D, 	& 0 \end{matrix} \right)^\top    &  \text{for } D \in [0,\frac{1}{2}],
						\\
						\frac{1}{3}\left( \begin{matrix}
							2D  -  1, 	&2D  -  1,	&2  -  	D, 	&0
						\end{matrix}\right)^\top   & 	\text{for } D \in [\frac{1}{2},2],	
						\\
						\frac{1}{3}\left(\begin{matrix}
							D + 1,	&D + 1,	&0,	&D - 2
						\end{matrix}\right)^\top   	&\text{for } D \in [2,\infty).
					\end{cases}
				\end{equation*}
				Note that for $D \in (\frac{1}{2},2)$ and $D \in (2,\infty)$ we have ${\aset_{D} = \{p_1,p_2,p_3\}}$ and ${\aset_D = \{p_1,p_2,p_4\}}$, respectively. Therefore, $\lmvec$ should not be differentiable at $D = 2$, which is verified by noting that
				\begin{equation*}
					\lmvec_{p_4}(D) = \begin{cases}
						2	&\text{for } D \in [0,2],	\\
						\frac{1}{3}D + \frac{4}{3}		&\text{for } D \in [2,\infty).
					\end{cases} 
				\end{equation*}
				However, $\lmWE$ is given by
				\begin{equation*}
					\lmWE(D) = \begin{cases}
						3D \quad& \text{for } D \in 	[0,\frac{1}{2}],	\\
						\frac{1}{3}D + \frac{4}{3} \quad & 	\text{for } D \in [\frac{1}{2}, \infty).
					\end{cases}
				\end{equation*}
				The full evolution of $\lmvec$ and $\lmWE$ is depicted in Figure~\ref{fig:Cost-differentiable}.
				\begin{figure}
					\centering
					\includegraphics[width = 0.6\textwidth]{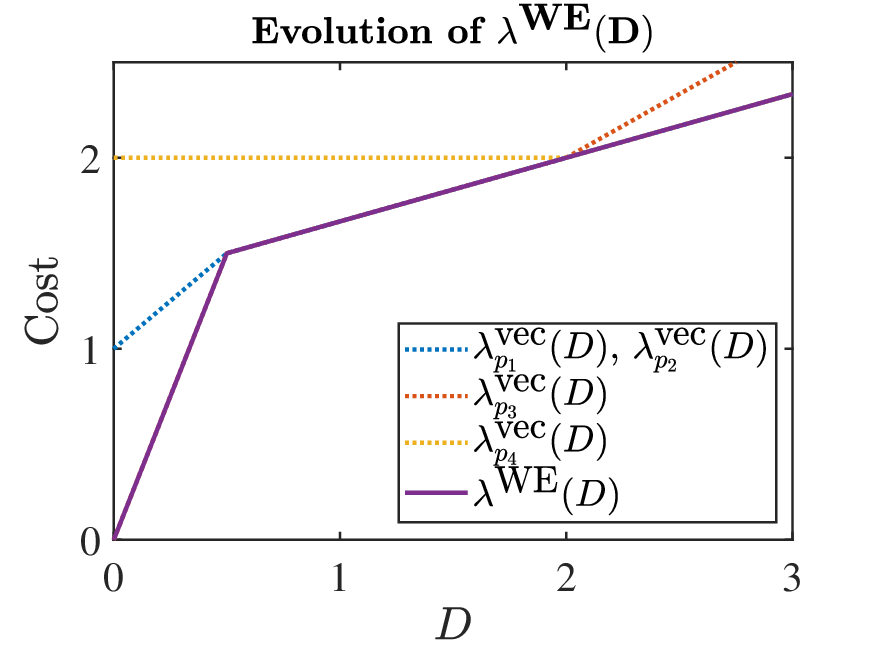}
					\caption{The evolution of $\lmWE$ and $\lmvec_{p_4}$ for the routing game discussed in Example~\ref{ex:break-in-active-set-vs-break-in-cost-2}.}
					\label{fig:Cost-differentiable}
				\end{figure}
				We see that $\lmWE$ is differentiable on $(\frac{1}{2},\infty)$, despite the fact that $\lmvec_{p_3}$ and $\lmvec_{p_4}$ clearly show a breakpoint at $D = 2$.
				\oprocend
			\end{enumerate}
		}
	\end{example}
	With the evolution of $\lmvec$ and $\lmWE$ established, our next goal is to characterize $\Gamma_D$ in a comprehensive way. 
	For ease of exposition we first define the set of \emph{directions of feasibility}. 
	\begin{definition}\longthmtitle{Direction of feasibility}\label{def:feasiblity}
		Let $(\PP,\CC,D)$ be given. 
		The set of \emph{directions of feasibility} $\MM_D$ is the set of all directions $f^\delta \in \HH_{1}$ in which the flow can be increased such that the no flow is assigned to or taken from paths that are inactive under WE, and a non-negative flow is assigned to paths that are unused under WE. That is,
		\begin{equation} \label{eq:feasible-increase-direction}
			\MM_D  :=  \setdef{f^{\delta}  \in  \HH_1 }{ f_{\aset_D \setminus \uset_D}^{\delta} \geq 0, \enskip  f_{(\aset_D)^c}^{\delta} = 0},
		\end{equation}
		where $(\aset_D)^c = \PP \setminus \aset_D$. \oprocend
	\end{definition}
	Intuitively, for any direction $f^\delta$ in $\MM_D$ there exists a WE $f^D \in \WW_{D} $ and a value $\bar{\epsilon} > 0$ such that for all $\epsilon \in [0,\bar{\epsilon}]$ we have $f^D + \epsilon f^\delta \in \FF_{D + \epsilon}$, and no flow is routed onto inactive paths. That is, there exists a WE in $\WW_{D}$ such that moving from that WE in the direction $f^\delta$ for a small enough amount will keep the resulting flow feasible, and contained within the set of paths that was active. The main result of this section is that $\Gamma_{D}$ can be obtained as the set of solutions of a VI problem, where the feasible set is given by $\MM_D$. Recall that $A \in \real^{n \times n}$ is the matrix related to the flow dependent part of the path-cost function $C$; i.e. $C(f) = Af + \beta$. We will show that the following holds:
	\begin{proposition} \longthmtitle{Directions of increase as solutions to a VI} \label{prop:characterize-directions-of-increase}
		Let $(\PP,\CC,D)$ be given. Then,
		\begin{equation*}
			\Gamma_D = \SOL(\MM_D,A).
		\end{equation*}
	\end{proposition}
	Working towards a proof of the above, we first show two intermediate statements. The first states that the set of directions of increase is contained in the set of directions of feasibility.
	\begin{lemma}\longthmtitle{$\Gamma_D$ is a subset of $\MM_D$}\label{le:gamma-m}
		Let $(\PP,\CC,D)$ be given. Then, $\Gamma_{D} \subseteq \MM_{D}$.
	\end{lemma}
	\ifinclude{
		\begin{proof}
			Let $f^\delta \in \Gamma_D$, and let $D_i,D_{i+1} \in \DD$ satisfy ${D \in [D_i,D_{i+1})}$. Then there exist $f^D \in \WW_D$, ${T \in (D, D_{i+1}]}$, and $f^T \in \WW_T$ such that ${f^T = f^D + (T-D)f^\delta}$. Thus we have ${f^\delta = (T - D)^{-1}(f^T - f^D)}$. Note that since ${D \in [D_i,D_{i+1})}$ and $T \in (D, D_{i+1}]$ it follows from Lemma~\ref{lem:inclusion-active-and-used} that $\uset_T \subseteq \aset_{D}$. Therefore $(\aset_{D})^c \subseteq (\uset_T)^c$. Consequently, by definition of the used set \eqref{eq:used-road-constraint}, we have $f_{(\uset_T)^c}^T = 0$ and so, $f_{(\aset_D)^c}^T = 0$.
			By definition of the active set \eqref{eq:relevant-road-constraint} and the WE condition \eqref{eq:WE-condition} we have $f_{(\aset_D)^c}^D = 0$. This shows that $f_{(\aset_D)^c}^\delta = 0$.
			
			Once again, by definition of the used set we have $f_{\aset_D \setminus \uset_D}^D = 0$ and by feasiblity of WE we have $f^T \geq 0$. Therefore $f_{\aset_D \setminus \uset_D}^\delta \geq 0$, which completes the proof.	
		\end{proof}
	}
	In the next statement we establish some properties of sets of the form $\SOL(\MM,A)$ where $\MM$ is of the form \eqref{eq:feasible-increase-direction}, which we require for proving Proposition~\ref{prop:characterize-directions-of-increase}.
	\begin{proposition} \longthmtitle{Properties of $\SOL(\MM,A)$}\label{prop:properties-solution-VI}
		Let $\PP$, $\CC$, $T \in \real$, and $\RR,\QQ \subseteq \PP$ satisfying $\RR \subseteq \QQ$ be given. In addition, let
		\begin{equation*}
			\MM := \setdef{f^\delta \in \HH_T}{f^\delta_{\QQ \setminus \RR} \geq 0, \quad f^\delta_{\QQ^c} = 0}.
		\end{equation*}
		We then have the following properties:
		\begin{enumerate}
			\item \label{pr:sol:1} $\SOL(\MM,A)$ is non-empty.
			\item \label{pr:sol:2} There exists a vector $\delta C \in \real^n$ such that  $f^\delta \in \SOL(\MM,A)$ if and only if $f^\delta \in \MM$ and $A f^\delta = \delta C$.
			\item \label{pr:sol:3} If $f^\delta_p > 0$ for some $f^\delta \in \SOL(\MM,A)$ and $p \in \PP$, then $A_p f^\delta = \min_{r \in \QQ} A_r f^\delta$.
			\item \label{pr:sol:4}  If $p \in \RR$ and $f^\delta \in \SOL(\MM,A)$, then ${A_p f^\delta = \min_{r \in \QQ} A_r f^\delta}$.
		\end{enumerate}
	\end{proposition}
	\ifinclude{
		\begin{proof}
			\emph{Part~\ref{pr:sol:1}:} The first claim follows by noting that ${f^\delta \in \SOL(\MM,A)}$ if and only if it solves \cite[Section 1.3.1]{FF-JSP:03}
			\begin{equation} \label{eq:optim-prob}
				\begin{aligned}
					\text{minimize } & \quad 	\frac{1}{2}f^\top A f	\\
					\text{subject to } & \quad f \in \MM.
				\end{aligned}
			\end{equation}
			Since $A$ is positive semi-definite and $\MM$ is closed and convex, this minimization problem has a non-empty solution set, showing that $\SOL(\MM,A)$ is non-empty.
			
			\emph{Part~\ref{pr:sol:2}:} To prove the ``only if'' side, let ${f^{\delta_1},f^{\delta_2} \in \SOL(\MM,A)}$. By Definition~\ref{def:VI}, we have
			\begin{equation}\label{eq:A-ineq}
				\begin{aligned}
					(A f^{\delta_1})^\top (f^{\delta_2} - f^{\delta_1}) & \ge 0,
					\\
					(A f^{\delta_2})^\top (f^{\delta_1} - f^{\delta_2}) & \ge 0.
				\end{aligned}
			\end{equation}
			Since $A$ is symmetric and positive semidefinite we have
			\begin{align*}
				0 & \le (f^{\delta_1} - f^{\delta_2})^\top A (f^{\delta_1} - f^{\delta_2}) 
				\\
				& = (A f^{\delta_1})^\top (f^{\delta_1} - f^{\delta_2})  - (A f^{\delta_2})^\top (f^{\delta_1} - f^{\delta_2}) \le 0,
			\end{align*}
			where the last inequality follows from~\eqref{eq:A-ineq}. Thus, we have $A (f^{\delta_1} - f^{\delta_2}) = 0$.
			This shows that there exists a $\delta C \in \real^n$ such that $A f^\delta = \delta C$ for all $f^\delta \in \SOL(\MM,A)$. Now we show the ``if'' part of the claim.  Let $f^{\delta_0} \in \MM$ be a flow satisfying $A f^{\delta_0} = \delta C$. We wish to show that $f^{\delta_0} \in \SOL(\MM,A)$. Since this set is nonempty, consider some $f^\delta \in \SOL(\MM,A)$. From what we have shown above, we obtain $A f^{\delta_0} = A f^\delta$. Therefore, using this equality, the following derivation holds for any $f \in \MM$:
			\begin{align*}
				(A f^{\delta_0} )^\top (f - f^{\delta_0}) 	&= (A f^{\delta} )^\top (f - f^{\delta_0})	\\
				&= (A f^{\delta} )^\top (f - f^{\delta} + 	f^{\delta} - f^\delta_0)	\\
				&= (A f^{\delta} )^\top (f - f^{\delta}) + 	(A f^{\delta} )^\top (f^{\delta} - f^{\delta_0})	\\
				&= (A f^{\delta} )^\top (f - f^{\delta}) \geq 0.
			\end{align*}
			Here the final inequality holds since $f \in \MM$ and ${f^{\delta} \in \SOL(\MM,A)}$. We see that $(A f^{\delta_0})^\top (f - f^{\delta_0}) \geq 0$ for all $f \in \MM$ and therefore $f^{\delta_0} \in \SOL(\MM,A)$. This shows that the second claim holds.
			
			\emph{Part~\ref{pr:sol:3}:} Let $f^\delta \in \SOL(\MM,A)$ and assume for the sake of contradiction that there exist $p \in \PP$ and $r \in \QQ$ such that $f^\delta_p > 0$ and $A_p f^\delta > A_r f^\delta$. Since $f^\delta_p > 0$ and $r \in \QQ$ it follows by definition of $\MM$, that $f := f^\delta - \epsilon(e_p - e_r) \in \MM$ for small enough $\epsilon > 0$. Here $e_i$ is the vector defined by $(e_i)_i = 1$ and $(e_i)_j = 0$ for all $j \neq i$. We then have
			\begin{align*}
				(A f^\delta)^\top (f - f^\delta) &= 	\epsilon (f^\delta)^\top A(e_p - e_r)	\\
				&= \epsilon ( A_p f^\delta - A_r f^\delta ) 	< 0.
			\end{align*}
			This contradicts $f^\delta \in \SOL(\MM,A)$. We conclude that if $f^\delta_p \neq 0$ then ${A_p f^\delta = \min_{r \in \QQ} A_r f^\delta}$, proving the third claim. The same arguments can be used to prove the last claim.
		\end{proof}
	}
	We note that when the set $\RR$ is empty, the first three statements of the above proposition are equivalent to well-known results on the existence, essential uniqueness, and properties of the cost of WE, which can be proven in much the same way. To be specific, when $\RR$ is empty, the set $\MM$ is simply the feasible set for the routing game with $\PP = \QQ$, and $C(f) = A_{\QQ} f + b_{\QQ}$, where $A_{\QQ}$ is the matrix $A$ with the rows and columns associated to paths in $\QQ^c$ removed, and $b_{\QQ}$ is the vector with the elements associated to $\QQ^c$ removed. The first two claims then state the existence and essential uniqueness of WE and the third claim states that any used path has minimal cost among all paths, which is the WE condition. When $\QQ = \aset_{D}$ and $\RR = \uset_{D}$, we have $\MM = \MM_D$, and the third and fourth claim gain a useful intuitive interpretation. Namely, the third claim shows that for any $f^\delta \in \SOL(\MM_{D},A)$, any path in the active set that gains flow when moving in the direction of $f^\delta$ must have the smallest increase in cost among all active paths. The same holds for any used path, regardless of how the flow on that path changes.
	
	Finally we are ready to prove Proposition~\ref{prop:characterize-directions-of-increase}.
	\ifinclude{
		\begin{proof}[Proof of Proposition~\ref{prop:characterize-directions-of-increase}]
			First we show that $\SOL(\MM_D,A) \subseteq \Gamma_{D}$. To do this we rely heavily on the conclusions of Proposition~\ref{prop:properties-solution-VI}. Note the parallelism between the set $\MM$ given there and the set $\MM_D$ used here; that is $\MM = \MM_D$ when we set  $\QQ : = \aset_D$ and $\RR := \uset_D$. Let $f^D \in \WW_{D}$ be a WE satisfying $f^D_{\uset_{D}} > 0$. Such a WE exists since the set $\WW_{D}$ is convex. To see this, let $r \in \uset_{D}$. By definition of the used set there exists a WE $f^r \in \WW_{D}$ satisfying $f^r_r > 0$. Setting $f^D = \sum_{r \in \uset_{D}} \mu_r f^r$, where $\mu_r \in (0,1)$ for all $r \in \uset_{D}$ and $\sum_{r \in \uset_{D}} \mu_r = 1$ it follows that $f^D \in \WW_{D}$ and $f^D_{\uset_{D}} > 0$. Now, for a given $f^\delta \in \SOL(\MM_{D},A)$ and $\epsilon > 0$ we write
			\begin{equation*}
				f^{D + \epsilon} := f^D + \epsilon f^\delta.
			\end{equation*}
			Note that since $f^D_{\uset_{D}} > 0$ and $f^\delta \in \MM_{D}$ we have $f^{D + \epsilon} \geq 0$ as long as $\epsilon > 0$ is small enough. We will show that for such small enough $\epsilon$ we have $f^{D +\epsilon} \in \WW_{D+\epsilon}$ which then implies ${f^\delta \in \Gamma_D}$.  Let $p \in \PP$ be a path such that $f^{D + \epsilon}_p > 0$. This implies that either (a) $f^D_p > 0$ or (b) $f^\delta_p > 0$. If $f^D_p > 0$ it follows that ${p \in \uset_{D}}$. In addition, Proposition~\ref{prop:properties-solution-VI}-\ref{pr:sol:4} tells us that ${A_p f^\delta = \min_{r \in \aset_{D}} A_r f^\delta}$. Similarly, when $f^\delta_p > 0$, the definition of $\MM_{D}$ shows that $p \in \aset_{D}$ and so from Proposition~\ref{prop:properties-solution-VI}-\ref{pr:sol:3}, we have ${A_p f^\delta = \min_{r \in \aset_{D}} A_r f^\delta}$. Also for both cases (a) and (b), we have $p \in \aset_D$ and so ${C_p(f^D) = \min_{r \in \PP}C_r(f^D)}$. Using these properties and Proposition~\ref{prop:evolution-lmvec}, it follows that
			\begin{align*}
				C_p(f^{D + \epsilon}) &= C_p(f^D) + 	\epsilon A_p f^\delta
				\\ 
				&\le \min_{r \in \PP} \bigl( C_r(f^D) + 	\epsilon A_r f^\delta \bigr)
				\\ 
				&=  \min_{r \in \PP} C_r(f^{D + \epsilon}).
			\end{align*}
			Thus $f^{D+\epsilon}_p > 0$ implies ${C_p(f^{D + \epsilon}) = \min_{r \in \PP} C_r(f^{D + \epsilon})}$ which means the the WE condition \eqref{eq:WE-condition} is satisfied, and it follows that  $f^{D+\epsilon} \in \WW_{D + \epsilon}$. Therefore, ${\SOL(\MM_{D},A) \subseteq \Gamma_{D}}$.
			
			To show $\Gamma_{D} \subseteq \SOL(\MM_D,A)$, let $f^\delta \in \Gamma_D$. Since $\SOL(\MM_{D},A) \subseteq \Gamma_{D}$, we then know that there exists an ${f^{\delta_0} \in \Gamma_{D} \cap \SOL(\MM_D,A)}$. By Lemma~\ref{le:gamma-m} we have $f^\delta \in \MM_{D}$, and by Proposition~\ref{prop:evolution-lmvec} we have ${Af^\delta = A f^{\delta_0}}$. Therefore, it follows from Proposition~\ref{prop:properties-solution-VI} that ${f^\delta \in \SOL(\MM_D,A)}$. In conclusion, we obtain ${\Gamma_{D} = \SOL(\MM_{D},A)}$.
		\end{proof}
	}

	We next illustrate the equivalence established in Proposition~\ref{prop:characterize-directions-of-increase} for the routing game associated to Figure~\ref{fig:three-node-four-edge}, as discussed in Example~\ref{ex:secIV-2}. Consider the demand $D = 1$. Using \eqref{eq:Wheatstone-double-cost} and \eqref{eq:three-node-four-edge-WE} we obtain ${\aset_{D} = \{p_1,p_2,p_3,p_4\}}$ and ${\uset_{D} = \{p_3\}}$. Thus, by definition $\MM_{D} = \setdef{f^\delta \in \HH_{1}}{f^\delta_1, f^\delta_2,f^\delta_4 \geq 0}$. We also have ${\ker(A) = \setdef{f \in \real}{f_1 = f_2 = -f_3}}$. Note that ${\ker(A) \cap \MM_{D}}$ is nonempty and in fact, from~\eqref{eq:three-node-four-edge-increase}, we have $\Gamma_D = \ker(A) \cap \MM_{D}$. Therefore, we will show next that $\SOL(\MM_D,A) = \ker(A) \cap \MM_{D}$, which corroborates the claim of Proposition~\ref{prop:characterize-directions-of-increase}. Select some $f \in \ker(A) \cap \MM_{D}$. Since $A$ is positive semi-definite we obtain
	\begin{equation*}
		(A f^\delta)^\top (f - f^\delta) \leq 0
	\end{equation*}	
	for any $f^\delta \in \MM_{D}$. Using \eqref{eq:VI-condition} in the definition of $\SOL(\MM_{D},A)$ it follows that in order for ${f^\delta \in \SOL(\MM_{D},A)}$ to hold, we need $f^\delta \in \ker(A) \cap \MM_{D}$. Furthermore, when ${f^\delta \in \ker(A) \cap \MM_{D}}$, we have $(A f^\delta)^\top (f - f^\delta) = 0$ for any $f \in \MM_{D}$. In other words, $\SOL(\MM_D,A) = \ker(A) \cap \MM_{D}$. From  \eqref{eq:three-node-four-edge-increase}, we then have $\Gamma_{D} = \SOL(\MM_{D},A)$, as claimed.
	
	The established results have interesting consequences regarding the evolution of the WE in the final interval $(D_M,\infty)$. However, before we present these, we finish our analysis of the evolution of the set of WE by establishing some basic properties of the set $\Gamma_{D}$, which are now straightforward consequences of Proposition~\ref{prop:characterize-directions-of-increase}. First, note that  $\Gamma_{D}$ is closed, which follows from the fact that $\MM_{D}$ is closed \cite[Section 1.1]{FF-JSP:03}. In addition, $\Gamma_{D}$ is convex \cite[Theorem 2.3.5]{FF-JSP:03}.
	\begin{corollary} \label{cor:Gamma-convex-closed}
		Let $(\PP,\CC,D)$ be given. Then $\Gamma_D$ is nonempty, closed, and convex.
	\end{corollary}
	In Example~\ref{ex:secIV} we noted that the non-negative real line was divided into intervals in which the set of directions of increase $\Gamma_{D}$ remains constant. Using Propostion~\ref{prop:characterize-directions-of-increase}, we can show that this is not an artifact of the considered example and that it holds for any routing game.
	\begin{lemma} \label{cor:evolution-of-GammaD} \longthmtitle{$\Gamma_{D}$ is constant between breakpoints in $\DD$}\label{le:gamma-constant}
		Let $(\PP,\CC)$ and ${D_i,D_{i+1} \in \DD}$ be given. There exists a set $\Gamma^i \subset \HH_1$ such that for all  ${D \in(D_i,D_{i+1})}$ we have $\Gamma_{D} = \Gamma^i$. Furthermore, $\Gamma_{D_i} \subseteq \Gamma^i$ and $\Gamma^{i+1} \cap \Gamma^i = \emptyset$. 
	\end{lemma}
	\ifinclude{
		\begin{proof}
			The existence of $\Gamma^i$ such that $\Gamma_{D} = \Gamma^i$ for all ${D \in(D_i,D_{i+1})}$ follows directly from Proposition~\ref{prop:characterize-directions-of-increase} and observing that the set $\MM_{D}$ is same for all $D \in (D_i,D_{i+1})$, which follows from Corollary \ref{cor:interval-of-active-set}. Let $f^\delta \in \Gamma_{D_i}$ and pick $T \in (D_i,D_{i+1}]$, $f^{D_i} \in \WW_{D_i}$ and $f^T \in \WW_{T}$ such that
			\begin{equation*}
				f^T = f^{D_i} + (T - D_i)f^\delta.
			\end{equation*}
			Pick $\mu \in (0,1)$ and let $T_\mu = \coco_{\mu}(D_i, T)$ and $f^\mu = \coco_{\mu}(f^{D_i},f^T)$. Note that ${T_\mu \in (D_i,D_{i+1})}$, ${T \in (T_\mu, D_{i+1})}$, and from Corollary~\ref{cor:extended-affine-combinations}, $f^\mu \in \WW_{T_\mu}$. Moreover, $f^T = f^\mu + (T - T_\mu) f^\delta$. 
			Therefore $f^\delta \in \Gamma_{T_\mu} = \Gamma^i$,
			showing that $\Gamma_{D_i} \subseteq \Gamma^i$. For the last statement, note that if there exists $f^\delta \in \Gamma^i \cap \Gamma^{i+1}$ then the first conclusion of Proposition~\ref{prop:evolution-lmvec} implies $\delta C^{i+1} = \delta C^i$, contradicting the third conclusion of Proposition~\ref{prop:evolution-lmvec}. Therefore  $\Gamma^{i+1} \cap \Gamma^{i} = \emptyset$.
		\end{proof}
	}
	\subsection{Obtaining $D_M$}\label{sec:obtain-DM}
	In this section, we derive several properties pertaining to the interval $(D_M,\infty)$, where we recall that $D_M$ is the final finite-valued breakpoint in $\DD$. In particular, we give a rigorous way in which the value $D_M$ can be computed. In the process and to its own merit we find the directions in which $\lmvec$ varies over $(D_M,\infty)$, as well as associated directions of increase in $\Gamma^M$. This then allows us to obtain an expression for $\lmWE(D)$ for any $D \geq D_M$. Using this we can compute the final active set $\Ja_{M}$, after which the value of $D_M$ can be obtained.
	
	We start with the result that $\Gamma^M$ overlaps with $\SOL(\FF_1,A)$. This gives us a direct method for identifying $\delta C^M$, which is the direction in which the WE cost-vector $\lmvec$ evolves in the interval $(D_M,\infty)$. This will also prove remarkably useful later on, when investigating Braess's paradox.
	\begin{proposition} \longthmtitle{Finding $\delta C_M$ by solving $\SOL(\FF_1,A)$} \label{prop:characterize-D-infinity}
		Let $(\PP,\CC)$ and ${D \geq D_M = \max\big(\DD \setminus \{\infty\}\big)}$ be given.
		Then, we have:
		\begin{enumerate}
			\item $\Gamma_{D} \cap \SOL(\FF_1,A)$ is non-empty,
			\item $A f^\delta = \delta C^M$ for all $f^\delta \in \SOL(\FF_1,A)$,
			\item $\delta \lambda^M = \min_{r \in \PP} \delta C^M_p$.
		\end{enumerate}
	\end{proposition}
	\ifinclude{
		\begin{proof}
			We note that the second statement follows from the first. To see this, let ${f^\delta \in \Gamma_{D} \cap \SOL(\FF_1,A)}$. From Proposition~\ref{prop:evolution-lmvec} we have $A f^\delta = \delta C^M$, and Proposition \ref{prop:properties-solution-VI} then gives us $A \tillf^\delta = \delta C^M$ for all $f^\delta \in \SOL(\FF_1,A)$.
			Thus, to prove the first two claims, it is enough to show that $\Gamma_{D} \cap \SOL(\FF_1,A)$ is non-empty, which we now do.
			
			First we show that $\Gamma_{D} \cap \FF_{1}$ is non-empty. For any ${f^\delta \in \Gamma_{D}}$ we already have ${f^\delta \in \HH_1}$. Therefore to establish that $\Gamma_{D} \cap \FF_{1}$ is non-empty it is enough to prove the existence of $f^\delta \in \Gamma_{D}$ such that $f^\delta \geq 0$. This we will do by constucting a sequence $\{f^{\delta, i}\}_{i \in \mathbb{N}} \subset \Gamma_{D}$ that converges to an $f^\delta$ satisfying $f^\delta \ge 0$. Then, using the fact that $\Gamma_{D}$ is closed (as stated in Corollary~\ref{cor:Gamma-convex-closed}), we conclude that $f^\delta \in \FF_1 \cap \Gamma_D$.
			
			Let $f^D \in \WW_{D}$, and let $\{f^{T_i}\}_{i \in \naturalnumbers}$ be a sequence of WE satisfying $f^{T_i} \in \WW_{T_i}$, where $D_M \le D < T_1$, $T_i < T_{i+1}$ for all $i$, and $\lim_{i \rightarrow \infty} T_i = \infty$. We then define the sequence $\{f^{\delta,i}\}_{i \in \naturalnumbers}$ by setting
			\begin{equation*}
				f^{\delta,i} := (T_i - D)^{-1}(f^{T_i} - f^D). 
			\end{equation*}
			Since $D \geq D_M$ and $T_i > D$ for all $i \in \naturalnumbers$, it follows from Corollary~\ref{cor:extended-affine-combinations} that any convex combination of $f^D$ and $f^{T_i}$ is a WE; that is, for any $\epsilon \in [0,T_i - D)$ we find that $f^D + \epsilon f^{\delta,i}$ is a WE. By Definition~\ref{def:directions-of-increase}, we deduce that $f^{\delta,i} \in \Gamma_{D}$ for all $i \in \naturalnumbers$. In addition, the sequence $\{f^{\delta,i}\}_{i \in \naturalnumbers}$ is bounded. To show this, 
			let $t := \max_{r \in \PP} f^D_r$. Since $f^{T_i} \geq 0$ for all $i \in \naturalnumbers$ we obtain $(f^{T_i}_r - f^D_r) \geq -t$ for all $r \in \PP$ and all $i \in \naturalnumbers$. Therefore, $f^{\delta, i}_r \geq -t(T_i - D)^{-1}$ for all $r \in \PP$ and $i \in \naturalnumbers$. 
			Letting $\RR_i := \setdef{r \in \PP}{f^{\delta, i}_r < 0}$ it follows that
			\begin{equation*}
				\sum_{r \in \RR_i} f^\delta_r \geq -n t (T_i - D)^{-1}.
			\end{equation*}
			Since $f^{\delta, i} \in \HH_1$ the above implies that  $f^{\delta, i}_r \leq nt(T_i - D)^{-1} + 1$ for all $r \in \PP$. Since $T_1 \leq T_i$ for all $i \in \naturalnumbers$ we get ${-t(T_1 - D)^{-1} \leq f^{\delta,i}_r \leq - n t (T_1 - D)^{-1} + 1}$ for all $r \in \PP$ and $i \in \naturalnumbers$. In other words, $\{f^{\delta, i}\}_{i \in \mathbb{N}}$ is bounded.
			Therefore, the sequence contains a subsequence, denoted $\{f^{\delta, i_k}\}_{k \in \naturalnumbers}$, that converges. For this subsequence, since $\lim_{k \rightarrow \infty} T_{i_k} = \infty$ it follows that $\lim_{k \to \infty} f^{\delta, i_k}_r \geq 0$ for all $r \in \PP$. We see that
			\begin{equation*}
				f^{\delta} := \lim_{k \to \infty} f^{\delta, i_k} \geq 0.
			\end{equation*}
			Since $\Gamma_{D}$ is closed, it follows that $f^{\delta} \in \Gamma_{D}$. Thus, there exists $f^{\delta} \in \Gamma_{D} \cap \FF_1$.
			
			The final step is to show that $f^{\delta} \in \Gamma_{D} \cap \FF_1$ implies ${f^\delta \in \SOL(\FF_{1},A)}$. This we do by showing that $f^\delta$ is a WE of the routing game with the path cost function $\tillC(f) = Af$ and the feasible set $\FF_{1}$.
			
			Given an $f^{\delta} \in \Gamma_{D} \cap \FF_1$, let $p \in \PP$ satisfy $f^\delta_p > 0$. By definition of the set $\Gamma_{D}$ we then know that for some $T > D$ there exist WE $f^T \in \WW_{T}$ satisfying $f^T_p > 0$. Since $D \geq D_M$ this implies that the path $p$ is used, and therefore active, in the interval $(D_M,\infty)$. In other words, $p \in \Ja_{M}$.
			
			For the sake of contradiction assume that $r \in \PP$ satsifies $A_r f^\delta < A_p f^\delta$. Since $f^\delta \in \Gamma_{D}$ we have $A f^\delta = \delta C^M$, which implies $\lmvec(T) = \lmvec(D) + (T - D)A f^\delta$ for any $T \in (D,\infty)$. It follows that for large enough $T$ we have $\lmvec_r(T) < \lmvec_p(T)$, which implies that $p \notin \aset_{T}$. Since ${T > D}$, and $D \geq D_M$ this implies that for some ${T \in (D_M,\infty)}$, the path $p$ is not active. In other words $p \notin \Ja_{M}$. We have arrived at a contradiction, and therefore we conclude that $A_p f^\delta = \min_{r \in \PP}A_r f^\delta$. In other words, $f^\delta$ is a WE of the routing game with the path cost function $\tillC(f) = Af$ and the feasible set $\FF_{1}$. By Proposition \ref{prop:BP-VI} this implies $f^\delta \in \SOL(\FF_{1},A)$.
			
			For the third statement, note that when $\delta C^M_r < \delta \lambda^M$ for some $r \in \PP$, this implies that in the interval $(D_M,\infty)$, the cost of path $r$ under WE increases at a slower rate than the WE-cost. Since the cost of all paths in this interval increases at a constant rate, this implies that at some demand $T > D_M$, $\lmvec_r(T) < \lmWE(T)$. However, the WE-cost is the minimal cost of all paths under WE, and therefore this is not possible. Thus we have $\delta \lambda^M = \min_{r \in \PP} \delta C^M$. This concludes the proof.
		\end{proof}
	}
	The intuition behind the above result is as follows. Consider a demand $D$ in the final interval $[D_M,\infty)$, for which we have $\Gamma_{D} \subseteq \Gamma^M$ by Lemma~\ref{cor:evolution-of-GammaD}, and consider the set of directions in which the WE moves as demand increases from $D$; that is, consider the set $\Gamma_D$. Since the interval $[D_M,\infty)$ stretches to infinity without encountering any other breakpoint, there must be some direction in $\Gamma_D$ along which we can move indefinitely. For any direction $f^\delta$ which takes flow from some path, that is the vector has some negative component, one can only move a finite amount in that direction as sooner or later the flow on that path then becomes zero, and we can no longer move in that direction. Thus, there must be some $f^\delta \in \Gamma_{D} \cap \FF_1$. In addition, when moving in such a ``nonnegative'' direction of increase the cost of all paths receiving flow must remain minimal to satisfy the WE conditions. Therefore, the cost of paths receiving flow must show the minimal increase among all paths. This implies that in fact $f^\delta \in \SOL(\FF_1,A)$, which in turn shows the above result.
	Note that these observations do not imply  $\SOL(\MM_{D},A) = \SOL(\FF_1,A)$, as for instance demonstrated in Example~\ref{ex:secIV-2}. There we see that for any $D \geq D_M = 1$ there exist $f^\delta \in \Gamma_{D}$ and $p \in \PP$ such that $f^\delta_p < 0$.
	
	Having obtained $\delta C^M$ by solving $\SOL(\FF_{1},A)$, we now give a full characterization of $\Gamma^M \cap \SOL(\FF_{1},A)$ in the below result. Note that the result considers $\Gamma_D$ for some $D \geq D_M$ instead of $\Gamma^M$, but whenever $D > D_M$ we have $\Gamma_{D} = \Gamma^M$.
	\begin{proposition} \longthmtitle{Obtaining $\Gamma^M \cap \SOL(\FF_{1},A)$ } \label{prop:nonnegative-Gamma_M}
		Let $(\PP,\CC)$ and $D \geq D_M$ be given. The set $\Gamma_D \cap \SOL(\FF_{1},A)$ is equal to the set of solutions of the following minimization problem:
		\begin{equation} \label{eq:min-beta}
			\begin{aligned}
				\textnormal{ minimize } \quad	& \beta^\top f^\delta	\\
				\textnormal{subject to } \quad  &A f^\delta = \delta C^M,\\
				&f^\delta \in \FF_{1}.
			\end{aligned}
		\end{equation}
	\end{proposition}

	\ifinclude{
		\begin{proof}
			We start with the observation that taken together Propositions \ref{prop:properties-solution-VI} and \ref{prop:characterize-D-infinity} imply ${f^\delta \in \SOL(\FF_{1},A)}$ if and only if $f^\delta \in \FF_1$ and $Af^\delta = \delta C^M$. In other words, the feasibility set of the problem~\eqref{eq:min-beta} is equal to $\SOL(\FF_{1},A)$. Thus, our goal is to show that $f^\delta$ is an optimal solution of~\eqref{eq:min-beta} if and only if it is feasible and $f^\delta \in \Gamma_D$.
			
			We proceed with the ``if'' part, that is, if $f^\delta \in \Gamma_D$ and is a feasible point of~\eqref{eq:min-beta}, then it is an optimal solution of~\eqref{eq:min-beta}. To do so, let $f^{D_M} \in \WW_{D_M}$ and consider the minimization problem
			\begin{equation} \label{eq:min-lmvec}
				\begin{aligned}
					\textnormal{ minimize } \quad	& (f^{\delta})^\top(A f^{D_M} + \beta)\\
					\textnormal{subject to } \quad  &A f^\delta = \delta C^M,\\
					&f^\delta \in \FF_{1}.
				\end{aligned}
			\end{equation}
			This minimization problem is actually equivalent to \eqref{eq:min-beta}, meaning that $f^\delta$ is an optimal solution of~\eqref{eq:min-beta} if an only if it is so for~\eqref{eq:min-lmvec}. To see this, note that since $f^{D_M} \in \WW_{D_M}$, we have $Af^{D_M} + \beta = \lmvec(D_M)$. Therefore, if $p \in \PP$ satisfies $f^{D_M}_p > 0$, then it follows that $p \in \uset_{D_M}$, and by Lemma \ref{lem:inclusion-active-and-used} we then have $p \in \Ju_{M}$. We see that $p$ is in the used set on the interval $[D_M,\infty)$, and therefore it must maintain minimal cost under WE among all paths as the demand increase from $D$ to infinity. Consequently ${\delta C^M_p= \min_{r \in \PP} \delta C^M_r = \delta \lambda^M}$. Since $D \geq D_M$ we also have, by Proposition \ref{prop:evolution-lmvec}, that $\delta C^M = A f^\delta$ for any $f^\delta \in \Gamma_D$. Therefore $A_p f^\delta = \delta \lambda^M$ whenever $f^{D_M}_p > 0$. Since $f^{D_M} \in \FF_{D_M}$ this implies
			\begin{equation} \label{eq:D_M-lambda^M}
				(f^{\delta})^\top A f^{D_M} = D_M \delta \lambda^M \quad \text{ for all } f^\delta \in \SOL(\FF_{1},A).
			\end{equation}
			We see that the term $(f^{\delta})^\top A f^{D_M}$ is constant over the feasible set of \eqref{eq:min-lmvec}. Therefore $f^\delta$ solves \eqref{eq:min-lmvec} if and only if it minimizes the term $\beta^\top f^\delta$ over the feasible set. This is exactly the objective function of \eqref{eq:min-beta}, and since the feasible sets of \eqref{eq:min-beta} and \eqref{eq:min-lmvec} are the same, this shows that \eqref{eq:min-beta} and \eqref{eq:min-lmvec} have the same set of optimal solutions.
			
			Next we obtain a lower bound on the objective function of \eqref{eq:min-lmvec}. We have ${\lmvec(D_M) = Af^{D_M} + \beta}$, and in addition $\lmWE(D_M) = \min_{r \in \PP} \lmvec_r(D_M)$. In other words, $\lmvec(D_M) \geq \lmWE(D_M) \mymathbf{1}$, where $\mymathbf{1}$ is the vector of ones. For any $f^\delta \in \FF_{1}$ all elements of $f^\delta$ are nonnegative, and sum to one, and therefore it follows that for all $f^\delta \in \SOL(\FF_{1},A)$ we have
			\begin{align*}
				(f^{\delta})^\top(A f^{D_M} + \beta) &= (f^{\delta})^\top \lmvec(D_M)	\\
				&\geq \lmWE(D_M).
			\end{align*}

			Next we show that any $f^\delta \in \Gamma_D \cap \SOL(\FF_{1},A)$ achieves this lower bound. We know from Proposition \ref{prop:characterize-D-infinity} that $\Gamma_D \cap \SOL(\FF_{1},A)$ is non-empty, so we can pick ${\tillf^\delta \in \Gamma_D \cap \SOL(\FF_{1},A)}$. By Proposition \ref{prop:characterize-directions-of-increase} we then have $\tillf^\delta \in \SOL(\MM_{D},A)$. Now let $p \in \PP$ satisfy $\tillf^\delta_p > 0$. From the definition of $\Gamma_{D}$ it follows that $p \in \uset_{T} = \Ju_M$ for some $T > D$, and by Lemma \ref{lem:inclusion-active-and-used} this implies $p \in \aset_{D_M}$. Therefore $\lmvec_p(D_M) = \lmWE(D_M)$. Since $\tillf^\delta \in \FF_1$ it follows that
			\begin{equation} \label{eq:fdelta-lmvec-lmWE}
				(\tillf^\delta)^\top \lmvec(D_M)  = \lmWE(D_M),
			\end{equation}
			showing that $\tillf^\delta$ achieves the lower bound we established for \eqref{eq:min-lmvec}. In other words, $\tillf^\delta$ is a solution of \eqref{eq:min-lmvec}, and is therefore also a solution to \eqref{eq:min-beta}. This shows that $\Gamma_D \cap \SOL(\FF_{1},A)$ is contained within the set of solutions of \eqref{eq:min-beta}. 
			
			We now proceed with the ``only if'' part. Let $f^\delta$ be a solution of \eqref{eq:min-beta}, and therefore also a solution of \eqref{eq:min-lmvec}. We know from the derivation of~\eqref{eq:fdelta-lmvec-lmWE} that this implies $(f^\delta)^\top \lmvec(D_M) = \lmWE(D_M)$. Since ${\lmvec(D_M) \geq \lmWE(D_M) \mymathbf{1}}$ and $f^\delta \in \FF_{1}$ this gives us
			\begin{equation*}
				\lmvec_{p}(D_M) = \lmWE(D_M) \quad \text{ for all } p \in \PP \text{ satisfying }f^\delta_p > 0.
			\end{equation*}
			In other words, $f^\delta_p > 0$ implies $p \in \aset_{D_M}$. Together with ${f^\delta \in \FF_{1}}$ this shows that $f^\delta \in \MM_{D_M}$. Thus we have $f^\delta \in \MM_{D_M}$ and $A f^\delta = \delta C^M$. From Proposition \ref{prop:characterize-directions-of-increase} we have $\Gamma_{D_M} = \SOL(\MM_{D_M},A)$, from Proposition \ref{prop:evolution-lmvec} we have $A\widehat{f}^\delta = \delta C^M$ for all $\widehat{f}^\delta \in \Gamma_{D_M}$ and from Proposition \ref{prop:properties-solution-VI} it then follows that $\widehat{f}^\delta \in \Gamma_{D_M}$ if and only if $\widehat{f}^\delta \in \MM_{D_M}$ and $A\widehat{f}^\delta = \delta C^M$. Thus we see that $f^\delta \in \Gamma_{D_M}$. From Corollary \ref{cor:evolution-of-GammaD} we have $\Gamma_{D_M} \subseteq \Gamma^M$, and thus we conclude that $f^\delta \in \Gamma_D$ for any $D \geq D_M$. This completes the proof.
		\end{proof}
	}
	Next, we use the obtained results to express $\lmWE(D)$ for any $D \geq D_M$.
	\begin{lemma} \longthmtitle{Obtaining $\lmWE(D)$ for $D \geq D_M$} \label{lem:obtaining-lambda-D_M}
		Let $(\PP,\CC)$ and $D \geq D_M$ be given. In addition, let $\bar{\beta} = \beta^\top f^\delta$, where $f^\delta$ is a solution of \eqref{eq:min-beta}. Then, we have
		\begin{equation*}
			\lmWE(D) = \delta \lambda^M D + \bar{\beta}.
		\end{equation*}
	\end{lemma}
	\ifinclude{
		\begin{proof}
			First pick $D > D_M$ and let $f^\delta$ be a solution of \eqref{eq:min-beta}. From Proposition~\ref{prop:nonnegative-Gamma_M}, we have $f^\delta \in \Gamma^M \cap \SOL(\FF_{1},A)$. Additionally, $\lmWE(D) = \lmWE(D_M) + (D - D_M) \delta \lambda^M$ holds. Now consider the expression
			\begin{equation*}
				(f^{\delta})^\top\Big(A\big(f^{D_M} + (D - D_M)f^\delta \big) + \beta \Big),
			\end{equation*}
			where $f^{D_M} \in \WW_{D_M}$. Since $f^{D_M} \geq 0$ and $f^\delta \geq 0$ it follows that $f^{D_M} + (D - D_M)f^\delta \geq 0$. In combination with ${f^\delta \in \Gamma^M}$, this shows that $f^{D_M} + (D - D_M)f^\delta \in \WW_{D}$. Thus, we have $A\big(f^{D_M} + (D - D_M)f^\delta \big) + \beta = \lmvec(D)$. By the same arguments used to derive \eqref{eq:fdelta-lmvec-lmWE}, we obtain ${(f^{\delta})^\top\lmvec(D) = \lmWE(D)}$. Thus we have
			\begin{equation} \label{eq:lm-last-one}
				\begin{aligned}
					(f^{\delta})^\top\Big(A\big(f^{D_M} + (D - D_M)f^\delta \big) + \beta \Big) &= (f^{\delta})^\top\lmvec(D)
					\\
					&= \lmWE(D)
				\end{aligned}
			\end{equation}
			We also have the following:
			\begin{equation} \label{eq:lm-last-two}
				\begin{aligned}
					& (f^{\delta})^\top\Big(A\big(f^{D_M} + (D - D_M)f^\delta \big) + \beta \Big)
					\\
					&= (f^\delta)^\top A f^{D_M} + (D - D_M)(f^{\delta})^\top A f^\delta + \beta^\top f^\delta. 
				\end{aligned}
			\end{equation}
			From \eqref{eq:D_M-lambda^M} we have $(f^\delta)^\top A f^{D_M} = D_M \delta \lambda^M$. From arguments similar to those for \eqref{eq:fdelta-lmvec-lmWE} we also find ${(D - D_M)(f^{\delta})^\top A f^\delta = (D - D_M) \delta \lambda^M}$. These facts combined with expressions~\eqref{eq:lm-last-one} and~\eqref{eq:lm-last-two} yield
			\begin{align*}
				\lmWE(D) &=
				(f^{\delta})^\top\Big(A\big(f^{D_M} + (D - D_M)f^\delta \big) + \beta \Big)	\\ 
				&= D\delta \lambda^M + \beta^\top f^\delta.
			\end{align*}
			Since $f^\delta$ is a solution of \eqref{eq:min-beta} we therefore have
			\begin{equation*}
				\lmWE(D) = D\delta \lambda^M + \bar{\beta}.
			\end{equation*}
			The above holds for any $D > D_M$, and from Corollary \ref{cor:evolution-lmWE} it follows that the equality also holds for $D = D_M$, which concludes the proof.
		\end{proof}
	}
	Thus far we have characterized how $\lmvec$ and $\lmWE$ vary over the final interval. Using the latter, our next result shows how we can use this information to obtain $\Ja_{M}$.
	\begin{lemma} \longthmtitle{Obtaining $\Ja_{M}$} \label{eq:finding-JaM}
		For a given $(\PP,\CC,D)$, let $f^\delta$ be an optimal solution of~\eqref{eq:min-beta}. Consider the following index sets:
		\begin{equation*}
			\begin{aligned}
				\II_1 &= \setdef{r \in \PP}{f^\delta_r > 0},
				\\
				\II_2 &= \setdef{r  \in \PP}{f^\delta_r = 0, \hspace{2 pt} \delta C^M_r > \delta \lambda^M},	
				\\
				\II_3 &= \setdef{r \in \PP}{f^\delta_r = 0, \hspace{2 pt} \delta C^M_r = \delta \lambda^M}.
			\end{aligned}
		\end{equation*}
		Let $f^*$ be a solution of the following convex minimization problem:
		\begin{subequations} \label{eq:minimize-for-active-set}
			\begin{align}
				\textnormal{minimize }  
				& \quad f^\top A f + f^\top  \beta	
				\\
				\textnormal{subject to  }&
				\quad f_r = 0 &\forall r \in \II_2,	\label{eq:minimize-for-active-set2}
				\\
				& \quad f_r \geq 0 &\forall r \in \II_3, \label{eq:minimize-for-active-set3}
				\\ 
				& \quad C_r(f) \geq \delta \lambda^M D + \bar{\beta} &\forall r \in \II_3,	\label{eq:minimize-for-active-set4}
				\\
				& \quad C_r(f) = \delta \lambda^M D + \bar{\beta} &\forall r \in \II_1,	\label{eq:minimize-for-active-set5}
				\\
				& \quad \mymathbf{1}^\top f = D. \label{eq:minimize-for-active-set6}
			\end{align}
		\end{subequations}
		We then have that $p \in \Ja_{M}$ if and only if $\delta C_p^M = \delta \lambda^M$ and $C_p(f^*) = \delta \lambda^M D + \bar{\beta}$.
	\end{lemma}
	The following is a long proof, and for this reason we first provide some intuition to clarify the underlying ideas. The optimization problem \eqref{eq:minimize-for-active-set} is designed with the aim of identifying the set $\Ja_{M}$. Note that due to the previous results we can already obtain $\delta C^M$ and $\delta \lambda^M$ as well as an element $f^\delta$ of $\Gamma^M$, which provide quite some information about $\Ja_{M}$. For instance, on the interval $(D_M,\infty)$, the cost under WE of any path $r \in \II_2$ increases faster than $\delta \lambda^M$, where the latter is the minimum over all paths of the increase in cost under WE. Thus, the cost of these paths can not remain minimal on this interval, which immediately shows that $\Ja_{M} \cap \II_2 = \emptyset$. Similarly, since $f^\delta \in \Gamma^M$ we know that any path in $\II_1$ is used under WE in the interval $(D_M,\infty)$, which shows that $\II_1 \subseteq \Ja_{M}$.
	What remains is to find which paths in $\II_3$ are in $\Ja_{M}$, which can be done by solving \eqref{eq:minimize-for-active-set}. The reason that this works is that \eqref{eq:minimize-for-active-set} is designed in such a way that for any solution $f^*$ the flow $f^* + \epsilon f^\delta$ is a WE as long as $\epsilon > 0$ is large enough. Appropriately picking $\epsilon$ will then ensure that the constructed flow is a WE with demand in $(D_M,\infty)$. Consequently the active set $\Ja_{M}$ is given by the paths that have minimal cost given the flow $f^* + \epsilon f^\delta$. Using the conditions imposed by the constraints we can show that these are exactly the paths for which $\delta C_p^M = \delta \lambda^M$ and $C_p(f^*) = \delta \lambda^M D + \bar{\beta}$, completing the argument.
	Before we start the proof, we explicitely note that the result holds for \emph{any} $D \in \real$. Therefore it can be used without prior knowledge of $D_M$.
	\ifinclude{
		\begin{proof}[Proof of Lemma~\ref{eq:finding-JaM}]
			We start by noting that from \eqref{eq:min-beta} we have $f^\delta \geq 0$ and from Proposition \ref{prop:characterize-D-infinity} we have $\delta \lambda^M = \min_{r \in \PP}\delta C^M$. This implies that $\PP = \II_1 \cup \II_2 \cup \II_3$.
			
			The first part of the proof is now to show that %
			\begin{equation*} 
				\begin{aligned}
					(f^*)^\top A f^* + (f^*)^\top \beta = D(\delta \lambda^M D + \bar{\beta})
				\end{aligned}
			\end{equation*}
			holds for any optimizer $f^*$ of~\eqref{eq:minimize-for-active-set}. To prove that this is true, we construct a specific optimal solution $f'$ of~\eqref{eq:minimize-for-active-set} in the following way. Let $f^{T} \in \WW_{T}$ for some $T > D_M$, and let 
			\begin{equation*}
				{f' := f^{T} + (D - T)f^\delta}.
			\end{equation*}
			To show that $f'$ is an optimal solution of \eqref{eq:minimize-for-active-set} our first step is to show that it satisfies all the constraints. Then we obtain the objective function value $(f')^\top A f' + (f')^\top\beta$ and finally we show that this value is optimal, proving that $f'$ is a solution of \eqref{eq:minimize-for-active-set}.
			
			\emph{Step 1: $f'$ is feasible:}
			We start by considering the constraints on the paths in $\II_1$ given in~\eqref{eq:minimize-for-active-set5}. Note that since $f^\delta$ is a solution of \eqref{eq:min-beta}, Proposition~\ref{prop:nonnegative-Gamma_M} gives $f^\delta \in \Gamma_{T}$, and Proposition~\ref{prop:characterize-directions-of-increase} then implies $f^\delta \in \SOL(\MM_{T},A)$. Also note that we have
			\begin{equation} \label{eq:cost-at-f-derivation}
				\begin{aligned}
					C(f')	&= Af' + \beta,	\\
					&= A \big(f^T + (D-T) f^\delta \big) + \beta,	\\
					&= Af^T + \beta +(D-T)Af^\delta,	\\
					&= C(f^T) + (D-T)Af^\delta.
				\end{aligned}
			\end{equation}
			Now let $p \in \II_1$, which gives $f^\delta_p > 0$, and
			since $f^\delta \in \Gamma_{T}$ is a direction of increase, this implies that $p \in \uset_{T^+}$ for some $T^+ > T$. However, since $T \in (D_M,\infty)$,  we obtain ${p \in \Ju_{M} \subseteq \Ja_{M}}$. Therefore, we have $p \in \aset_{T}$ and so, $C_p(f^T) = \lmWE(T)$. Consequently, Lemma~\ref{lem:obtaining-lambda-D_M} then tells us that $C_p(f^T) = \delta \lambda^M T + \bar{\beta}$. Furthermore, the fact that $p \in \Ju_{M}$, in combination with Proposition~\ref{prop:properties-solution-VI}, implies $A_p f^\delta = \min_{r \in \Ja_M} A_r f^\delta$. It then follows from Corollary~\ref{cor:evolution-lmWE} that $A_pf^\delta = \delta \lambda^M$. Collecting these deduced facts that $C_p(f^T) = \delta \lambda^M T + \bar{\beta}$ and $A_pf^\delta = \delta \lambda^M$ and employing them in~\eqref{eq:cost-at-f-derivation} then gives us
			\begin{align*}
				C_p(f')	&= \delta \lambda^M T + \bar{\beta} + (D - T) \delta \lambda^M,	\\
				&= \delta \lambda^M D + \bar{\beta}.
			\end{align*}
			Thus, $f'$ satisfies the constraint~\eqref{eq:minimize-for-active-set5}. Similar arguments can be used to show that any path $p \in \II_3$ satisfies $C_p(f^T) \geq \delta \lambda^M T + \bar{\beta}$ and $A_pf^\delta = \delta \lambda^M$, leading to the conclusion that ${C_p(f') \geq \delta \lambda^M D + \bar{\beta}}$. That is, $f'$ satisfies~\eqref{eq:minimize-for-active-set3}. 
			
			To show that the constraint on paths in $\II_2$ holds, let $p \in \II_2$, which by definition means $\delta C^M_p > \delta \lambda^M$. Since $f^\delta \in \Gamma_{T}$, from Proposition~\ref{prop:properties-solution-VI}, we have $\delta C^M = A f^\delta$. In combination with ${\delta \lambda^M = \min_{r \in \Ja_M} A_r f^\delta}$ we see that $\delta C^M_p > \min_{r \in \PP} \delta C^M_r$. In other words, the cost under WE of path $p$ can not remain minimal on the entire interval $(D_M,\infty)$ and therefore $p \notin \Ja_{M}$. Since $T \in (D_M,\infty)$, this implies $p \notin \aset_{T}$, which gives $p \notin \uset_{T}$. Therefore, $f^T_p = 0$. Furthermore, since $p \notin \aset_{T}$ and $f^\delta \in \SOL(\MM_{T},A)$ it follows from the definition of $\MM_{T}$ that $f^\delta_p = 0$. Consequently, $f'_p = 0$, and so, $f'$ satisfies the constraint~\eqref{eq:minimize-for-active-set2}. 
			
			For the constraint~\eqref{eq:minimize-for-active-set3} on the paths in $\II_3$, note that since $f^T \in \FF_T$, we have $f^T \geq 0$. For any $p$ with $f^\delta_p = 0$ it then follows that $f'_p \geq 0$, as required. That the final constraint~\eqref{eq:minimize-for-active-set6} holds follows from the definition of $f'$. Thus, in summary, $f'$ satisfies all constraints in \eqref{eq:minimize-for-active-set}, and is therefore feasible.
			
			\emph{Step 2: Obtaining an expression for $(f')^\top A f' + (f')^\top$:}
			The next step required for showing $f'$ is an optimal solution of~\eqref{eq:minimize-for-active-set} is to obtain an expression of ${(f')^\top A f' + (f')^\top \beta}$, which will later be shown to be the lower bound of the objective function of~\eqref{eq:minimize-for-active-set} over the feasible set. To get this expression, we show that if $f'_p \neq 0$ for some $p$, then ${C_p(f') = \delta \lambda^M D + \bar{\beta}}$. This fact is consequently used to show that ${(f')^\top A f' + (f')^\top \beta = (f')^\top C(f') = D(\delta \lambda^M D + \bar{\beta})}$. 
			
			Let $p$ be a path such that $f'_p \neq 0$. It follows that either $f^\delta_p > 0$ or ${f^\delta_T > 0}$ (note that both vectors are nonnegative, so values less than zero are not possible). For the first case, $f^\delta_p > 0$, we have $p \in \II_1$ which we have already shown implies ${C_p(f') = \delta \lambda^M D + \bar{\beta}}$. For the second case, $f^T_p > 0$, we have $p \in \uset_{T} \subseteq \aset_{T}$, and therefore we have ${C_p(f^T) = \lmWE(T) = \delta \lambda^M T + \bar{\beta}}$.
			
			Setting $\MM_{T} = \MM$ in Proposition~\ref{prop:properties-solution-VI} yields $A_p f^\delta = \min_{r \in \aset_{T}} A_r f^\delta$, and Corollary~\ref{cor:evolution-lmWE} then implies $A_pf^\delta = \delta \lambda^M$. Using these conclusions in~\eqref{eq:cost-at-f-derivation}, we obtain  $C_p(f') = \delta \lambda^M D + \bar{\beta}$. In summary, $f'_p \neq 0$ implies $C_p(f') = \delta \lambda^M D + \bar{\beta}$ and we then have
			\begin{equation} \label{eq:derivation-fstar}
				\begin{aligned}
					(f')^\top A f' + (f')^\top \beta &= (f')^\top C(f),	\\
					&= D(\delta \lambda^M D + \bar{\beta}).	
				\end{aligned}
			\end{equation}

			\emph{Step 3: $f'$ is optimal:}
			To finish proving that $f'$ is an optimal solution of \eqref{eq:minimize-for-active-set}, we show that $D(\delta \lambda^M D + \bar{\beta})$ is in fact a lower bound on $f^\top A f + f^\top \beta$ for any $f$ satisfying the constraints of~\eqref{eq:minimize-for-active-set}. Therefore, let $f$ be an element of the feasible set of \eqref{eq:minimize-for-active-set}. Consider $p$ such that $f_p \neq 0$. Since $\PP = \II_1 \cup \II_2 \cup \II_3$, the constraints in \eqref{eq:minimize-for-active-set} then imply $p \in \II_1 \cup \II_3$. From constraint~\eqref{eq:minimize-for-active-set5}, if $p \in \II_1$, then we have $C_p(f) = \delta \lambda^M D + \bar{\beta}$. On the other hand, if $p \in \II_3$, then we obtain $C_p(f) \geq \delta \lambda^M D + \bar{\beta}$ from~\eqref{eq:minimize-for-active-set4}. 
			
			Since $\mymathbf{1}^\top f =D$, the same derivation as in \eqref{eq:derivation-fstar} then gives
			\begin{equation} \label{eq:ineq-objective} 
				f^\top A f + f^\top \beta \geq D(\delta \lambda^M D + \bar{\beta}D).
			\end{equation}
			We see that $D(\delta \lambda^M D + \bar{\beta})$ is a lower bound on the objective function value of \eqref{eq:minimize-for-active-set}, and $f'$ achieves this lower bound. Since $f'$ is also feasible for this minimization problem, it follows that it is an optimal solution of \eqref{eq:minimize-for-active-set}.
			
			From the above we draw the conclusion that any optimizer $f^*$ of \eqref{eq:minimize-for-active-set} satisfies $(f^*)^\top A f^* + (f^*)^\top \beta = D(\delta \lambda^M D + \bar{\beta})$. As shown in the derivation of~\eqref{eq:ineq-objective}, any feasible $f$ and path $p$ satisfying $f_p \neq 0$ satisfy $C_p(f) \geq \delta \lambda^M D + \bar{\beta}$. Consequently we have the following for any optimizer $f^*$ of \eqref{eq:minimize-for-active-set}:
			\begin{equation} \label{eq:optimizer-equal-cost}
				C_p(f^*) = \delta \lambda^M D + \bar{\beta} \quad \text{for all } p \text{ such that } f^*_p \neq 0.
			\end{equation}
			The next part of the proof is to establish that for any optimizer $f^*$ of \eqref{eq:minimize-for-active-set}, there exists $D^+ > 0$ such that the following holds:
			\begin{equation*}
				f^{D^+} := f^* + (D^+ - D)f^\delta \in \WW_{D^+}.
			\end{equation*}

			\emph{Step 4: $f^{D^+}$ is a WE:} We start by noting that if $f^*_p < 0$, then the constraint~\eqref{eq:minimize-for-active-set5} along with the definition of the set $\II_1$ imply $f^\delta_p > 0$. Thus, there exists large enough $D^+$ such that $f^{D^+} \geq 0$, and therefore, $f^{D^+} \in \FF_{D^+}$. Next, a similar derivation as in~\eqref{eq:cost-at-f-derivation} gives us
			\begin{equation} \label{eq:fDplus-derivation}
				C(f^{D^+}) = C(f^*) + (D^+ - D)Af^\delta.
			\end{equation}
			Now choose $p$ such that $f^{D^+}_p > 0$. This implies that either $f^\delta_p > 0$ or $f^*_p > 0$. For the first case, $f^\delta_p > 0$, we have $p \in \II_1$ and the constraint~\eqref{eq:minimize-for-active-set5} gives us $C_p(f^*) = \delta \lambda^M D + \bar{\beta}$, and we have already shown in Step 1 that in this case $A_pf^\delta = \delta \lambda^M$. Thus, we get
			\begin{equation} \label{eq:fDplus-goal}
				C_p(f^{D^+}) = \delta \lambda^M D^+ + \bar{\beta}.
			\end{equation}
			Now consider the second case $f^*_p > 0$. Using \eqref{eq:optimizer-equal-cost}, we deduce $C_p(f^*) = \delta \lambda^M D + \bar{\beta}$. Furthermore, constraint~\eqref{eq:minimize-for-active-set2} implies that $p \in \II_1 \cup \II_3$. From Step 1, we get that for $p \in \II_1 \cup \II_3$, the expression $\delta C^M_p = \delta \lambda^M$ holds. In combination with \eqref{eq:fDplus-derivation} this shows that \eqref{eq:fDplus-goal} holds.
			
			To establish that $f^{D^+}$ is a WE, all that remains to be shown is that ${C_p(f^{D^+}) \geq \delta \lambda^M D^+ + \bar{\beta}}$ whenever $f^{D^+}_p = 0$. Therefore, let $p$ be a path such that $f^{D^+}_p = 0$. This can occur when $f^*_p < 0$ and $f^\delta_p > 0$, however, in this case the previous arguments already show that \eqref{eq:fDplus-goal} holds. The only other way in which $f^{D^+}_p = 0$ is when $f^*_p = 0$ and $f^\delta_p = 0$. We split this scenario in two cases. First we consider $\delta C_p^M > \delta \lambda^M$. In this case it follows from $\delta C^M = A f^\delta$ in combination with \eqref{eq:fDplus-derivation} that for large enough $D^+$ we get
			\begin{equation} \label{eq:fDplus-goal2}
				C_p(f^{D^+}) \geq \delta \lambda^M D^+ + \bar{\beta}.
			\end{equation}
			The second case is ${\delta C_p^M = \delta \lambda^M}$ (In Step 3 we already argued that $\delta C_p^M = \delta \lambda^M$ for all $p \in \II_1$ and from the defintion of $\II_2$ and $\II_3$ it follows that $\delta C_p^M < \delta \lambda^M$ is not possible.) For this case \eqref{eq:minimize-for-active-set4} gives $C_p(f^*) \geq \delta \lambda^M D + \bar{\beta}$. Once again, using $\delta C^M = A f^\delta$ in combination with \eqref{eq:fDplus-derivation} we find that \eqref{eq:fDplus-goal2} holds. In conclusion, for large enough $D^+$ we have $f^{D^+} \geq 0$ and $f^{D^+}_p = 0$ implies \eqref{eq:fDplus-goal2}, and $f^{D^+}_p > 0$ implies  \eqref{eq:fDplus-goal}. In other words, as long as $D^+$ is large enough, $f^{D^+}$ is a WE. It follows that we can pick $D^+$ such that $D^+ \in (D_M,\infty)$ and $f^{D^+} \in \WW_{D^+}$.
			
			To finish the proof we now have $\Ja_{M} = \aset_{D^+}$. Since the cost under WE is unique, it follows that $p \in \Ja_{M}$ if and only if it has minimal cost among all paths for the flow $f^{D^+}$. In Step 4, we have shown that the paths with minimal cost are exactly those for which $\delta C^M_p = \delta \lambda^M$ and $C_p(f^*) = \delta \lambda^M D + \bar{\beta}$ hold and therefore, this concludes the proof.
		\end{proof}
	}

	Now that we can derive $\Ja_{M}$ using the above result, we finish this section by showing that we can also find $D_M$, as well as an associated WE $f^{D_M}$.
	\begin{corollary}\longthmtitle{Obtaining $D_M$}\label{cor:obtainDM}
		Let $(\PP,\CC)$ be given, and consider the following minimization problem:
		\begin{equation*}
			\begin{aligned}
				\textnormal{ minimize } \quad	& \mymathbf{1}^\top f	\\
				\textnormal{subject to } \quad  &C_p(f) \leq C_r(f) \quad &\text{ for all } p \in \Ja_{M} \text{ and } r \in \PP,	\\
				&f_r = 0 \quad &\text{ for all } r \in (\Ja_{M})^c,	\\
				&f \geq 0.
			\end{aligned}
		\end{equation*}
		For any solution $f^*$ of the above we have $\mymathbf{1}^\top f^* = D_M$ and $f^* \in \WW_{D_M}$.
	\end{corollary}
	\ifinclude{
		\begin{proof}
			Since any WE for a demand in $(D_M,\infty)$ satisfies the constraints, we see that the feasible set is non-empty. Also note that due to the affine and non-strict nature of the constraints, this implies that any WE in $\WW_{D_M}$ is also feasible. 
			It is then easy to prove that any solution $f^*$ to the given minimization problem must be a WE, and also that any flow of the form $f^* + \epsilon f^\delta$, where $f^\delta \in \Gamma_{D_M} \cap \FF_{1}$ and $\epsilon > 0$, is a WE as well. This shows that $f^*$ is a WE in the interval $[D_M, \infty)$, and it follows that $\mymathbf{1}^\top f^* = D_M$, which implies $f^* \in f^{D_M}$. 
		\end{proof}
	}

	\section{Braess's paradox for changing demand} \label{sec:BP-results}
	\begin{figure}
		\centering
		\includegraphics[width = 0.6\textwidth]{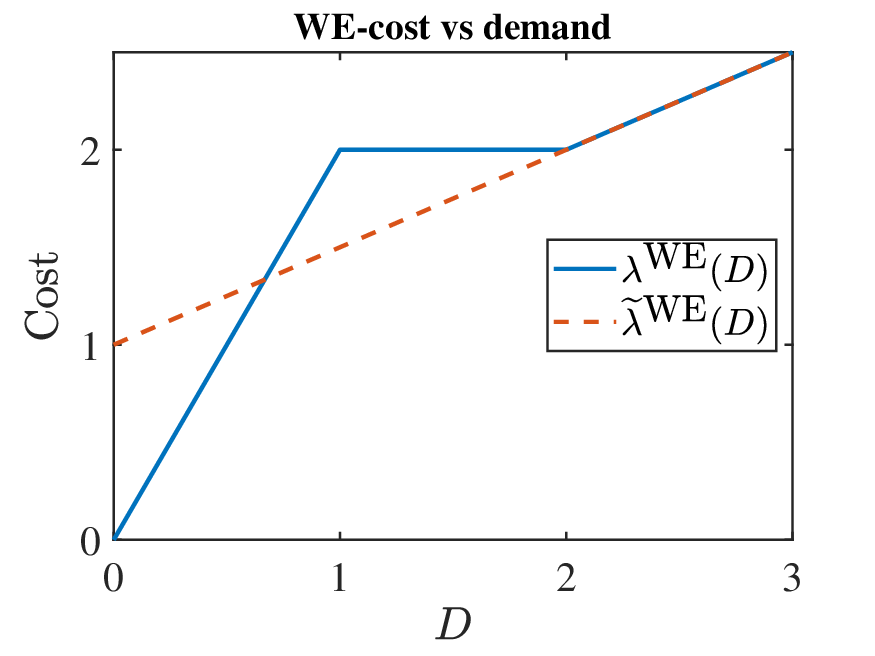}
		\caption{An illustration of BP for a routing game over the Wheatstone network, Figure~\ref{fig:wheatstone}, with and without edge $e_5$ present and with costs defined by \eqref{eq:Wheatstone-simple-path-cost}.}
		\label{fig:wheatstone-cost-bp}
	\end{figure}

	A well-known phenomenon in the context of traffic networks is \textit{Braess's paradox} (BP). A routing game is said to be subject to this paradox when removal of a path from the network leads to lower cost for all participants. The classical example is the one discussed in Example~\ref{ex:secIV-1}, considering the Wheatstone network depicted in Figure~\ref{fig:wheatstone}, where the cost functions are given by \eqref{eq:Wheatstone-simple-edge-costs}. Consequently, path-cost functions are given by \eqref{eq:Wheatstone-simple-path-cost}, and we have an explicit expression for the WE in \eqref{eq:Wheatstone-simple-WE}. Using these expressions, we find that $\lmWE(1) = 2$. Alternatively, we can look at the same network, but with edge $e_5$ removed. The cost functions remain the same, but the path $p_3$ is no longer accessible. In this case the WE is given by $\tillf^D = (\frac{1}{2}, \enskip \frac{1}{2})^\top$, and we have $\tillmWE(1) = 1.5$. Notice that for this level of demand, $D = 1$, removing $e_5$ has improved the situation for all participants.
	
	A natural question now is ``What happens for other levels of demand?''. For the Wheatstone network, the answer to this question is known to some extent. It has been shown that for this network, using affine cost functions, the associated routing game is subject to Braess's paradox on at most a bounded set of demands; i.e., there is an upper bound on the demand above which the routing game is no longer subject to the paradox \cite{EIP-SLP:97, VZ-EA:15}. Still, it is useful to study the given example in a bit more detail to see what happens for different levels of demand. We will therefore analyze this example further. We give explicit expressions for $\lmWE(D)$ and $\tillmWE(D)$, and show that there exists demand where the presence of path $p_3$ is strictly beneficial. 
	\begin{example} \longthmtitle{Evolution of BP} \label{ex:secV}
		\rm{
			\begin{enumerate}[wide= 0pt,label=(\alph*), ref=\ref{ex:secV}\alph*]
				\item \label{ex:secV-1} For the Wheatstone network (Example~\ref{ex:secIV-1} and Figure~\ref{fig:wheatstone})with the edge $e_5$ present we have
				\begin{equation*}
					\lmWE(D) = \begin{cases}
						2D \quad& \text{if } 0 \leq D \leq 1,	\\
						2 \quad & \text{if } 1 \leq D \leq 2,	\\
						\frac{D}{2} + 1 \quad & \text{if } 2 \leq D.
					\end{cases}
				\end{equation*}
				For the network without $e_5$ present, the WE for any $D \geq 0$ is given by
				\begin{equation*}
					\tillf^D = \left(\begin{array}{cc}
						\frac{D}{2},	&\frac{D}{2}
					\end{array}\right)^\top,
				\end{equation*}
				which gives us
				\begin{equation*}
					\tillmWE(D) = \frac{D}{2} + 1.
				\end{equation*}
				The situation is illustrated in Figure~\ref{fig:wheatstone-cost-bp}.
				We see that the presence of path $p_3$ is beneficial when ${0 \leq D < \frac{2}{3}}$, detrimental when $\frac{2}{3} < D < 2$, and neutral when $D \in \{\frac{2}{3}\} \cup [2,\infty)$. Clearly, the game being subject to Braess's paradox at one level of demand does not necessarily imply that adding the responsible path is detrimental overall, and the situation warrants further investigation. The aim of this section is therefore to study how the influence of a set of paths causing BP changes as the demand varies.
				\begin{figure}
					\centering
					\includegraphics[width = 0.6\textwidth]{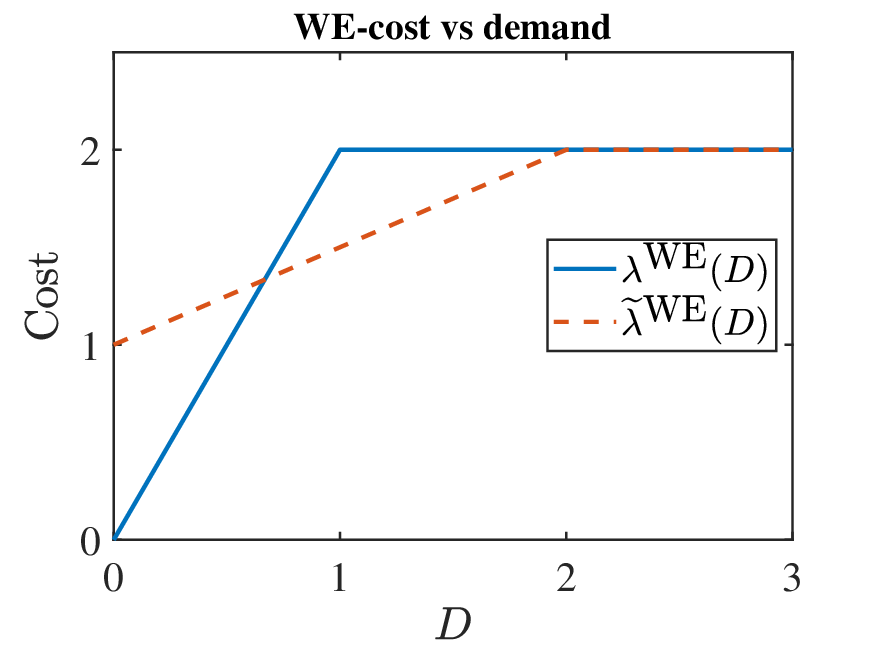}
					\caption{An illustration of BP for a routing game over the network in Figure~\ref{fig:three-node-four-edge}, with and without the path $(e_1,e_4)$ present, and with costs defined by \eqref{eq:Wheatstone-double-cost}.}
					\label{fig:wheatstone-merged-bp}
				\end{figure}
				\item \label{ex:secV-2}
				We note that classically, Braess's paradox refers to the situation where removal of an \emph{edge} (or a set of edges) leads to a lesser cost for all participants, but here we will look at a slightly more generalized form of the paradox, where removal of a \emph{path}, or set of paths, from the network leads to a lesser cost for all participants. The previous example shows an instant where these two cases are the same. However, it is possible that a a routing game is subject to an instance of BP that only emerges as the result of removing a (set  of) path(s), but not as the result of removing a (set of) edge(s). To show this, we revisit Example~\ref{ex:secIV-2}. That is, we consider a routing game over the network in Figure~\ref{fig:three-node-four-edge}, with path-cost given by \eqref{eq:Wheatstone-double-cost}. Figure~\ref{fig:wheatstone-merged-bp} shows the comparison between the cost of this game  and the modified game where path $p_4$ has been removed. Note that the case is essentially identical to that of Example~\ref{ex:secIV-1}. Observe now that the game is subject to BP on the interval $D \in (\frac{2}{3},1)$, but in this case there is no edge which is responsible for the inefficiency. Instead it is the presence of the path $(e_3,e_2)$ that causes BP. \oprocend
			\end{enumerate}
		}
	\end{example}
	\subsection{Preliminaries for analyzing BP}
	To ease the exposition of the upcoming analysis, we introduce definitions and useful facts related to the concepts of \emph{necessary sets}, \emph{modified games} and \emph{directions of decrease}.
	\subsubsection{Necessary sets}
	The set of all subsets of $\PP$ that can not be removed from the game without changing the WE at the demand $D$ is called the \emph{necessary set} at $D$. Formally, we have the following definition.
	\begin{definition}\longthmtitle{Necessary sets}\label{def:necessary}
		For a given $(\PP,\CC)$ we say that a set $\SSnec \subseteq \PP$ is \emph{necessary} at demand $D$ if 
		\begin{equation*}
			f^D_{\SSnec} \neq 0 \text{ for all } f^D \in \WW_{D}. 
		\end{equation*}
		That is, for every WE $f^D$, there exists at least one path in the set $\SSnec$ that takes nonzero flow. We use $\NN_D$ to denote the set of all necessary sets at $D$; that is,
		\begin{equation*}
			\NN_D := \setdef{\SSnec \subseteq \PP}{f^D_{\SSnec} \neq 0 \text{ for all } f^D \in \WW_{D}}.  
		\end{equation*}
		When $\SSS \notin \NN_{D}$ we say that $\SSS$ is \emph{unnecessary} at demand $D$. \oprocend
	\end{definition}
	Although at first glance the necessary set may seem closely related to the used set $\uset_{D}$, we note that unlike $\uset_{D}$ and $\aset_{D}$, the necessary set is not guaranteed to stay constant in between the breakpoints in $\DD$. For instance, in Example~\ref{ex:secIV-2} we have  $\DD = \{0,1,\infty\}$. However, we can deduce from the expression given for $\WW_{D}$ in \eqref{eq:three-node-four-edge-WE} that the singleton set $\{p_3\}$ is necessary at all demands in $(0,2)$ while $\{p_3\}$ is unnecessary for $D \geq 2$.
	\subsubsection{Modified games}
	To better accomodate our path-based notion of BP, we will introduce the concept of a \emph{modified game}. Given a routing game $(\PP,\CC)$, a modified game is constructed by removing a set of paths $\SSrem \subset \PP$ from the game, which is achieved by forcing the flow on paths in $\SSrem$ to be zero, while keeping all cost functions the same. We use $(\tillPP,\CC)$ to denote such a modified game, where $\tillPP = \PP \setminus \SSrem$. We also use the following notation:
	\begin{align*}
		\tillPP &:= \PP \setminus \SSrem,	
		\\
		\tillFF_D &:= \setdef{\tillf \in \realnonnegative^{n}}{\sum_{p \in \tillPP}\tillf_p = D, \enskip \tillf_{\SSrem} = 0},	
		\\
		\tillHH_D &:= \setdef{\tillf \in \real^{n}}{\sum_{p \in \tillPP}\tillf_p = D, \enskip \tillf_{\SSrem} = 0}.
	\end{align*}
	For notational convenience, the dimension of the flows $\tillf \in \tillFF_{D}$ of the modified game is kept equal to $n = \abs{\PP}$.
	For a modified game, we say that $\tillf^D$ is a WE when $\tillf^D \in \tillFF_D$ and for all $p \in \tillPP$ such that $\tillf_p > 0$ we have
	\begin{equation*}
		C_p(\tillf^D) \leq C_r(\tillf^D) \quad \text{for all } r \in \tillPP.
	\end{equation*}
	Analogously, we use the notation $\widetilde{(\cdot)}$ for other items related to the modified game. We write
	\begin{equation}\label{eq:items-modified-game}
		\begin{aligned}
			\tillWW_{D}		&:= \SOL(\tillFF_D,C),			\\
			\tillmvec(D)	&:= C(\tillf^D) \text{ for any } \tillf^D \in \tillWW_{D},									\\
			\tilaset_D		&:= \setdef{p  \in  \tillPP }{ \tillmvec_p(D)  \leq \tillmvec_r(D) \text{ for all } r \in \tillPP},								\\
			\tiluset_D		&:= \setdef{p \in \tillPP}{ \exists \tillf^D \in \tillWW_D \text{ such that } \tillf^D_p > 0},	\\
			\tillmWE(D) &:= \tillmvec_p(D) \text{ for any } p \in \tilaset_{D},	\\
			\tillMM_D		&:= \setdef{\tillf^{\delta}  \in  \tillHH_1 }{\tillf_{\tilaset_D \setminus \tiluset_D}^{\delta} \geq 0, \enskip \tillf_{(\tilaset_D)^c}^{\delta} = 0},			\\
			\tillGamma_D	&:= \SOL(\tillMM_D,A).			
		\end{aligned}
	\end{equation}
	Similarly, we use $\tillDD$ to denote the set of breakpoints of a modified game and use $\tillD_i$ to denote the $i$-th breakpoint of $\tillDD$. We write $\tillM$ for the index of the greatest finite valued breakpoint $\tillD_{\tillM}$ in $\tillDD$ and use $\delta \widetilde{\lambda}^i$ and $\delta \widetilde{C}^i$ to denote the directions in which respectively $\tillmWE$ and $\tillmvec$ evolve on the interval between $\tillD_i$ and $\tillD_{i+1}$.
	
	When we need to discuss multiple modified games simultanuously, as is sometimes the case in a proof, we use a similar notation for concepts related to these modified games, replacing $\widetilde{(\cdot)}$ with $\widecheck{(\cdot)}$ or with $(\cdot)'$, $(\cdot)''$, or $(\cdot)'''$.  (e.g. we use $\widecheck{\PP}$, and $\PP'$, $\PP''$, and $\PP'''$ to denote the related sets of paths, and similarly for the feasible set, WE-cost, etc.)
	
	We note that a modified game is technically not a routing game as presented in Section~\ref{sec:model}. The reason for this is that the feasible set has additional restrictions and the definition of WE for a modified game only takes into account a subset of the paths in $\PP$. Of course, instead of setting $\tillf_{\SSrem} = 0$, we could simply drop this set of paths from consideration, and let $\tillf \in \realnonnegative^{\abs{\tillPP}}$. In this case we could define $\tillFF_D$ without imposing additional restrictions. However, removing a set of paths from a routing game can result in a situation which can no longer be represented by a graph. That is, in that case there does not exists a graph, an associated origin-destination pair and a set of cost functions such that the resulting routing game over all paths from the origin to the destination has the same cost function as that of the modified game. However, the results and proofs in this text are in no way dependent on the existence of such a graph representation for the given set of paths and associated cost functions. Therefore all results that we established for routing games also hold for modified games. We specifically highlight that $\tillWW_{D}$ and $\tillGamma_D$, as defined in~\eqref{eq:items-modified-game}, are therefore indeed the set of WE of the modified game and the set of directions of increase of the modified game respectively, as the notation suggests.
	\subsubsection{Directions of decrease}
	The results in Section~\ref{sec:variation-of-WE} consider the directions of \emph{increase} in which the WE moves as the demand increases, but for some of the upcoming statements it is useful to instead consider the set of \emph{directions of decrease}:
	\begin{definition} \longthmtitle{Set of directions of decrease} \label{def:directions-of-decrease}
		Let $(\PP,\CC,D)$ be given. The set of \emph{directions of decrease} $\Gamma^-_D$
		is the set of all directions $f^\delta \in \HH_{-1}$ in which the flow can be decreased, starting from some flow in $\WW_{D}$, such that the new flow is a WE as long as the decrease is small enough. That is,
		\begin{align*}
			\Gamma^-_{D} := \setdef{f^\delta \in \HH_{-1}
			}{
				\exists f^D \in \WW_{D}, \enskip \bar{\epsilon} > 0 \text{ such that } f^{D} + \epsilon f^\delta \in \WW_{D - \epsilon} \enskip \forall \epsilon \in [0,\bar{\epsilon}]}.
		\end{align*}
	\end{definition}
	Similar to before, we define the set $\MM_{D}^-$ of feasible descent directions as
	\begin{equation*} 
		\MM_D^-  :=  \setdef{f^{\delta}  \in  \HH_{-1} }{f_{\aset_D \setminus \uset_D}^{\delta} \geq 0, \enskip  f_{(\aset_D)^c}^{\delta} = 0}.
	\end{equation*}
	The arguments made in Section~\ref{sec:variation-of-WE} can then be repeated to obtain the following modified version of Proposition~\ref{prop:characterize-directions-of-increase}:
	\begin{lemma}\longthmtitle{Directions of decrease as solutions to a VI} \label{charaterize-gamma-minus}
		Let $(\PP,\CC,D)$ be given. Then,
		\begin{equation*}
			{\Gamma_D^- = \SOL(\MM_D^-,A)}.
		\end{equation*}
	\end{lemma}
	\subsubsection{Properties of $V$}
	We recall from Section~\ref{sec:preliminaries} the function $V$, and similarly define the function $\widetilde{V}$ for a modified game as
	\begin{equation} \label{eq:V-for-modified}
		\begin{aligned}
			V(D) = &\min_{f \in \FF_D} \sum_{e_k \in \EE} \int_0^{f_{e_k}}C_{e_k}(z)dz,
			\\
			\tillV(D) := &\min_{\tillf \in \tillFF_D} \sum_{e_k \in \EE} \int_0^{\tillf_{e_k}}C_{e_k}(z)dz.
		\end{aligned}
	\end{equation}
	Using the inclusion $\tillFF_D \subseteq \FF_{D}$ in combination with Proposition~\ref{prop:beckmann} then yields us the following observations on modified games and necessary sets:
	\begin{lemma} \longthmtitle{Relations between the original and the modified game} \label{lem:dominated-V}
		Let $(\PP,\CC,D)$ and $\SSrem \subset \PP$ be given, and let $V$ and $\tillV$ be given by \eqref{eq:V-for-modified}. The following then hold:
		\begin{itemize}
			\item $V(D)	\leq \tillV(D)$,
			\item $V(D)	= \tillV(D)$ if and only if $\SSrem \notin \NN_D$,	
			\item if $\SSrem \notin \NN_D$, then $\tillf^D \in \tillWW_D$ if and only if ${\tillf^D_{\SSrem} = 0}$ and ${\tillf^D \in \WW_{D}}$. As a consequence, we then have ${\lmWE(D) = \tillmWE(D)}$.
		\end{itemize}
	\end{lemma}
	\begin{figure}
		\centering
		\includegraphics[width = 0.6\textwidth]{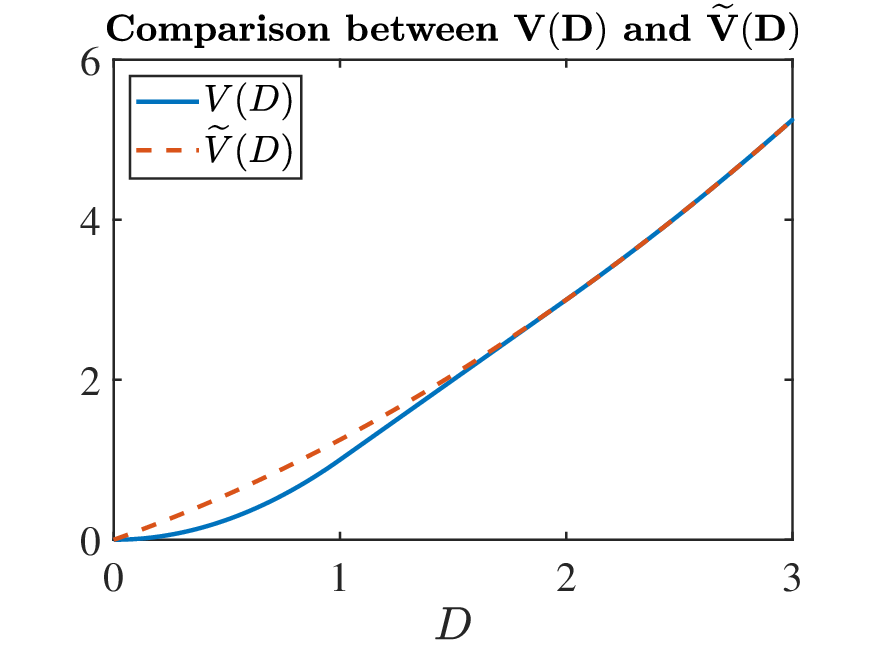}
		\caption{A comparison between $V$ and $\tillV$ for the Wheatstone network in Figure~\ref{fig:wheatstone}, with and without the path $p_3 = (e_1,e_5,e_4)$ present and with costs defined by \eqref{eq:Wheatstone-simple-path-cost}.}
		\label{fig:VD-comparison}
	\end{figure}
	The above statements are illustrated in Figure~\ref{fig:VD-comparison}, which shows the functions $V$ and $\tillV$ for the routing games considered in Example~\ref{ex:secV-1}. Here, the path removed from $\PP$ to form the modified game is $\SSrem = \{p_3\}$. Notice that for demand in the interval $(0,2)$, the singleton set $\SSrem = \{p_3\}$ is a necessary set and as expected, we have  $V(D) <  \tillV(D)$ for all $D \in (0,2)$. Moreover, exactly at the point $D = 2$, the set $\SSrem$ is no longer necessary and we have $V(2) = \tillV(2)$.
	\subsection{Evolution of WE-costs}
	Our upcoming discussion on BP depends for a large part on analysis of the WE-cost $\lmWE$. We show how knowledge on the evolution of this cost can be leveraged to detect when a game is subject to BP. In particular, we look at how the slope of $\lmWE$ changes at the breakpoints in $\DD$. Since $\lmWE$ is potentially non-differentiable at these points, it is useful to consider both the left-hand and right-hand side derivatives:
	\begin{equation*}
		\delta \lambda^+(T) : =\frac{\partial^+}{\partial D} \lmWE(D)\Big|_{D=T},	\quad \delta \lambda^-(T) := \frac{\partial^-}{\partial D} \lmWE(D)\Big|_{D=T}.	
	\end{equation*}
	Note that whenever $T \notin \DD$, we have $\delta \lambda^+(T) = \delta \lambda^-(\bar{D})$. Using Proposition~\ref{prop:properties-solution-VI} with $\MM = \MM_{D}$ and $\MM = \MM^-_{D}$, we then have the following observation:
	\begin{corollary} \longthmtitle{Relation between $\SOL(\MM_{D},A)$ and slope of $\lmWE$} \label{cor:delta-lambda-representation}
		Let $(\PP,\CC,D)$ be given. We have
		\begin{align*}
			\delmplus(D) &= \min_{r \in \aset_{D}} A_r f^\delta \quad \hspace{10 pt} \text{for all } f^\delta \in \SOL(\MM_{D},A),
			\\
			\delmminus(D) &= -\min_{r \in \aset_{D}} A_r f^\delta \quad \text{for all } f^\delta \in \SOL(\MM_{D}^-,A).
		\end{align*}
	\end{corollary}

	We see that sets of the form $\SOL(\MM_D,A)$ are important when studying $\delmplus(D)$ and $\delmminus(D)$. Our next result therefore concerns sets of this form and derives properties that will used in studying the evolution of cost. 
	\begin{lemma} \label{lem:dominated-feasible-set} \longthmtitle{Properties of $\SOL(\MM,A)$}
		Let $\PP$, $\CC \subset \classfn$, and sets $\QQ,\RR,\tillQQ,\tillRR \subseteq \PP$ satisfying $\RR \subseteq \QQ$ and $\tillRR \subseteq \tillQQ$ be given. In addition, let
		\begin{align*}
			\MM &:= \setdef{f^\delta \in \HH_1}{f^\delta_{\QQ \setminus \RR} \geq 0, \quad f^\delta_{\QQ^c} = 0},	
			\\
			\tillMM &:= \setdef{f^\delta \in \HH_1}{f^\delta_{\tillQQ \setminus \tillRR} \geq 0, \quad f^\delta_{\tillQQ^c} = 0},	\\
			\MM^- &:= \setdef{f^\delta \in \HH_{-1}}{f^\delta_{\QQ \setminus \RR} \geq 0, \quad f^\delta_{\QQ^c} = 0},	\\
			\tillMM^- &:= \setdef{f^\delta \in \HH_{-1}}{f^\delta_{\tillQQ \setminus \tillRR} \geq 0, \quad f^\delta_{\tillQQ^c} = 0}.
		\end{align*}
		For given $f^\delta \in \SOL(\MM,A)$, $\tillf^\delta \in \SOL(\tillMM,A)$, ${f^{\delta-} \in \SOL(\MM^-,A)}$ and ${\tillf^{\delta-} \in \SOL(\tillMM^-,A)}$,
		if $\tillMM \subseteq \MM$ (respectively, $\tillMM^- \subseteq \MM^-)$ then we have
		\begin{align*}
			\min_{r \in \tillQQ} A_r \tillf^\delta &\geq 	\min_{r \in \QQ} A_r f^\delta,	\\
			\text{(respectively}, \quad \min_{r \in \tillQQ} A_r \tillf^{\delta-} &\leq 	\min_{r \in \QQ} A_r f^{\delta-}).
		\end{align*}
	\end{lemma}
	\ifinclude{
		\begin{proof}
			We first prove the result for the case $\tilde{\MM} \subseteq \MM$. For $f^\delta \in \SOL(\MM,A)$, let $p \in \PP$ be a path satisfying $f^\delta_p \neq 0$. Then, either $f^\delta_p > 0$ or $p \in \RR$. For both these cases, from Proposition~\ref{prop:properties-solution-VI}, we have $A_p f^\delta = \min_{r \in \QQ}A_r f^\delta$. Using this and the fact that ${f^\delta \in \HH_1}$ and $f^\delta_{\QQ^c} = 0$, we obtain ${(f^\delta)^\top A f^\delta = \min_{r \in \QQ}A_r f^\delta}$. Similarly we find $(\tillf^\delta)^\top A \tillf^\delta = \min_{r \in \tillQQ}A_r \tillf^\delta$.
			Assume for the sake of contradiction that
			\begin{equation}\label{eq:dominated-feasible-set-contradiction-1}
				\min_{r \in \tillQQ}A_r \tillf^\delta < \min_{r \in \QQ} A_r f^\delta.
			\end{equation}
			Since $f^\delta \in \SOL(\MM,A)$ and $\tillf^\delta \in \tillMM \subseteq \MM$, we have 
			\begin{equation}\label{eq:VI-delta-deltatilde}
				(f^\delta)^\top A (\tillf^\delta - f^\delta) \geq 0,
			\end{equation}
			We now have the following derivation 
			\begin{align*}
				(\tillf^\delta)^\top A \tillf^\delta = \min_{r \in \tillQQ}A_r \tillf^\delta	 < \min_{r \in \QQ}A_r f^\delta	& = (f^\delta)^\top A f^\delta	
				\\
				&\leq (f^{\delta})^\top A \tillf^\delta,
			\end{align*}
			where the first inequality is due to~\eqref{eq:dominated-feasible-set-contradiction-1} and the second inequality follows from~\eqref{eq:VI-delta-deltatilde}.
			The above implies
			\begin{equation} \label{eq:dominated-feasible-set-contradiction-2}
				(\tillf^\delta)^\top A(f^\delta - \tillf^\delta) > 	0.
			\end{equation}
			On the other hand, since $A$ is positive semi-definite
			\begin{equation} \label{eq:semidefinite-delta-deltatilde}
				(\tillf^{\delta} - f^{\delta})^\top A(f^{\delta} - \tillf^{\delta}) \leq 0.
			\end{equation}
			Expanding the left-hand side of this inequality and using~\eqref{eq:VI-delta-deltatilde}, we obtain ${(\tillf^{\delta})^\top A(f^{\delta} - \tillf^{\delta}) \leq 0}$, which contradicts \eqref{eq:dominated-feasible-set-contradiction-2}. Therefore the premise is false and the proof concludes. For the case $\tillMM^- \subseteq \MM^-$, the result can be proven using the same arguments, where the direction of the inequality reverses since $f^{\delta-}$ is in $\HH_{-1}$ rather than $\HH_1$.
		\end{proof}
	}
	Using the above result, we obtain the following:
	\begin{lemma} \longthmtitle{Slope of $\lmWE$ for constant used and active set} \label{lem:simpe-cost-evolution}
		Let $(\PP,\CC)$, $D_{i} \in \DD$ and $i \in [M]$ be given. Then,
		\begin{align*}
			\Ju_{i-1}  &=  \uset_{D_{i}} \Rightarrow \delta \lambda^{i-1} > \delta \lambda^i, \qquad \Ju_{i}  =  \uset_{D_{i}} \Rightarrow \delta \lambda^{i-1} < \delta \lambda^i, 	\\
			\Ja_{i-1}  &=  \aset_{D_i} \Rightarrow \delta \lambda^{i-1} < \delta \lambda^{i}, \qquad \Ja_{i}  = \aset_{D_i} \Rightarrow \delta \lambda^{i-1} > \delta \lambda^{i}.
		\end{align*}
	\end{lemma}
	\ifinclude{
		\begin{proof}
			We start with the claim for the case $\Ju_{i-1} = \uset_{D_{i}}$. From Corollary~\ref{cor:evolution-lmWE} we have $\delta \lambda^{i-1} = \min_{r \in  \Ja_{i-1}} \delta C^{i-1}_r$. Now let $D \in (D_{i-1},D_i)$. From Proposition~\ref{prop:evolution-lmvec} we then have $\delta C^{i-1}_r = A f^\delta$ for any ${f^\delta \in \Gamma_{D}}$. Furthermore, Proposition~\ref{prop:characterize-directions-of-increase} gives ${\Gamma_{D} = \SOL(\MM_{D},A)}$ and it follows from the definition of  $\MM_{D}$ and Corollary \ref{cor:interval-of-active-set} that
			\begin{equation*}
				\MM_{D} = \setdef{f^\delta \in \HH_{1}}{f^\delta_{\Ja_{i-1} \setminus \Ju_{i-1}} \geq 0, \enskip f^\delta_{(\Ja_{i-1})^c} = 0}.
			\end{equation*}	
			In summary, we found that $\delta \lambda^{i-1} = \min_{r \in  \Ja_{i-1}} A_rf^\delta$, where ${f^\delta \in \SOL(\MM_{D},A)}$, and $\MM_{D}$ is given by the above equality.
			
			We can obtain a similar result for $\delta \lambda^i$. From Proposition~\ref{prop:evolution-lmvec} we have ${\lmvec(T) = \lmvec(D_i) + \delta C^{i}}$ for any $T \in [D_i,D_{i+1}]$, and from Corollary~\ref{cor:evolution-lmWE} we have ${\lmWE(T) = \lmWE(D_i) + (T - D_i)\delta \lambda^i}$. Since the WE $\lmWE(T)$ is the minimum of all $\lmvec(T)$ it follows that it is equal to the costs of those paths that have minimum cost under WE at demand $D_i$, and that keep minimal cost as $\lmvec$ moves in the direction of $\delta C^{i}$. In other words, $\delta \lambda^i = \min_{r \in  \aset_{D_{i}}} \delta C^i$.
			Proposition~\ref{prop:evolution-lmvec} also gives $\delta C^i = A f^\delta$ for any $f^\delta \in \Gamma_{D_i}$, and Proposition~\ref{prop:characterize-directions-of-increase} gives $\Gamma_{D_i} = \SOL(\MM_{D_i},A)$. In summary, we found that $\delta \lambda^{i} = \min_{r \in  \aset_{D_{i}}} A_r f^\delta$, where ${f^\delta \in \SOL(\MM_{D_i},A)}$, and $\MM_{D_i}$ is given by
			\begin{equation*}
				\MM_{D_i} = \setdef{f^\delta \in \HH_{1}}{f^\delta_{\aset_{D_{i}} \setminus \uset_{D_i}} \geq 0, \enskip f^\delta_{(\aset_{D_i})^c} = 0}.
			\end{equation*}

			By assumption we have $\Ju_{i-1} = \uset_{D_{i}}$, and from Lemma~\ref{lem:inclusion-active-and-used} we have $\Ja_{i-1} \subseteq \aset_{D_{i}}$. Therefore we have $\MM_{D} \subseteq \MM_{D_i}$. The rest of the proof follows in the same manner as the proof of Lemma~\ref{lem:dominated-feasible-set}, with $\MM = \MM_{D_{i}}$ and $\tillMM = \MM_{D}$.
			The only difference is that the inequalities \eqref{eq:dominated-feasible-set-contradiction-1} and \eqref{eq:dominated-feasible-set-contradiction-2} are non-strict while the inequality in \eqref{eq:semidefinite-delta-deltatilde} is strict, therefore preserving the contradiction. To see that \eqref{eq:semidefinite-delta-deltatilde} holds strictly, note that it holds with equality only if $Af^\delta = A \tillf^\delta$. However, in our case we have  $Af^\delta = \delta C^i$ and $A \tillf^\delta = \delta C^{i-1}$. Thus equality of \eqref{eq:semidefinite-delta-deltatilde} would imply $\delta C^i = \delta C^{i-1}$, which contradicts Proposition~\ref{prop:evolution-lmvec}.
			
			This completes the proof for the case $\Ju_{i-1} = \uset_{D_i}$. The claim for the case ${\Ja_{i-1} = \aset_{D_i}}$ follows by the same arguments, where we find $\MM_{D_i} \subseteq \MM_{D}$ instead of $\MM_{D} \subseteq \MM_{D_i}$. The other claims can be proven similarly, where we consider $\MM^-_{D}$ and $\MM^-_{D_i}$ instead of $\MM_{D}$ and $\MM_{D_i}$.
		\end{proof}
	}
	In light of Proposition~\ref{prop:characterize-directions-of-increase}, the above is an intuitive result. For instance, consider the first scenario when $\Ju_{i-1} = \uset_{D_{i}}$. Pick $D \in (D_{i-1},D_i)$. Then, we know from Lemma~\ref{le:gamma-constant} that $\Gamma_{D} \neq \Gamma_{D_i}$. For this to hold, using Proposition~\ref{prop:characterize-directions-of-increase}, we require $\MM_D \neq \MM_{D_i}$. Now focusing on the definition of the direction of feasibility given in~\eqref{eq:feasible-increase-direction}, we obtain $\Ja_{i-1} = \aset_D \neq \aset_{D_{i}}$. However, Lemma~\ref{lem:inclusion-active-and-used} shows that $\Ja_{i-1} \subseteq \aset_{D_{i}}$, and so, $\Ja_{i-1} \subset \aset_{D_{i}}$. Thus, we obtain $\MM_{D} \subset \MM_{D_i}$. This is an important inclusion which shows that the set in which the directions of increase must lie has grown strictly larger as ones moves from the interval $(D_{i-1},D_i)$ to the point $D_i$. In other words, there are more paths at $D_i$ which are feasible for carrying flow. Since the division of flow is a result of the traffic participants trying to minimize their own travel time, it is natural that an increase in options will decrease the ``rate'' at which the WE-cost grows as the demand increases. A similar reasoning can be deduced for the second case $\Ja_{i-1} = \aset_{D_{i}}$ of Lemma~\ref{lem:simpe-cost-evolution}, where
	there are less options for the evolution of the flow at $D_i$ than in the interval $(D_{i-1},D_i)$, and so, the ``rate'' at which the WE-cost grows as the demand increases becomes larger. We can also see this in Example~\ref{ex:secIV-1}, when comparing Figures~\ref{fig:wheatstone-WE}~and~\ref{fig:wheatstone-cost-bp}. At the demand $D_1 = 1$ (which is a breakpoint), paths $p_1$ and $p_2$ become active, but these paths are not in the active set in the interval $(0,1)$. On the other hand, the used set is same in the interval $(0,1)$ and at the point $D_1  = 1$. Thus, from the first implication of Lemma~\ref{lem:simpe-cost-evolution}, the slope of the cost decreases at $D_1 = 1$. This is depicted in Figure~\ref{fig:wheatstone-cost-bp}. Conversely, at demand $D_2 = 2$ path $p_3$ leaves the used set and it is therefore no longer possible to decrease the amount of flow on this path. At the same time, we have that the active set is same in interval $(1,2)$ and the point $D_2 = 2$. Hence, by the second implication of Lemma~\ref{lem:simpe-cost-evolution}, the slope of the cost increases, as observed in Figure~\ref{fig:wheatstone-cost-bp}. Note that this intuition is exactly the one that is defied by Braess's paradox, where we see that more options increases travel time.
	
	A final observation following Lemma~\ref{lem:simpe-cost-evolution} is related to Example~\ref{ex:break-in-active-set-vs-break-in-cost-2}, in which the slope of $\lmWE$ does not change at $D = 2$, despite this being a breakpoint. From Lemma~\ref{lem:simpe-cost-evolution} we see that this can only happen when both the active and used set change simultaneously at the considered breakpoint. That is, there needs to be both a path that loses all flow, and there must be a previously inactive path that becomes active. The slope of $\lmWE$ then remains constant when the effects on the evolution of $\lmWE$ of these changes in the used and active set are exactly opposite. 
	
	As a last step before we finally turn our attention to Braess's paradox, we look at the implications on the evolution of the WE-cost of a modified game when the set $\SSrem$ is \emph{not} necessary. We show that in this case the (left- and right-hand) derivative of $\tillmWE$ upper bounds the  (left- and right-hand) derivative of $\lmWE$.
	\begin{lemma} \longthmtitle{Necessary sets and the slope of $\lmWE$} \label{lem:dominated-cost-increase}
		Let $(\PP,\CC,D)$ and $\SSrem \subset \PP$ be given. If $\SSrem \notin \NN_D$ where $D >0$, then the following hold:
		\begin{align*}
			\tildelmplus(D) &\geq \delmplus(D),	
			\\
			\tildelmminus(D) &\geq \delmminus(D).
		\end{align*}
	\end{lemma}
	\ifinclude{
		\begin{proof}
			We start by proving the first inequality. From Corollary~\ref{cor:delta-lambda-representation} we obtain
			\begin{equation}\label{eq:delta-lam-plus}
				\delmplus(D) = \min_{r \in \aset_{D}} A_r f^\delta, \quad \tildelmplus(D) = \min_{r \in \tilaset_{D}} A_r \tillf^\delta,
			\end{equation}
			where $f^\delta \in \SOL(\MM_{D},A)$ and $\tillf^\delta \in \SOL(\tillMM_D,A)$. Here,
			\begin{equation*} 
				\tillMM_D = \setdef{\tillf^\delta \in \tillHH_1}{\tillf_{\tilaset_D \setminus \tiluset_D}^{\delta} \geq 0, \enskip  \tillf_{(\tilaset_D)^c}^{\delta} = 0}.
			\end{equation*}
			Now, since $\SSrem \notin \NN_D$ it follows from Lemma~\ref{lem:dominated-V} that $\tillf^D \in \tillWW_D$ if and only if $\tillf^D \in \WW_{D} \cap \tillFF_D$. Consequently ${\tiluset_{D} \subseteq \uset_{D} \cap (\SSrem)^c}$ and in addition we have $\tillmvec(D) = \lmvec(D)$, which implies ${\tilaset_{D} = \aset_{D} \cap (\SSrem)^c}$. These facts collectively imply ${\tillMM_{D} \subseteq \MM_{D}}$ and it follows from Lemma~\ref{lem:dominated-feasible-set} that 
			\begin{equation*}
				\min_{r \in \tilaset_{D}} A_r \tillf^\delta \geq \min_{r \in \aset_{D}} A_r f^\delta.
			\end{equation*}
			Thus, using~\eqref{eq:delta-lam-plus} in the above inequality, we obtain
			\begin{equation*}
				\tildelmplus (D) \geq
				\delmplus(D).
			\end{equation*}
			This establishes the first inequality. The second inequality follows by similar arguments, but considering sets of directions of decrease $\MM_{D}^-$ in Lemma~\ref{lem:dominated-feasible-set}.
		\end{proof}
	}

	\subsection{Conditions revealing Braess's paradox} \label{sec:BP-conditions}
	Now we are ready to give our first result concerning Braess's paradox. It states that if a set $\SSrem \subset \PP$ is unecessary at $D$, then either this set subjects the game to BP at some lower levels of demand, or it is unnecessary for all lower levels of demand.
	\begin{proposition} \longthmtitle{Sets not in $\NN_D$ are ``non-essential'' for lower demands} \label{prop:unnecesary-and-BP}
		Let $(\PP,\CC,D)$ and $\SSrem \subset \PP$ be given. If $\SSrem \notin \NN_D$ then exactly one of the following holds:
		\begin{itemize}
			\item $\SSrem \notin \NN_T$ for all $T \in [0,D]$.
			\item There exists $D^-,D^+$ satisfying $0 < D^- < D^+ \leq D$ and
			\begin{equation} \label{eq:dominated-lmWE}
				\lmWE(T) > \tillmWE(T) \quad \text{for all } T \in (D^-,D^+).
			\end{equation}
		\end{itemize}
	\end{proposition}
	\ifinclude{
		\begin{proof}
			First we recall from Proposition~\ref{prop:pw-affine} that $V$ and $\tillV$ are continuously differentiable, with ${\frac{\partial}{\partial T}V(T) = \lmWE(T)}$ and $\frac{\partial}{\partial T}\tillV(T) = \tillmWE(T)$. Furthermore, since ${\SSrem \notin \NN_D}$, from Lemma~\ref{lem:dominated-V} we have $V(D) = \tillV(D)$ and ${\lmWE(D) = \tillmWE(D)}$. Define $D^+ \leq D$ as the smallest value such that $\lmWE(T) = \tillmWE(T)$ for all $T \in [D^+,D]$.
			In other words, the derivatives of $V$ and $\tillV$ are equal on the interval $[D^+,D]$. Consequently, since $V(D) = \tillV(D)$, this implies $V(T) = \tillV(T)$ for all $T \in [D^+,D]$. By Lemma~\ref{lem:dominated-V} it follows that $\SSrem \notin \NN_T$ for all $T \in [D^+,D]$. If $D^+ = 0$ this gives us $\SSrem \notin \NN_T$ for all $T \in [0,D]$, which corresponds to the first scenario.
			
			Alternatively, if $D^+ > 0$, then it remains to be shown that there exists a $D^- \in (0,D^+)$ such that \eqref{eq:dominated-lmWE} holds. First, we note that using the fact $\SSrem \notin \NN_{D^+}$ in Lemma~\ref{lem:dominated-cost-increase} gives us $\tildelmminus(D^+) \geq \delmminus(D^+)$. That is, either $\tildelmminus(D^+) = \delmminus(D^+)$ or $\tildelmminus(D^+) > \delmminus(D^+)$. The former of these is not possible. To see this, recall that by definition of $D^+$, we have $\lmWE(D^+) = \tillmWE(D^+)$. Since $\lmWE$ and $\tillmWE$ are continuous, piece-wise affine functions with only finitely many points in which the functions are not differentiable, it follows that if $\tildelmminus(D^+) = \delmminus(D^+)$, then there exists some $\epsilon > 0$ such that $\lmWE(T) = \tillmWE(T)$ for all $T \in (D^+ - \epsilon,D^+]$. This however contradicts the definition of $D^+$ as the smallest value such that $\lmWE(T) = \tillmWE(T)$ for all $T \in [D^+,D]$. Therefore, $\tildelmminus(D^+) = \delmminus(D^+)$ is not possible and we have $\tildelmminus(D^+) > \delmminus(D^+)$. This then implies that there exists some $D^- > 0$ such that $D^- < D^+$ and \eqref{eq:dominated-lmWE} holds, completing the proof.
		\end{proof}
	}

	Proposition~\ref{prop:unnecesary-and-BP} gives us our first way of detecting sets of paths that are at least candidates for subjecting a game to BP. In Example~\ref{ex:secIV-1} for instance, Figure~\ref{fig:wheatstone-WE} shows that the set $\{p_3\}$ is an unnecessary set for any demand $D \geq 2$, and indeed this path subjects the game to BP at a set of demands lower than $2$.  Similarly in Example~\ref{ex:secIV-2}, we deduce from \eqref{eq:three-node-four-edge-WE} that the set $\{p_1, p_2\}$ is unnecessary for all demands, and $\{p_3\}$ is unnecessary for any demand $D \geq 2$. Again, path $p_3$ subjects the game to BP at lower levels of demand, but the set $\{p_1, p_2\}$ is merely ``non-essential'' in the sense that removing either or both of the paths in this set does not change the edge flow under WE and the associated WE-cost. Note that the converse of Proposition~\ref{prop:unnecesary-and-BP} does not hold; that is, a set of paths that subjects the game to BP at some demand does not always become an unnecessary set at some higher level of demand. We discuss this further in Example~\ref{ex:complex-counter-example}. 
	
	We note that as a method for finding BP in a given network, Proposition~\ref{prop:unnecesary-and-BP} has drawbacks. It does not differentiate between sets of paths that subject the game to BP and sets of paths that are only unnecessary. In addition, it only tells us about a potential BP at some lower level of demand. 
	It would be preferable if we could find a condition which guarantees BP at the level of demand under consideration. The next result is the essential observation that allows us to find more efficient ways of detecting BP. First, we require the following definition.
	
	\begin{definition}\longthmtitle{Affine extension functions}\label{def:affine-ext}
		Let $(\PP,\CC)$ and $\vecPP \subseteq \PP$ be given, and consider the game $(\vecPP,\CC)$. For $i \in [\vecM]_0$,  we define the following function:
		\begin{equation*}
			T \mapsto u_{\vecPP,i}(T) := \vecmWE(\vecD_i) + (T - \vecD_i) \delta \veclamb^{i},
		\end{equation*}
		where $\vecD_i \in \vecDD$. Note that $u_{\vecPP,i}$ is simply the affine function that describes $\vecmWE$ on the interval $[\vecD_i,\vecD_{i+1})$, but extended to $\realnonnegative$. We call $u_{\vecPP,i}$ an \emph{affine extension}. \oprocend 
	\end{definition}
	The next results shows that for any affine extension $u_{\tillPP,i}$ of any modified game $(\tillPP,\CC)$, there exists another modified game $(\widecheck{\PP},\CC)$ with $\widecheck{\PP} \subseteq \tillPP$ that achieves a WE-cost equal to or less than $u_{\tillPP,i}(D)$ for all $D \leq \widetilde{D}_{i+1}$.
	\begin{lemma} \longthmtitle{An upper bound on the minimum WE-cost over all modified games} \label{lem:lambda_vs_u}
		Let $(\PP,\CC)$ be given. For any $\tillPP \subseteq \PP$, $i \in [\widetilde{M}]_0$ and $D \leq {\widetilde{D}_{i + 1}}$, there exists a set $\SSrem_{i,D} \subset \PP$ such that for the routing game $(\widecheck{\PP},\CC)$, where $\widecheck{\PP} := \PP \setminus \SSrem_{i,D}$, the associated WE-cost, given by $\checklmWE_{i,D}$, satisfies
		\begin{equation*}
			\checklmWE_{i,D}(D) \leq u_{\widetilde{\PP},i}(D).
		\end{equation*}
	\end{lemma}
	\ifinclude{
		\begin{proof}
			For notational convenience, we proof the result for the case $\widetilde{\PP} = \PP$. The case $\widetilde{\PP} \subset \PP$ can be proven using the same arguments.
			Let $\Ju_i$ be the used-set for the game $(\PP,\CC)$ on the interval $(D_i,D_{i+1})$, and let $\SSS' = (\Ju_i)^c$. We consider the game $(\PP',\CC)$, where ${\PP' := \PP \setminus \SSS'}$. Clearly, $(\Ju_i)^c$ is unnecessary with respect to the game $(\PP,\CC)$ on interval $(D_i,D_{i+1})$; i.e., $(\Ju_i)^c \notin \NN_D$ for all ${D \in (D_i,D_{i+1})}$. By Lemma~\ref{lem:dominated-V}, the cost under WE for the game $(\PP',\CC)$, denoted $\lambda'^{\operatorname{WE}}(\cdot)$, then satisfies 
			\begin{equation}\label{eq:lm-lmprime}
				\lambda'^{\operatorname{WE}}(D) = \lmWE(D), \quad \text{for all }D \in [D_i,D_{i+1}].
			\end{equation} This already shows that 
			${\lambda'^{\operatorname{WE}}(D) = u_{\PP,i}(D)}$ holds on the interval $[D_i,D_{i+1})$. Therefore, if $D_i = 0$, the proof would be complete. 	
			
			Now, assume $D_i > 0$ and note $\PP' = \RR'^{\operatorname{use}}_D = \RR'^{\operatorname{use}}_T$ for all ${D,T \in (D_i,D_{i+1})}$. It follows from Corollary~\ref{cor:interval-of-active-set} that there exist $D'_j,D'_{j+1} \in \DD'$ such that ${(D_i,D_{i+1}) \subseteq (D'_j,D'_{j+1})}$ and $\JJ'^{\operatorname{use}}_j = \PP'$. Therefore, using this inclusion and~\eqref{eq:lm-lmprime}, we have $\delta \lambda'^j = \delta \lambda^i$. This fact in combination with $\lambda'^{\operatorname{WE}}(D_i) = \lmWE(D_i)$ shows that ${\lambda'^{\operatorname{WE}}(D) = u_{\PP,i}(D)}$ for all $D \in [D'_j,D'_{j+1})$. As before, if $D'_j = 0$, the proof would be complete.
			
			When $D'_j > 0$, note that $\JJ'^{\operatorname{act}}_j = \JJ'^{\operatorname{use}}_j = \PP'$. From Lemma~\ref{lem:inclusion-active-and-used} we then have $ \RR'^{\operatorname{act}}_{D'_j} = \JJ'^{\operatorname{act}}_j$, and it follows from Lemma~\ref{lem:simpe-cost-evolution} that $\delta \lambda'^{j-1} > \delta \lambda'^j$. This in combination with $\delta \lambda'^j = \delta \lambda^i$ and $\lambda'^{\operatorname{WE}}(D'_j) = u_{\PP,i}(D'_j)$ shows that $\lambda'^{\operatorname{WE}}(D) < u_{\PP,i}(D)$ for all $D \in [D'_{j-1},D'_{j})$. As before, if $D'_{j-1} = 0$, the proof would be complete. Note that the conclusions up until this point also give ${u_{\PP',j-1}(D) < u_{\PP,i}(D)}$ for all $D < D'_{j}$.
			
			If $D'_{j-1} > 0$, we can define $\SSS'' = (\JJ'^{\operatorname{use}}_{j-1})^c$ and consider the game $(\PP'',\CC)$ , where $\PP'' := \PP \setminus \SSS''$. First we show that $\PP''$ is nonempty. To see this, note that $\PP'' = \emptyset$ implies $\JJ'^{\operatorname{use}}_{j-1} = \emptyset$, which in turn implies that for all demands $D \in (D'_{j-1},D'_{j})$ and all WE $f'^D \in \WW_{D}$, $f'^D_p = 0$ for all $p 
			\in PP'$. This however means that the $f^D \in \FF_{0}$, which contradicts the assumption that $D'_{j-1} > 0$. We conclude that $\PP''$ is nonempty.
			
			We also have $\JJ'^{\operatorname{use}}_{j-1} \subseteq \PP'$, and Corollary~\ref{cor:interval-of-active-set} gives $\JJ'^{\operatorname{use}}_{j-1} \neq \JJ'^{\operatorname{use}}_{j}$. Since $\JJ'^{\operatorname{use}}_{j} = \PP'$ we see that $\JJ'^{\operatorname{use}}_{j-1} \subset \PP'$. Therefore $\PP'' \subset \PP'$.
			
			We can now apply the arguments made for comparing $\lambda'^{\operatorname{WE}}$ with $u_{\PP,i}$ to compare $\lambda''^{\operatorname{WE}}$ with $u_{\PP',j-1}$, which gives
			\begin{equation}
				\lambda''^{\operatorname{WE}}(D) =  u_{\PP',j-1} \quad \text{for all } D \in [D''_{k},D''_{k+1}),
			\end{equation}
			where $k$ is such that  ${[D'_{j-1},D'_{j}) \subseteq [D''_{k},D''_{k+1})}$ for some  $D''_{k},D''_{k+1} \in \DD''$. If $D''_k = 0$, the proof is complete, and if $D''_k > 0$, the same arguments as before give
			\begin{align*}
				\lambda''^{\operatorname{WE}}(D) &<  u_{\PP',j-1} \quad \text{for all } D \in [D''_{k-1},D''_{k}),
				\\
				u_{\PP'',k-1} &<  u_{\PP',j-1} \quad \text{for all } D \leq D''_{k},
			\end{align*}
			where $D''_{k-1} \in \DD''$ satisfies $D''_{k-1} < D''_{k} \leq D'_{j-1}$. This then shows that the result holds on the interval $[D''_{k-1},D_{i+1})$. If $D''_{k-1} = 0$, the proof is complete. If $D''_{k-1} > 0$ we can again repeat our arguments, to extend the interval on which the statement is shown to hold. Each time we repeat the arguments the set of paths under consideration is a strict subset of the previous considered set (e.g. $\PP'' \subset \PP' \subset \PP$). Since there are only finitely many paths, we can only repeat the arguments finitely many times, but we can always repeat the argument as long as the used set at the lowest value in the interval where the statement is shown to hold is non-empty. We conclude that after a finite number of repititions the used set at lowest value of the interval where the statement is shown to hold is empty, which implies that demand at this point is zero. Therefore, the proof is complete.
		\end{proof}
	}
	We outline how the above result can help detect BP. Recall from the start of this section that the routing game discussed in Example~\ref{ex:secIV} and associated to Figure~\ref{fig:wheatstone} is subject to BP when $D \in (\frac{2}{3},2)$. This instance of BP is caused by the set of paths $\SSrem = \{p_3\}$. We wish to derive that the game is subect to BP on this interval using Lemma~\ref{lem:lambda_vs_u}, without knowing apriori that we need to construct a modified game by removing the path $p_3$. For this, imagine that the WE-cost $\lmWE$ of the original game is availabe to us. Observe from the figure that the affine extension $u_{\PP,2}$ is the extension to $\realnonnegative$ of the linear map describing $\lmWE$ on the interval $(2,\infty)$ (in fact, $u_{\PP,2}$ is the same as $\tillmWE$).
	From Lemma~\ref{lem:lambda_vs_u} we then have that, for any $D \in [0,2)$, there exists some set of paths $\hatSSS$ which when removed from $\PP$ yield a routing game with $\hatlmWE(D) \leq u_{\PP,i}(D)$. However, the figure show that the WE-cost $\lmWE$ lies above this line $u_{\PP,2}$ for any $D \in (\frac{2}{3},2)$.   Thus, we have $\lmWE(D) > \hatlmWE(D)$ on the interval  $D \in (\frac{2}{3},2)$, which shows that the network is subject to BP on this interval.
	
	Note that we cannot always count on a BP to be revealed by comparing $\lmWE$ with functions of the form $u_{\PP,i}$. For instance, the BP present in Example~\ref{ex:secIV-2}, shown in Figure~\ref{fig:wheatstone-merged-bp}, is not revealed in this way, but instead by comparing $\lmWE$ with $u_{\vecPP,0}$, where $\vecPP = \PP \setminus \{(e_1,e_4)\}$. The full potential of Lemma~\ref{lem:lambda_vs_u} for detecting BP is revealed later in Proposition~\ref{prop:BP-ifandonlyif}. However, before we can establish that result, we present intermediate statements which themselves give additional useful methods for detecting BP. The first of these statements is that as a consequence of Lemma~\ref{lem:lambda_vs_u}, any increase in the slope of $\lmWE$ at a breakpoint implies the game is subject to BP at lower levels of demand.
	\begin{corollary} \longthmtitle{BP revealed by increase in slope of $\lmWE$} \label{prop:delta-lambda-increase-BP}
		Let $(\PP,\CC)$, $D_i \in \DD$ and $i \in [M]$ be given. If $\delta \lambda^{i-1} < \delta \lambda^i$, then there exists, for all ${T \in [D_{i-1},D_i)}$, a set $\SSrem_T \subset \PP$ such that $\tillmWE(T) < \lmWE(T)$, where $\tillmWE(T)$ is the WE-cost at demand $T$ for the modified game formed by removing paths in $\SSrem_T$.
	\end{corollary}
	\ifinclude{
		\begin{proof}
			The result follows from Lemma~\ref{lem:lambda_vs_u} after noting that $\lmWE(D_i) = u_{\emptyset,i}(D_i)$ and that for $D \in (D_{i-1},D_i)$ we have
			\begin{align*}
				\lmWE(D) &= \lmWE(D_i) + (D - D_{i})\delta \lambda^{i-1},
				\\
				u_{\PP,i}(D) &= u_{\emptyset,i}(D_i) + (D - D_{i})\delta \lambda^{i}.
			\end{align*}
		\end{proof}
	}
	Lemma~\ref{lem:lambda_vs_u} and the above reveal potentially useful ways of detecting BP, however, they rely on investigating the WE-cost for multiple levels of demand. The following result finally gives us a sufficient condition for a game to be subject to BP at one level of demand $D$, which requires no investigation of modified games, or of the same game at multiple levels of demand.
	
	\begin{proposition} \longthmtitle{Paths losing flow reveals BP} \label{prop:losing-flow-implies-BP}
		Let $(\PP,\CC,D)$ be given. If ${\Gamma_{D} \cap \FF_1 = \emptyset}$, then there exists a set $\SSrem$ such that ${\tillmWE(D) < \lmWE(D)}$.
	\end{proposition}
	\ifinclude{
		\begin{proof}
			Let $D$ be such that $\Gamma_{D} \cap \FF_1 = \emptyset$. In addition, let $D_i,D_{i+1} \in \DD$ be such that $D \in [D_{i},D_{i+1})$ and let ${\SSrem := (\Ju_{i})^c}$. We consider the modified game over the set of paths ${\tillPP := \PP \setminus \SSrem = \Ju_{i}}$. Note that by definition $\SSrem \notin \NN_{T}$ for all $T \in [D_{i},D_{i+1})$ and so, by Lemma~\ref{lem:dominated-V}, we therefore have
			$\lmWE(T) = \tillmWE(T)$ for all $T \in [D_{i},D_{i+1}]$. Consequently, for proving the result, it suffices to show that there exists a set $\hatPP \subset \tillPP$ such that $\hatlmWE(D) < \tillmWE(D)$.
			
			Our first aim is to show that for the game defined over the set of paths $\tillPP := \Ju_{i}$ we have  $D < \tillD_{\tillM}$; that is, $D$ does not lie in the ``final'' interval of the modified game. This we do by proving that $\tillGamma_D \cap \FF_1 = \emptyset$. Indeed this is enough as from Proposition~\ref{prop:characterize-D-infinity}, if $D \ge \tillD_{\tillM}$, then $\tillGamma_D \cap \SOL(\FF_1,A)$ is nonempty. 
			
			Note that we have $\SSrem \notin \NN_{D}$ from above. Consequently, by Lemma~\ref{lem:dominated-V}, we have $\tillf^D \in \tillWW_{D}$ if and only if $\tillf^D \in \WW_{D}$ and $\tillf^D_{\SSrem} = 0$. Now pick $\tillf^D \in \tillWW_{D}$ and $\tillf^\delta \in \tillGamma_D$, such that $\tillf^D + \epsilon \tillf^\delta \in \tillWW_{D + \epsilon}$ as long as $\epsilon$ is small enough. For any $\epsilon > 0$ that then also satisfies $D + \epsilon \in (D,D_{i+1})$ we thus obtain a WE of the modified game over the set $\Ju_{i}$ at demand $D + \epsilon$. Note that we have $(\Ju_{i})^c \notin \NN_{D}$ and $(\Ju_{i})^c \notin \NN_{D + \epsilon}$. It follows from Lemma~\ref{lem:dominated-V} that  $\tillf^D \in \WW_{D}$, and as long as $\epsilon$ is small enough we also have $\tillf^D + \epsilon \tillf^\delta \in \WW_{D + \epsilon}$, which shows that $\tillf^\delta \in \Gamma_{D}$. Thus we see that $\tillGamma_D \subseteq \Gamma_{D}$, which implies $\tillGamma_D \cap \FF_1 = \emptyset$ and so, we have $D < \tillD_{\tillM}$. 
			
			Now let $\tillD_j,\tillD_{j+1} \in \tillD$ satsify $D \in [\tillD_{j},\tillD_{j+1})$. Since $\tilaset_T = \tillJa_j = \tillPP$ for all $T \in (\tillD_{j},\tillD_{j+1})$ it follows from Lemma~\ref{lem:inclusion-active-and-used} that $\tillJa_j = \tilaset_{\tillD_{j+1}}$. Lemma~\ref{lem:simpe-cost-evolution} therefore gives us $\delta \tilllamb^j <
			\delta \tilllamb^{j+1}$. The statement then follows from Corollary~\ref{prop:delta-lambda-increase-BP}.
		\end{proof}
	}

	Proposition~\ref{prop:losing-flow-implies-BP} gives us a feasible way of detecting Braess's paradox in a network at one specific level of demand. The downside is that the given condition is only sufficient; a game can still be subject to BP even when the given condition is not satisfied. For instance, considering the routing game defined by \eqref{eq:Wheatstone-simple-path-cost} in Example \ref{ex:secIV-1} and looking at the evolution of the associated WE in Figure~\ref{fig:wheatstone-WE}, we see that ${\Gamma_{D} \cap \FF_1 = \emptyset}$ for $D \in [1,2)$, revealing that the game is subject to BP in this range of demands, while the condition is not satisfied for $D \in (\frac{2}{3},1)$. Furthermore, in Example~\ref{ex:secIV-2}, in the routing game defined by \eqref{eq:Wheatstone-double-cost}, the condition ${\Gamma_{D} \cap \FF_1 = \emptyset}$ does not hold for any level of demand, completely missing the fact that the game is subject to BP in that example. To address this shortcoming our next result gives a necessary and sufficient condition for existence of Braess's paradox. 
	\begin{proposition} \longthmtitle{Final cost evolution of modified games reveals all BPs} \label{prop:BP-ifandonlyif}
		Let $(\PP,\CC)$ be given. The routing game is subject to a Braess's paradox at demand $D$ if and only if there exists a set $\SSrem \subset \PP$ such that $\tillGamma_{D} \cap \FF_1 \neq \emptyset$ and
		\begin{equation}\label{eq:u-lm-ineq}
			u_{\tillPP,\tillM}(D) < \lmWE(D).
		\end{equation}
	\end{proposition}
	\ifinclude{
		\begin{proof}
			First assume that $(\PP,\CC)$ is subject to Braess's paradox at demand $D$. That is, there exists a set $\SSS'$ such that $\lambda'^{\operatorname{WE}}(D) < \lmWE(D)$, where $\lambda'^{\operatorname{WE}}$ stands for the WE-cost for the game $(\PP \setminus \SSS',\CC)$. To reiterate, we need to show that there exists $\SSrem$ such that $\tillGamma_{D} \cap \FF_1 \neq \emptyset$ and~\eqref{eq:u-lm-ineq} holds. Consider therefore the case when $\Gamma'_D \cap \FF_1 = \emptyset$. It follows from Proposition~\ref{prop:losing-flow-implies-BP} that the routing game $(\PP \setminus \SSS',\CC)$ is also subject to Braess's paradox at demand $D$. That is, there exists a set $\SSS'' \subset \PP'$ such that ${\lambda''^{\operatorname{WE}}(D) < \lambda'^{\operatorname{WE}}(D)}$. With this we deduce that since $\PP$ is subject to BP caused by the set $\SSS'$ and $\PP'$ is subject to BP by removing the paths $\SSS''$, then $\PP$ is subject to BP caused by the set $\SSS' \cup \SSS''$. We can repeat this argument until we find a set $\hatSSS$ such that the modified game $(\PP \setminus \hatSSS)$ satisfies $\hatGamma_D \cap \FF_1 \neq \emptyset$. Let $\hatSSS' = (\hatuset_D)^c$. We then consider the game $(\hatPP \setminus \hatSSS',\CC)$. Once again, if $\hatGamma'_D \cap \FF_1 = \emptyset$, this means that the game $(\hatPP \setminus \hatSSS',\CC)$ is subject to BP at demand $D$, and therefore there exists a set $\hatSSS'' \subset \hatPP'$ such that for the game over $(\hatPP' \setminus \hatSSS'',\CC)$ we find $\widehat{\lambda}''^{\operatorname{WE}}(D) < \widehat{\lambda}'^{\operatorname{WE}}(D)$. We can repeat the above arguments until we find a set $\SSrem \subset \PP$ that satisfies $\tillPP := \PP \setminus \SSrem$, $\tiluset_D = \tillPP$, $\tillGamma_D \cap \FF_1 \neq \emptyset$, and $\tillmWE(D) < \lmWE(D)$. 
			
			Next, let $ \tillf^\delta \in \tillGamma_D \cap \FF_1$. Since $\tiluset_D = \tillPP$ we have $\tillmWE_p(D) = \tillmWE_r(D)$ for all $p,r \in \tillPP$. In addition we have
			\begin{equation*}
				\tillMM_D = \setdef{f^\delta \in \HH_{1}}{f^\delta_{\tillPP^c} = 0},
			\end{equation*}
			and from Proposition~\ref{prop:characterize-directions-of-increase} we know that $\tillf^\delta \in \SOL(\tillMM_D,A)$.
			It follows from Propositions~\ref{prop:properties-solution-VI} that $A_p\tillf^\delta = A_r\tillf^\delta$ holds for all $p,r \in \tillPP$. In other words, as we move in the direction of $\tillf^\delta$, the cost of all paths remain equal. Since $\tillf^\delta \in \FF_1$ it also follows that for any $\tillf^D \in \tillWW_{D}$ and any $\epsilon > 0$ we have $\tillf^D + \epsilon \tillf^\delta \geq 0$. Consequently, $\tillf^D + \epsilon \tillf^\delta$ is a WE for any $\epsilon > 0$.
			From this it follows that $\tillf^\delta \in \tillGamma_T$ for all $T \geq D$. Since we know from Lemma~\ref{cor:evolution-of-GammaD} that $\tillGamma^i \neq \tillGamma^{i+1}$ for all $i \in [\tillM]$, it follows that $D \in [\tillD_{\tillM},\infty)$. 
			As a consequence, using the definition of the affine extension function $u_{\tillPP,\tillM}$ we arrive at $\tilllamb(D) = u_{\tillPP,\tillM}(D)$. 
			And since $\tilllamb(D) < \lmWE(D)$, we conclude that $u_{\tillPP,\tillM}(D) < \lmWE(D)$. This shows one direction of the implication.
			
			For the other direction, assume that there exists some $\SSrem \subset \PP$ such that $u_{\tillPP,\tillM}(D) < \lmWE(D)$. If $D < \tillD_{\tillM}$, then it follows from Lemma~\ref{lem:lambda_vs_u} that the game $(\PP,\CC)$ is subject to BP at demand $D$. If $D \geq \tillD_{\tillM}$, then we have $\tillmWE(D) = u_{\tillPP,\tillM}(D)$ and therefore $\tillmWE(D) < \lmWE(D)$, which shows that the game is subject to BP at demand $D$. This completes the proof.
		\end{proof}
	}
	The merit of the above proposition lies in that for different levels of demand we do not have to check for BP related to a set $\SSrem$ separately. Instead, given $\SSrem$, we have one function, namely $u_{\tillPP,\tillM}$ which can be used to check for BP related to $\SSrem$ for all demands. This can be done simply by checking whether the value of that function exceeds the achieved WE-cost or not. However, the condition depends on the choice for $\SSrem$ and as such, it can be computationally infeasible to check this condition for all possible $\SSrem$ in order to detect BP. Despite this, it is still a useful result. For instance, we observed earlier that for Example~\ref{ex:secIV-1}, comparing $\lmWE$ to $u_{\PP,2}$ reveals BP on the interval $D \in (\frac{2}{3},1)$ and we can now generalize this observation with the following result:
	\begin{corollary} \longthmtitle{Easily obtainable upper bound on achievable WE-cost}
		Let $(\PP,\CC,D)$ be given. If ${u_{\PP,M}(D) < \lmWE(D)}$, then the network is subject to Braess's paradox at demand $D$. 
	\end{corollary}
	Note that thanks to Lemma \ref{lem:obtaining-lambda-D_M} we have a direct method of obtaining $u_{\PP,M}$, which highlights the usefulness of the above corollary.
	
	Before finishing our exposition on BP in this section, we use the obtained results to show that any game can only be subject to BP on a finite interval of demands. To the best of our knowledge, this result has not been established before in full generality, though versions limited to the Wheatstone network in Figure~\ref{fig:wheatstone} have been obtained in~\cite{EIP-SLP:97,VZ-EA:15}.
	\begin{proposition} \longthmtitle{Braess's paradox occurs on a finite interval}
		Let $(\PP,\CC)$ be given. There exists a value $D^{\operatorname{BP}} \geq 0$ such that for any $\SSrem$, we have $\lmWE(D) \leq \tillmWE(D)$ for all $D \geq D^{\operatorname{BP}}$.
	\end{proposition}
	\ifinclude{
		\begin{proof}
			For all $T \geq \max(D_M,\tillD_{\tillM})$ we have
			\begin{align*}
				\lmWE(T) &= \lmWE(D_M) + (T-D_M) \delta \lambda^M,
				\\
				\tillmWE(T) &= \tillmWE(\tillD_{\tillM}) + (T-\tillD_{\tillM}) \delta \tilllamb^{\tillM}.
			\end{align*}
			For the sake of contradiction, assume that $\delta \tilllamb^{\tillM} < \delta \lambda^M$. From Proposition~\ref{prop:pw-affine}, we have
			\begin{equation} \label{eq:par-v-delta}
				\begin{aligned}
					\frac{\partial}{\partial D}\big(\tillV(T)  -  V(T)\big)  &=  \tillmWE(T)  -  \lmWE(T) \\
					&=  \lmWE(D_M)  -  \tillmWE(\tillD_{\tillM})	 
					\\ 
					&+  (T  -  D_M) \delta \lambda^M   -  (T -  \tillD_{\tillM}) \delta \tilllamb^{\tillM} 
				\end{aligned}
			\end{equation}
			for $T \geq \max(D_M,\tillD_{\tillM})$. Since we assume $\delta \tilllamb^{\tillM} < \delta \lambda^M$, the above relation implies that for large enough $T$, we get $\tillV(T) < V(T)$ which contradicts Lemma~\ref{lem:dominated-V}. Therefore, we obtain $\delta \tilllamb^{\tillM} \geq \delta \lambda^M$.
			Now consider two cases: (a) ${\delta \tilllamb^{\tillM} = \delta \lambda^M}$ and (b) $\delta \tilllamb^{\tillM} > \delta \lambda^M$. For (a) note that if $\tillmWE(T) < \lmWE(T)$ for any ${T \geq \max(D_M,\tillD_{\tillM})}$, then we arrive at a similar contradiction with Lemma~\ref{lem:dominated-V} as before. Thus, for case (a), we must have ${\tillmWE(T) \geq \lmWE(T)}$ for all $T \geq \max(D_M,\tillD_{\tillM})$. For case (b), from~\eqref{eq:par-v-delta}, for all large values of $T$, we have $\tillmWE(T) \geq \lmWE(T)$. Hence, combining the reasoning of both cases, we find that there exists some value $D^{\operatorname{BP}}$ such that $\tillmWE(T) \geq \lmWE(T)$ for all $T \geq D^{\operatorname{BP}}$. This completes the proof.
		\end{proof} 
	}
	Note that $D^{\operatorname{BP}}$ can be strictly larger then $D_M$, as is the case in Example~\ref{ex:secIV-2}, where $D_M = 1$ while Figure~\ref{fig:wheatstone-merged-bp} shows that the game is subject to BP on the interval $D \in (\frac{2}{3},2)$.
	\subsection{Implications for Braess's paradox}
	In the final part of this section, we discuss how the results we have obtained on BP show that even when a game is known to be subject to BP, one should be careful in drawing the conclusion that the responsible set of paths is better removed. We start from Proposition~\ref{prop:unnecesary-and-BP} from which it may seem that the usefulness of a set $\SSrem$ that is not necessary at some demand $D$ is questionable. If this set of paths had not been present, either the WE-cost would have stayed the same for all lower levels of demand, or better yet, would have decreased for some of these demands. 
	
	However, Proposition~\ref{prop:unnecesary-and-BP} does not conclude anything about behavior at higher demands. A path that is unnecessary for one level of demand may be very important when the demand is higher. It can even be the case that a path is necessary at some level of demand, becomes unnecessary at a higher level of demand, and finally becomes a necessary part of the ``final'' set of used paths $\Ju_{M}$. This phenomenon is showcased in the following example:
	\begin{example} \longthmtitle{Paths causing BP can become useful again} \label{ex:complex-counter-example}
		\begin{figure}
			\centering
			\begin{minipage}[t]{0.5 \textwidth}
				\centering
				\begin{tikzpicture}[,->,shorten >=1pt,auto,node distance=1.5 cm,
					semithick]
					
					\tikzset{VertexStyle/.style = {shape          = circle,
							ball color     = white,
							text           = black,
							inner sep      = 2pt,
							outer sep      = 0pt,
							minimum size   = 18 pt}}
					\tikzset{EdgeStyle/.style   = {
							double          = black,
					}}
					\tikzset{LabelStyle/.style =   {draw,
							fill           = white,
							text           = black}}
					\node[VertexStyle](A){$v_o$};
					\node[VertexStyle, below right=of A](B){};
					\node[VertexStyle, right=of B](C){};
					\node[VertexStyle, above right=of C](D){$v_d$};
					\node[VertexStyle, above right= of A, xshift=10.7mm](E){};
					
					\path (A)  edge      node{$e_1 $}(E)
					(E)  edge      node{$e_2 $} (D)
					(A)  edge             node[below]{$e_3 \enskip$} (B)
					(E)	edge[dashed]			node{$e_5$}	(C)
					(B)	edge			node{$e_7$}	(C)
					(C)	edge			node[below]{$\enskip e_4$}	(D)
					(B)	edge[bend right = 80]	node[below]{$\enskip e_6$}	(D);
				\end{tikzpicture}
				\caption{Modification of the Wheatstone network, in which there is one additional path.}
				\label{fig:complex-counter}
			\end{minipage} \hfill
			\begin{minipage}[t]{0.48 \textwidth}
				\centering
				\includegraphics[width =\textwidth]{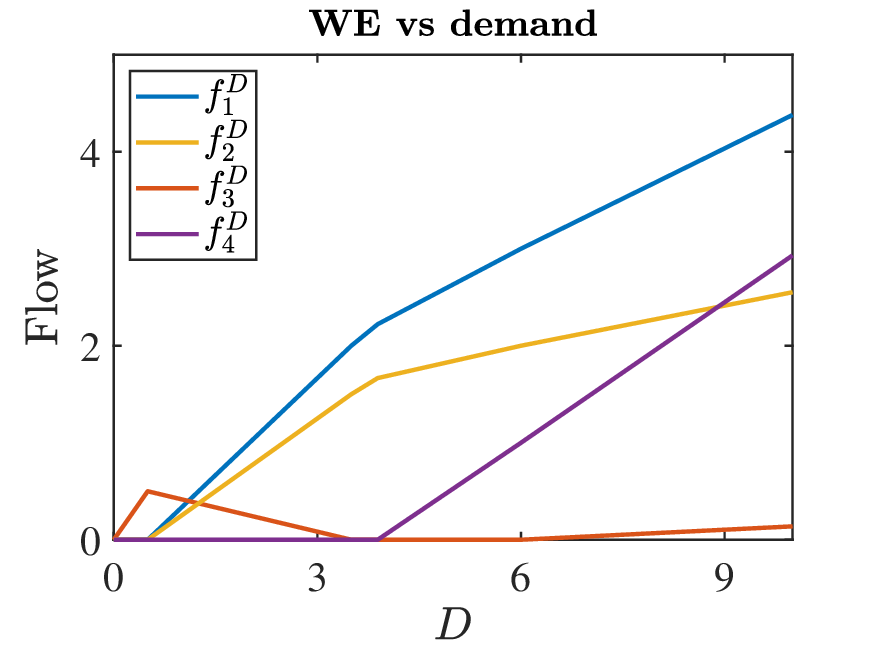}
				\caption{The evolution of $f^D$ for the routing game over the network in Figure~\ref{fig:complex-counter} defined by the costs \eqref{eq:complex-counter-cost}.}
				\label{fig:complex-counter-WE}
			\end{minipage}
		\end{figure}
		\rm{
			Consider the network in Figure~\ref{fig:complex-counter}, and let the cost functions of the edges be given by
			\begin{equation*}
				\begin{aligned}
					&C_{e_1}(f_{e_1}) = 2 f_{e_1},	&		&C_{e_2}(f_{e_2}) = f_{e_2} + 1,	
					\\
					&C_{e_3}(f_{e_3}) = f_{e_3} + 1,	&	&C_{e_4}(f_{e_4}) 	= 2f_{e_4}, \hspace{8 pt}	
					\\
					&C_{e_5}(f_{e_5}) = 0,	&		&C_{e_6}(f_{e_6}) = f_{e_6} + 5,	
					\\
					&	&C_{e_7}(f_{e_7}) = f_{e_7}.	&
				\end{aligned}
			\end{equation*}
			This network has four paths from origin to destination, namely $p_1 = (e_1,e_2)$, ${p_2 = (e_3,e_7,e_4)}$, $p_3 = (e_1,e_5,e_4)$ and ${p_4 = (e_3,e_6)}$. The path-cost function is then given by
			\begin{equation} \label{eq:complex-counter-cost}
				C(f) = Af + b = \left(\begin{array}{ccccc}
					3	&0	&2	&0	\\
					0	&4	&2	&1	\\
					2	&2	&4	&0	\\
					0	&1	&0	&2	\\
				\end{array}\right)f + \left(\begin{array}{c}
					1	\\
					1	\\
					0	\\
					6	\\
				\end{array}\right).
			\end{equation}
			We see the evolution of the WE, which is unique in this case, in Figure~\ref{fig:complex-counter-WE}. The explicit expression for the WE is 
			\begin{equation*}
				f^D  = 
				\begin{cases}
					\left(\begin{array}{cccc}
						0	&0	&D	&0   
					\end{array}\right)^\top & \text{for } D  \in  [0,\frac{1}{2}],	\\
					\left(\begin{array}{cccc}
						\frac{8D - 4}{12}	&\frac{6D - 3}{12}	&\frac{7 - 2D}{12}	&0    
					\end{array}\right)^\top 
					& \text{for } D   \in  [\frac{1}{2},\frac{7}{2}],	\\
					\left(\begin{array}{cccc}
						\frac{4D}{7}	&\frac{3D}{7}	&0	&0    
					\end{array}\right)^\top & \text{for } D   \in   [\frac{7}{2},\frac{35}{9}],	\\
					\left(\begin{array}{cccc}
						\frac{7D + 15}{19}	&\frac{3D + 20}{19}	&0	&\frac{9D - 35}{19}   
					\end{array}\right)^\top & \text{for } D   \in  [\frac{35}{9},6],
					\\
					\left(\begin{array}{cccc}
						\frac{10D + 27}{29}	&\frac{4D + 34}{29}	&\frac{D - 6}{29}&\frac{14D - 55}{29}	   
					\end{array}\right)^\top & \text{for } D   \in  [6,\infty).
				\end{cases}
			\end{equation*}
			Note that for $D \in [1,\frac{1}{2}]$ path $p_3$ carries all flow. Then, on the interval $[\frac{1}{2},\frac{7}{2}]$ path $p_3$ loses all flow as it is rerouted onto the paths $p_1$ and $p_2$, and consequently, $p_3$ is unnecessary in the interval $[\frac{7}{2},\frac{35}{9}]$. Finally $p_3$ becomes necessary again for $D > 6$ and we see that $\Ju_{M} = \PP$ in this case. \oprocend
		}
	\end{example}
	The above shows that one needs to be cautious in dismissing a set of paths that are unnecessary at a certain demand to be useless or detrimental overall. 
	Though this set subjects the game to BP at some level of demand, it may become useful again at higher levels of demand. However, even for lower levels of demand the conclusions are not that straightforward. The following result shows that when a set $\SSrem$ subjects the game to BP at some demand $D$, its presence must have been strictly beneficial at some lower level of demand.
	\begin{corollary} \longthmtitle{Paths causing BP are useful at lower demands} \label{cor:BP-beneficial}
		Let $(\PP,\CC,D)$ and $\SSrem \subset \PP$ be given. If $\lmWE(D) > \tillmWE(D)$, then there exist $D^-,D^+$ such that ${0 < D^- < D^+ < D}$ and 
		\begin{equation*}
			\lmWE(T) < \tillmWE(T) \quad \text{for all } T \in (D^-,D^+).
		\end{equation*}
	\end{corollary}
	\ifinclude{
		\begin{proof}
			The arguments for proving this are similar to those for Proposition~\ref{prop:unnecesary-and-BP}. We note that $\lmWE$ and $\tillmWE$ are the derivatives of $V$ and $\tillV$, respectively, and are in addition continuous and piece-wise affine with only finitely many points at which they are not differentiable. From Lemma~\ref{lem:dominated-V} we know that $V(T) \leq \tillV(T)$ for all $T \in \realnonnegative$, and we can derive from the same lemma that if $V(T) = \tillV(T)$, then ${\lmWE(D) = \tillmWE(D)}$. It follows that $\lmWE(D) > \tillmWE(D)$ implies $V(D) < \tillV(D)$. Since $V(0) = \tillV(0) = 0$ this implies that for some range of demands $(D^-,D^+)$ with $0 \leq D^- < D^+ < D$ we must have $\lmWE(T) < \tillmWE(T)$ for all $T \in (D^-,D^+)$.
		\end{proof}
	}
	The above result lends us a different perspective on Braess' paradox. We have already seen that the phenomenon is highly dependent on the demand, but this corollary shows that even though addition of a set of paths may increase the travel-time of all participants at one level of demand, looking at a more complete picture, we see that the same set of paths must have decreased travel time for some lower level of demand. When making the decision to keep or remove a set of paths from a network it would thus be helpful to consider the effect of those paths on the network for the entire range of demands in which the network is planned to function. 
	
	This leads naturally to the question of how to quantify the value of a path to the network while considering a range of demands. We first consider the following function
	\begin{equation*}
		J(D) = \int_{0}^{D} \tillmWE(z) - \lmWE(z)dz
	\end{equation*}
	as a somewhat simplistic measure of the value of a set of paths $\SSrem$ to the network on the range of demands from zero to $D$.
	For this we have the following result:
	\begin{proposition} \longthmtitle{Benefits of a path using a simple measure}\label{pr:J}
		Let $(\PP,\CC,D)$ and $\SSrem \subset \PP$ be given. We have
		\begin{equation*}
			J(D) \geq 0 \text{ for all } D \in [0,\infty),
		\end{equation*}
		with $J(D) = 0$ if and only if $\SSrem \notin \NN_D$.
	\end{proposition}
	\ifinclude{
		\begin{proof}
			Since $\lmWE(D) = V'(D)$ and ${\tillmWE(D) = \tillV'(D)}$ we get
			\begin{align*}
				\int_{0}^{D} \tillmWE(z) - \lmWE(z)dz	&= \tillV(z)|_0^D - V(z)|_0^D	\\
				&= \tillV(D) - V(D)	\\
				&\geq 0,
			\end{align*}
			with equality holding if and only if $\tillV(D) = V(D)$. Using Lemma~\ref{lem:dominated-V} the proof is finished.
		\end{proof}
	}
	Using this measure, we note that in the worst-case scenario, the set $\SSrem$ is ``neutral'' to the network, which only occurs when $\SSrem$ is unnecessary at demand $D$. If this is not the case, then the presence of $\SSrem$ is strictly beneficial to some degree.
	
	As mentioned, this may be a fairly naive measure. Even when not knowing anything about the levels of demand that a network is likely to carry, we may weigh certain levels more than others. A reasonable approach may be to weigh each level of demand by that amount of demand, as it is more important for the system to perform well when demand is high than when demand is low. In this case we have the measure
	\begin{equation*}
		W(D) = \int_{0}^{D} z\big(\tillmWE(z) - \lmWE(z)\big)dz.
	\end{equation*}
	For which we have the following result:
	\begin{proposition} \longthmtitle{Detriments of a path using a measure weighed by demand}\label{pr:W}
		Let $(\PP,\CC,D)$ and $\SSrem \subset \PP$ be given, where $D > 0$. If $\SSrem \notin \NN_D$ then
		\begin{equation*}
			W(D) \leq 0
		\end{equation*}
		with equality holding if and only if $\SSrem \notin \NN_T$ for all $T \in (0,D)$.
	\end{proposition}
	\ifinclude{
		\begin{proof}
			The result can be obtained by integration by parts:
			\begin{align*}
				W(D)	&= \int_{0}^{D} z\big(\tillmWE(z) - \lmWE(z)\big)dz	\\
				&= z\big(\tillV(z) - V(z)\big)|_0^D - \int_{0}^{D}\tillV(z) - V(z)dz	\\
				&= \int_{0}^{D}V(z) - \tillV(z)dz.
			\end{align*}
			Here the last equality follows since $\SSrem \notin \NN_D$ and therefore we have, by Lemma~\ref{lem:dominated-V}, that $V(D) = \tillV(D)$. From the same lemma we know that $V(T) \leq \tillV(T)$ for all ${T \in [0,D]}$ and it follows that we end up with an integral over a nonpositive function. This is only equal to the zero function if ${\tillV(T) = V(T)}$ for all $T \in (0,D)$, which happens, by Lemma~\ref{lem:dominated-V}, if and only if $\SSrem$ is unnecessary for all $T \in (0,D)$. This completes the proof.
		\end{proof}
	}
	Thus, if we use this measure, and we look at a range of demands from zero to a demand at which the set $\SSrem$ is unnecessary, then it follows that  $\SSrem$ is at best ``neutral'' to the network performance, and this best case scenario only occurs when $\SSrem$ is not used, that is, the set is unnecessary for the whole range of considered demands.

	\section{Conclusion}
	We have studied non-atomic, single origin-destination routing games with affine, non-decreasing cost functions on the edges, and how changes in the demand of such a game affect the set of Wardrop equilibria. We have characterized the set of directions in which the set of WE varies as the set of solutions of a variational inequality problem. Subsequently, we have used this characterization and related results to obtain various easy-to-check sufficient conditions for the presence of Braess's paradox in a network, as well as a necessary and sufficient condition that can be usefully employed, but is computationally intractable to check in full. We have also shown that any set of paths responsible for BP at some level of demand must at other levels of demand strictly reduce the WE-cost. Using two measures on the value of a set of paths we have shown that even when a set of paths is observed to cause BP at some demand, removal of that set from the network could still be detrimental to the performance of the network overall.
	
	For future work, we aim to investigate to what extend these results can be generalized to routing games with more general cost functions, and to networks with multiple origin-destination pairs. There are also several open questions whose answers could help with the efficient implementations of methods for detecting BP on the basis of our presented results. Among those is the question if and to what extent there exist efficient ways of selecting subsets $\SSrem$ that when used in Proposition~\ref{prop:BP-ifandonlyif} supply the tighest possible upper bounds on the WE-cost of the original routing game. Finding such tightest upper bounds would equate to fully revealing the presence of BP, so this may be challenging but worthwhile line of inquiry.
	
	
	\begin{appendices}
		\section{Proof of Corollary~\ref{cor:interval-of-active-set}} \label{ap:proof-of-corollary-used-sets}
		\ifinclude{
			\begin{proof}
				The claim about active sets $\aset_{D}$ is shown in~\cite[Section 4]{RC-VD-MS:21}. That is, there exist points $\DD := ( D_0,D_1,\cdots,D_M,D_{M+1})$ with $D_0  =  0$, $D_{M+1} =  \infty$ and $D_{j} > D_{j - 1}$ for all ${j \in [M+1]}$ along with sets $\{\Ja_0,\Ja_1,\cdots,\Ja_M\}$ such that $\aset_D = \Ja_i$ for all $i \in [M]$ and $D \in (D_i,D_{i+1})$. Using this result and for the defined points in $\DD$ we will show the existence of sets $\{\Ju_0,\Ju_1,\cdots,\Ju_M\}$ such that $\uset_D = \Ju_i$ for all $D \in (D_i,D_{i+1})$ and all $i$. Pick some $D_i,D_{i+1} \in \DD$. Let $D \in (D_i,D_{i+1})$ and $p \in \uset_D$. Consider any other demand $D' \in (D_i,D_{i+1})$ with $D' \not = D$. Associated to $D$ and $D'$, select $T$ such that $T \in (D_i,D_{i+1})$ and $D'$ can be written as a convex combination of $T$ and $D$. That is, $D' = \mu D + (1-\mu) T$ for some $\mu \in (0,1)$. Since $p \in \uset_D$, there exists $f^D \in \WW_{D}$ such that $f^D_p > 0$. Furthermore, $f^T_p \ge 0$ for any $f^T \in \WW_{T}$. Since $\aset_D = \aset_T$, from Lemma~\ref{lem:facts-on-evolution-for-affine}-2, there exists $f^{D'} \in \WW_{D'}$ such that $f^{D'} = \mu f^D + (1-\mu) f^T$ and so $f^{D'}_p > 0$. Thus, $p \in \uset_{D'}$. Since $D$ and $D'$ were selected arbitrarily, we conclude that the used set remains the same in the interval $(D_i,D_{i+1})$.
				
				Now, for $D_i,D_{i+1},D_j,D_{j+1} \in \DD$ with $i \neq j$ let $D^- \in (D_i,D_{i+1})$ and ${D^+ \in (D_j,D_{j+1})}$. Note that therefore $\aset_{D^-} \neq \aset_{D^+}$. For the sake of contradiction assume that $\uset_{D^-} = \uset_{D^+}$. For any $p \in \uset_{D^-}$ we know that ${C_p(f^{D^-}) \leq C_r(f^{D^-})}$ for all $r \in \PP$ and $f^{D^-} \in \WW_{D^-}$. Similarly $C_p(f^{D^+}) \leq C_r(f^{D^+})$ for all $r \in \PP$ and $f^{D^+} \in \WW_{D^+}$. Since $C(\cdot)$ is an affine function (see \eqref{eq:path-cost}) it follows that for $f^{T} = \coco_{\mu}(f^{D^-},f^{D^+})$ with $f^{D^-} \in \WW_{D^-}$, $f^{D^+} \in \WW_{D^+}$, $\mu \in [0,1]$ and $T = \coco_{\mu}(D^-,D^+)$ we have $C_p(f^{T}) \leq C_r(f^{T})$ for all $p \in \uset_{D^-}$. Also note that $f^T_p > 0$ implies $f^{D^-}_p > 0$ or $f^{D^+}_p > 0$, which in turn gives us $p \in \uset_{D^-}$.
				In other words $f^T$ is a WE. Since $\aset_{D^-} \neq \aset_{D^+}$ we can without loss of generality pick $r \in \PP$ such that $r \in \aset_{D^-}$ and $r \notin \aset_{D^+}$. In other words, the cost of path $p$ is minimal at the demand $D^-$, but not minimal at demand $D^+$. Using \eqref{eq:path-cost} and the fact that $f^T$ is a WE it follows that the cost of $r$ is not minimal under WE at demand $T = \coco_{\mu}(D^-,D^+)$. Since $D^- \in (D_i,D_{i+1})$ we can set $\mu$ such that $T \in (D_{i},D_{i+1})$. This would then imply  that $\aset_{T} \neq \aset_{D^-}$, which contradicts Corollary~\ref{cor:interval-of-active-set}. Therefore our premise is false, which shows that $\uset_{D^-} \neq \uset_{D^+}$, finishing the proof.
			\end{proof}
		}
	\end{appendices}
	
	\bibliographystyle{ieeetr}
	\bibliography{bibliography.bib}
	
\end{document}